\documentclass[groupedaddress,onecolumn,reprint,aip,jmp]{revtex4-1}
\usepackage{amsmath,amssymb,amsfonts,amsthm}
\usepackage{mathrsfs}
\usepackage[paperwidth=210mm,paperheight=297mm,centering,hmargin=2.5cm,vmargin=2.5cm]{geometry}

\newtheorem{theorem}{Theorem}[section]
\newtheorem{lemma}[theorem]{Lemma}
\newtheorem{proposition}[theorem]{Proposition}
\newtheorem{corollary}[theorem]{Corollary}

\theoremstyle{definition}
\newtheorem{definition}[theorem]{Definition}
\newtheorem{example}[theorem]{Example}

\theoremstyle{remark}
\newtheorem{remark}[theorem]{Remark}

\newcommand{\R}{{\mathbb R}}
\newcommand{\Z}{{\mathbb Z}}
\newcommand{\C}{{\mathbb C}}
\newcommand{\s}{S}
\newcommand{\h}{\mathscr H}

\begin{document}

\title{Reduction by symmetries in singular quantum-mechanical problems: general scheme and application to Aharonov-Bohm model}

\author{A.G.~Smirnov}
\email{smirnov@lpi.ru}
\affiliation{I.~E.~Tamm Theory Department, P.~N.~Lebedev
Physical Institute, Leninsky prospect 53, Moscow 119991, Russia}

\begin{abstract}
We develop a general technique for finding self-adjoint extensions of a symmetric operator that respect a given set of its symmetries. Problems of
this type naturally arise when considering two- and three-dimensional Schr\"odinger operators with singular potentials. The approach is based on
constructing a unitary transformation diagonalizing the symmetries and reducing the initial operator to the direct integral of a suitable family of
partial operators. We prove that symmetry preserving self-adjoint extensions of the initial operator are in a one-to-one correspondence with
measurable families of self-adjoint extensions of partial operators obtained by reduction. The general scheme is applied to the
three-dimensional Aharonov-Bohm Hamiltonian describing the electron in the magnetic field of an infinitely thin solenoid. We construct all self-adjoint extensions of this Hamiltonian, invariant under translations along the solenoid and rotations around it, and explicitly find their eigenfunction expansions.
\end{abstract}

\maketitle

\onecolumngrid

\section{Introduction}

The spectral analysis of Schr\"odinger operators can be often facilitated if some symmetries of the problem are known. This happens if the spectral problem for the operators representing symmetries can be explicitly solved. In this case, the initial operator can be represented as a direct sum (or, more generally, direct integral) of partial operators corresponding to fixed eigenvalues of symmetries. In this way, two- and three-dimensional problems can be often reduced to solving certain one-dimensional Schr\"odinger equations (this is the case, for example, for spherically symmetric problems).

The present paper is concerned with such a reduction by symmetries in the case of singular quantum-mechanical problems. It is well known~\cite{Albeverio1988,GTV2012} that strong singularities in the potential may lead to the lack of self-adjointness of the corresponding Schr\"odinger operator on its natural domain. As a result, the quantum model is no longer fixed uniquely by the potential and different quantum dynamics described by various self-adjoint extensions of the initial Schr\"odinger operator are possible. In this paper, we are interested in the extensions that respect some given set of symmetries of the initial operator. We propose a general technique for the reduction of such symmetry preserving extensions and apply it to the spectral analysis of the three-dimensional Aharonov-Bohm Hamiltonian describing the electron in the magnetic field of an infinitely thin solenoid.

In a general form, the problem can be posed as follows. Suppose $H$ is a closed operator in a separable Hilbert space $\mathfrak H$ and $\mathfrak X$ is a subset of the algebra $L(\mathfrak H)$ of linear bounded everywhere defined operators in $\mathfrak H$. We assume that $\mathfrak X$ is involutive (i.e., the adjoint $T^*$ of every $T\in \mathfrak X$ belongs to $\mathfrak X$) and consists of pairwise commuting operators. Further, we assume $\mathfrak X$ is a set of symmetries of $H$, by which we mean that $H$ commutes with elements of $\mathfrak X$ in the sense of the next definition.

\begin{definition}\label{d0000}
A linear operator $H$ in $\mathfrak H$ with the domain $D_H$ is said to commute with $T\in L(\mathfrak H)$ if $T\Psi\in D_H$ and $HT\Psi=TH\Psi$ for any $\Psi \in D_H$.
\end{definition}

Our aim is to show how spectral decompositions of symmetries can be used to represent closed (and, in particular, self-adjoint) symmetry preserving extensions of $H$ as direct integrals of suitable partial operators. To clarify main ideas, we assume for a while that $\mathfrak H$ falls into an orthogonal direct sum of eigenspaces of symmetries. That is, we suppose that there is a countable family $\{\mathfrak S(s)\}_{s\in \mathscr S}$ of nontrivial closed pairwise orthogonal subspaces of $\mathfrak H$ such that
\begin{equation}\label{decompo}
\mathfrak H=\bigoplus_{s\in \mathscr S} \mathfrak S(s)
\end{equation}
and
\[
T\Psi = g_T(s)\Psi, \quad \Psi\in \mathfrak S(s),
\]
for every $s\in \mathscr S$ and $T\in \mathfrak X$, where $g_T(s)$ are some complex numbers (since $\mathfrak S(s)$ are nontrivial, the eigenvalues $g_T(s)$ are defined uniquely). Let $P_s$ denote the orthogonal projection of $\mathfrak H$ onto $\mathfrak S(s)$. We would like to have a decomposition of the form
\begin{equation}\label{decomposition}
H = \bigoplus_{s\in \mathscr S} \mathcal H(s),
\end{equation}
where $\mathcal H(s)$ are closed operators in $\mathfrak S(s)$ for all $s\in\mathscr S$ and $\bigoplus_{s\in \mathscr S} \mathcal H(s)$ is, by definition, the operator in $\mathfrak H$ whose graph consists of all $(\Psi,\tilde\Psi)\in \mathfrak H\oplus\mathfrak H$ such that $(P_s\Psi,P_s\tilde\Psi)$ is in the graph of $\mathcal H(s)$ for every $s\in \mathscr S$. As the operator $\bigoplus_{s\in \mathscr S} \mathcal H(s)$ obviously commutes with all projections $P_s$, equality~(\ref{decomposition}) can hold only if $H$ commutes with $P_s$ for every $s\in \mathscr S$. In fact, the latter condition is also sufficient for the existence of decomposition~(\ref{decomposition}) (see Theorems~2 and~3 in Sec.~45 of~Ref.~\onlinecite{AkhiezerGlazman}). Given a decomposition of form~(\ref{decomposition}), it can be easily shown that $\mathcal H(s)$ is actually the restriction of $H$ to the domain $D_H\cap \mathfrak S(s)$ for every $s\in\mathscr S$. Thus, the partial operators $\mathcal H(s)$ are uniquely determined by $H$.

In general, the condition that $\mathfrak S(s)$ are eigenspaces of symmetries is insufficient to guarantee that $H$ commutes with $P_s$ for all $s\in\mathscr S$ (otherwise we could choose $\mathfrak X$ to be the one-element set containing the identity operator in $\mathfrak H$ and consider the decomposition $\mathfrak H=\mathfrak H'\oplus \mathfrak H^{\prime\bot}$, where $\mathfrak H'$ is an arbitrary closed subspace of $\mathfrak H$; it would follow that every closed operator in $\mathfrak H$ commutes with the orthogonal projection onto $\mathfrak H'$). The commutation between $H$ and $P_s$ can be ensured, however, if we require that
\begin{equation}\label{exactdecomp}
P_s\in \mathcal A(\mathfrak X)
\end{equation}
for every $s\in\mathscr S$, where $\mathcal A(\mathfrak X)$ is the smallest strongly closed subalgebra of $L(\mathfrak H)$ containing the set $\mathfrak X$ and the identity operator in $\mathfrak H$ (in other words, $\mathcal A(\mathfrak X)$ is the strong closure of the algebra generated by $\mathfrak X$ and the identity operator). Indeed, let $\mathcal M$ be the subset of $L(\mathfrak H)$ consisting of all operators commuting with $H$. Since the sum and product of any two operators in $L(\mathfrak H)$ commuting with $H$ also commute with $H$, we conclude that $\mathcal M$ is an algebra. Moreover, it easily follows from the closedness of $H$ that $\mathcal M$ is closed in the strong operator topology (see Lemma~1 in~Ref.~\onlinecite{JMAA2013}) and, hence, $\mathcal A(\mathfrak X)\subset \mathcal M$. Condition~(\ref{exactdecomp}) therefore implies that $P_s\in \mathcal M$ or, in other words, that $H$ commutes with $P_s$ for all $s\in\mathscr S$.

We shall say that an eigenspace decomposition of form~(\ref{decompo}) is exact for $\mathfrak X$ if condition~(\ref{exactdecomp}) is fulfilled for all $s\in\mathscr S$. Given an exact decomposition for $\mathfrak X$, it is easy to describe symmetry preserving closed extensions of $H$. Indeed, suppose (\ref{decomposition}) holds and a closed extension $\tilde {\mathcal H}(s)$ of $\mathcal H(s)$ is given for every $s\in \mathscr S$. Then the operator $\tilde H = \bigoplus_{s\in \mathscr S}\tilde {\mathcal H}(s)$ is clearly a closed extension of $H$ commuting with all elements of $\mathfrak X$. Conversely, if $\tilde H$ is a closed extension of $H$ commuting with all elements of $\mathfrak X$, then the above considerations applied to $\tilde H$ instead of $H$ show that $\tilde H$ falls into the direct sum of its partial operators $\tilde {\mathcal H}(s)$, which are obviously closed extensions of $\mathcal H(s)$. Moreover, $\tilde H$ is densely defined if and only if all $\tilde {\mathcal H}(s)$ are densely defined, in which case we have $\tilde H^*= \bigoplus_{s\in \mathscr S}\tilde {\mathcal H}(s)^*$. Hence, $\tilde H$ is self-adjoint if and only if $\tilde {\mathcal H}(s)$ is self-adjoint for every $s\in \mathscr S$.

Thus, given an exact decomposition for $\mathfrak X$, the operator $H$ falls into the direct sum of its partial operators $\mathcal H(s)$, and the symmetry preserving closed (resp., self-adjoint) extensions of $H$ are precisely the direct sums of closed (resp., self-adjoint) extensions of $\mathcal H(s)$.

Exactness condition~(\ref{exactdecomp}) can be given another, equivalent, formulation that is better suited for verifying in concrete applications. More specifically, we claim that (\ref{exactdecomp}) holds for all $s\in\mathscr S$ if and only if the following condition is fulfilled:
\begin{itemize}
\item[(E)] For every $s,s'\in \mathscr S$ such that $s\neq s'$, there is $T\in \mathfrak X$ such that $g_T(s)\neq g_T(s')$.
\end{itemize}
Indeed, suppose (E) holds. Since $\mathfrak X$ is involutive, the elements of $\mathcal A(\mathfrak X)$ can be characterized using von~Neumann's bicommutant theorem (see Sec.~\ref{s3} for details): an operator in $L(\mathfrak H)$ belongs to $\mathcal A(\mathfrak X)$ if and only if it commutes with every element of commutant of $\mathfrak X$, i.e., with every operator in $L(\mathfrak H)$ that commutes with all elements of $\mathfrak X$. To prove~(\ref{exactdecomp}), we have to verify this condition for $P_s$. Let $R\in L(\mathfrak H)$ commute with all elements of $\mathfrak X$. Since $TP_s = P_s T = g_T(s)P_s$ for every $s\in\mathscr S$ and $T\in \mathfrak X$, we have $P_{s'}RTP_s = g_T(s) P_{s'}RP_s = g_T(s') P_{s'}RP_s$ for all $T\in \mathfrak X$ and $s,s'\in\mathscr S$. In view of~(E), this implies that $P_{s'}RP_s = 0$ for all $s,s'\in\mathscr S$ such that $s\neq s'$. As $\Psi = \sum_{s\in\mathscr S} P_s\Psi$ for any $\Psi\in \mathfrak H$, we obtain
\[
P_{s}R\Psi = \sum_{s'\in\mathscr S} P_sRP_{s'}\Psi = P_s RP_s = \sum_{s'\in\mathscr S} P_{s'}RP_s\Psi = RP_s\Psi
\]
for every $s\in\mathscr S$ and $\Psi\in \mathfrak H$. Thus, $P_s$ commutes with $R$ for every $s\in \mathscr S$, and~(\ref{exactdecomp}) is proved.
We now show that, conversely, (\ref{exactdecomp}) implies~(E). Suppose, to the contrary, that (\ref{exactdecomp}) holds and there exist $s_1,s_2\in\mathscr S$ such that $s_1\neq s_2$ and $g_T(s_1)=g_T(s_2)$ for every $T\in \mathfrak X$. We then have $\tilde PT=T\tilde P = g_T(s_1)\tilde P$ for every $T\in \mathfrak X$, where $\tilde P=P_{s_1}+P_{s_2}$ is the orthogonal projection of $\mathfrak H$ onto $\mathfrak S(s_1)\oplus \mathfrak S(s_2)$. It follows that every element of $\mathfrak X$ commutes with $\tilde PR\tilde P$ for every $R\in L(\mathfrak H)$. In view of~(\ref{exactdecomp}) and the characterization of elements of $\mathcal A(\mathfrak X)$ given above, this implies that $P_{s_1}$ commutes with $\tilde PR\tilde P$ for every $R\in L(\mathfrak H)$, which is obviously false. We thus arrive at a contradiction and our claim is proved.

The above discussion is based on the assumption that the Hilbert space $\mathfrak H$ falls into a direct sum of eigenspaces of symmetries. In general, this assumption does not hold and one has to use direct integral decompositions instead. More specifically, we can try to find a positive measure $\nu$, a $\nu$-measurable family $\mathfrak S$ of Hilbert spaces, and a unitary operator $V\colon \mathfrak H\to \int^\oplus \mathfrak S(s)\,d\nu(s)$ such that every $T\in \mathcal X$ is representable in the form $V^{-1}\mathcal T^{\nu,\mathfrak S}_g V$ for some $\nu$-measurable complex function $g$, where $T^{\nu,\mathfrak S}_g$
is the operator of multiplication by $g$ in $\int^\oplus \mathfrak S(s)\,d\nu(s)$ (we
refer the reader to Appendix~\ref{app1} for the notions related to direct integrals of Hilbert spaces). A triple $(\nu,\mathfrak S,V)$ satisfying this condition is said to be a diagonalization for $\mathfrak X$. Since countable direct sums of separable Hilbert spaces can be always identified with suitable direct integrals (see Sec.~\ref{s_const_fam} and formula~(\ref{3dirint})), the eigenspace decompositions considered above can be viewed as a particular case of diagonalizations, with the role of $\nu$ and $V$ played by the counting measure\footnote{The counting measure on $\mathscr S$ is, by definition, the measure assigning to every finite subset of $\mathscr S$ the number of points in this set.}  on $\mathscr S$ and the natural identification between $\mathfrak H$ and $\bigoplus_{s\in\mathscr S} \mathfrak S(s)$ respectively.

Given a diagonalization $(\nu,\mathfrak S,V)$, we would like to have a decomposition of $H$ of the form
\begin{equation}\label{dint_repr}
H = V^{-1} \int^\oplus \mathcal H(s) \,d\nu(s)\,V,
\end{equation}
where $\mathcal H$ is a $\nu$-measurable family of closed operators in $\mathfrak S$. This decomposition is an analogue of~(\ref{decomposition}) in the direct integral setting. As we have seen, such a decomposition, in general, does not exist even in the direct sum case. It turns out, however, that this problem does not arise if we confine ourselves to \textit{exact} diagonalizations that are singled out by the requirement that
\begin{equation}\label{exactcond}
V^{-1}\mathcal T^{\nu,\mathfrak S}_g V\in \mathcal A(\mathfrak X)
\end{equation}
for every $\nu$-measurable $\nu$-essentially bounded function $g$. It will be shown that, in contrast to direct sum decompositions, exact diagonalizations exist for every involutive set $\mathfrak X\subset L(\mathfrak H)$ of pairwise commuting operators. Exactness condition~(\ref{exactcond}) (which is an analogue of~(\ref{exactdecomp})) implies that $VHV^{-1}$ commutes with $\mathcal T^{\nu,\mathfrak S}_g$ for every $\nu$-measurable $\nu$-essentially bounded $g$ since, as shown above, $H$ commutes with all elements of $\mathcal A(\mathfrak X)$. This allows us to apply the von~Neumann's reduction theory~\cite{Neumann} (or, more precisely, its generalization~\cite{Nussbaum, Barriere1951} for the case of unbounded operators) to $VHV^{-1}$ and obtain decomposition~(\ref{dint_repr}) as a result. Moreover, given an exact diagonalization $(\nu,\mathfrak S,V)$ for $\mathfrak X$ and a decomposition of form~(\ref{dint_repr}), the symmetry preserving closed (resp., self-adjoint) extensions of $H$ are precisely the operators of the form
\begin{equation}\nonumber
V^{-1} \int^\oplus \tilde{\mathcal H}(s) \,d\nu(s)\,V,
\end{equation}
where $\tilde{\mathcal H}$ is a $\nu$-measurable family of operators in $\mathfrak S$ such that $\tilde{\mathcal H}(s)$ is a closed (resp., self-adjoint) extension of the partial operator $\mathcal H(s)$ for $\nu$-a.e. $s$.\footnote{Throughout the paper, a.e. means either 'almost every' or 'almost everywhere'.} Condition~(\ref{exactcond}) is, as a rule, inconvenient for concrete applications because it involves the algebra $\mathcal A(\mathfrak X)$ rather than the set $\mathfrak X$ itself. We shall see, however, that a simple exactness criterion similar to condition~(E) can be obtained under very mild restrictions on the measure $\nu$ (Theorem~\ref{t0a}). This generalized condition~(E) can usually be easily verified for concrete examples.

The results formulated above suggest that, in general, one can find the symmetry preserving self-adjoint extensions of $H$ by doing the following steps:
 \begin{itemize}
\item[(I)] Find an exact diagonalization $(\nu,\mathfrak S,V)$ for $\mathfrak X$.

\item[(II)] Compute partial operators $\mathcal H(s)$ satisfying~(\ref{dint_repr}).

\item[(III)] Find self-adjoint extensions of the partial operators $\mathcal H(s)$.
\end{itemize}

In practice, the operator $H$ often comes as the closure of some non-closed operator $\check H$. While $\check H$ is usually given by some explicit formula, finding an explicit description of $H$ may be a difficult task. We shall see, however, that every symmetry of $\check H$ is also a symmetry of $H$ (Lemma~\ref{l1}). Moreover, finding partial operators of $H$ can be effectively reduced to some computations involving $\check H$ (Proposition~\ref{p_reduction}). For this reason, the knowledge of $\check H$ is actually sufficient for doing steps (I)-(III) for $H$.

We apply the general construction described above to the three-dimensional model of an electron in the magnetic field of an
infinitely thin solenoid. The Hamiltonian for this model is formally given by the differential expression
\begin{equation}\label{diff_expr}
\sum_{j=1}^3\left(-i\partial_{x_j} + \mathscr A_j(x)\right)^2,
\end{equation}
where $x=(x_1,x_2,x_3)\in\R^3$ and the vector potential $\mathscr
A=(\mathscr A_1,\mathscr A_2,\mathscr A_3)$ has the form
\[
\mathscr A_1(x) = -\frac{\phi x_2}{x_1^2+x_2^2},\quad \mathscr A_2(x) = \frac{\phi x_1}{x_1^2+x_2^2},\quad \mathscr A_3(x)=0.
\]
Here, the real parameter $\phi$ is equal to the flux of the magnetic field through the solenoid divided by $2\pi$.

This model was originally considered by Aharonov and Bohm~\cite{AB1959}. Self-adjoint extensions in the two-dimensional variant of this model, as well as their eigenfunction expansions and corresponding scattering amplitudes, were analysed in~Refs.~\onlinecite{Tyutin1974,AdamiTeta1998,DabrowskiStovicek1998}.

The vector potential $\mathscr A$ is smooth on the domain $\mathscr O=\{(x_1,x_2,x_3)\in\R^3 : x_1^2+x_2^2>0\}$, which is obtained by excluding the $x_3$-axis from $\R^3$. Hence, (\ref{diff_expr}) naturally determines an operator $\h^\phi$ in the space $C^\infty(\mathscr O)$ of
smooth functions on~$\mathscr O$,
\begin{multline}
(\h^\phi \Phi)(x) = \sum_{j=1}^3\left(-i\partial_{x_j} + \mathscr A_j(x)\right)^2\Phi(x) = \\= \left(-\Delta +\frac{2i\phi}{x_1^2+x_2^2}(x_2\partial_{x_1}-x_1\partial_{x_2})+
\frac{\phi^2}{x_1^2+x_2^2}\right)\Phi(x),\quad \Phi\in C^\infty(\mathscr O). \label{checkH}
\end{multline}
Restricting $\mathscr H^\phi$ to the space $C_0^\infty(\mathscr O)$ of smooth functions on $\mathscr O$ with compact support and passing to $\Lambda$-equivalence classes, where $\Lambda$ is the Lebesgue measure on $\R^3$, we obtain a (non-closed) densely defined symmetric operator $\check H^\phi$ in $L_2(\R^3)$:
\begin{align}
& D_{\check H^\phi} = \left\{ [\Phi]_\Lambda : \Phi\in C_0^\infty(\mathscr O)\right\}, \nonumber\\
& \check H^\phi[\Phi]_\Lambda = [\h^\phi \Phi]_\Lambda,\quad \Phi\in C_0^\infty(\mathscr O), \nonumber
\end{align}
where $[\Phi]_\Lambda$ denotes the $\Lambda$-equivalence class corresponding to $\Phi$. We define the operator $H^\phi$ in $L_2(\R^3)$ as the closure of $\check H^\phi$,
\begin{equation}\label{Hphidef}
H^\phi = \overline{\check H^\phi}.
\end{equation}

Let $\mathcal G$ be the Abelian group of linear operators in $\R^3$ generated by translations along the $x_3$-axis and rotations around the $x_3$-axis.
Given $G\in \mathcal G$, we denote by $T_G$ the unitary operator in $L_2(R^3)$ taking $\Psi$ to $[\Psi\circ G^{-1}]_\Lambda$ for any $\Psi\in L_2(\R^3)$. Clearly, $G\to T_G$ is a linear representation of $\mathcal G$ in $L_2(\R^3)$. It is straightforward to check that $\check H^\phi$ (and, hence, $H^\phi$) commutes with $T_G$ for any $G\in \mathcal G$. The above abstract scheme can therefore be applied to $H=H^\phi$ and the set $\mathfrak X$ consisting of all $T_G$ with $G\in \mathcal G$. As a result, we shall describe all self-adjoint
extensions of $H^\phi$ commuting with $T_G$ for any $G\in \mathcal G$ and find eigenfunction expansions for such extensions.

The paper is organized as follows. In Sec.~\ref{sec_results}, we formulate the results concerning self-adjoint extensions and eigenfunction expansions for the Aharonov-Bohm model. In Secs.~\ref{s3}--\ref{s6}, we elaborate on the general scheme of constructing symmetry preserving extensions outlined above.
Sec.~\ref{s3} is devoted to preliminaries concerning the commutation properties of operators in Hilbert space. In Sec.~\ref{s5}, we establish the existence of exact diagonalizations and prove an exactness criterion, analogous to condition~(E) discussed above. In Sec.~\ref{s6}, we consider direct integral decompositions for closed operators possessing a given set of symmetries and describe their symmetry preserving extensions. The aim of Secs.~\ref{s_one}--\ref{s11} is to derive eigenfunction expansions for the Aharonov-Bohm model. Sec.~\ref{s_one} is concerned with self-adjoint extensions and eigenfunction expansions for one-dimensional Schr\"odinger operators. In Sec.~\ref{s_meas}, we treat the measurability questions for families of such operators. Secs.~\ref{s_one} and~\ref{s_meas} provide a background for performing steps~(I)-(III) for $H^\phi$ in Sec~\ref{s8}. As a result, we represent symmetry preserving self-adjoint extensions of $H^\phi$ as direct integrals of suitable one-dimensional Schr\"odinger operators. In Sec~\ref{s9}, we study a special construction, called the direct integral of measures,  which can be viewed as a generalization of direct products of measures and arises naturally in the spectral analysis of direct integrals of operators. In Sec.~\ref{s11}, we combine the results of Secs.~\ref{s_one}, \ref{s8}, and~\ref{s9} to obtain the eigenfunction expansions of symmetry preserving self-adjoint extensions of $H^\phi$. In Appendices~\ref{app0} and~\ref{app1}, we give the necessary background material concerning measure theory and direct integral decompositions of operators in Hilbert space respectively.

\section{Formulation of results for the Aharonov-Bohm model}
\label{sec_results}

In this section, we give a precise description of self-adjoint extensions of $H^\phi$ commuting with $T_G$ for all $G\in \mathcal G$ and their eigenfunction expansions. The proofs of results stated below will be given in Sec.~\ref{s11}.

Let $S=\Z\times \R$, where $\Z$ is the set of integers. For $\phi\in\R$, let the subset $A^\phi$ of $S$ be given by
\begin{equation}\label{Aphi}
A^\phi =\{(m,p)\in S: |m+\phi|<1 \}.
\end{equation}
For every Borel\footnote{We endow $S$ with the topology induced from $\R^2$. The Borel $\sigma$-algebra of $S$ consists of all Borel subsets of $\R^2$ that are contained in $S$. The Borel structure on $A^\phi$ is induced from $S$.} real function $\theta$ on $A^\phi$, we shall define a map $\mathcal W^\phi_\theta$ on $S\times \R$ such that
\begin{enumerate}
\item[(a)] For every $s\in S$ and $E\in\R$, $\mathcal W^\phi_\theta(s,E)$ is a smooth complex function on $\mathscr O$ that is locally square-integrable on $\R^3$.

\item[(b)] The equalities\footnote{For brevity, we let $\mathcal W^\phi_\theta(s,E|x)$ denote the value of the function $\mathcal W^\phi_\theta(s,E)$ at the point $x$: $\mathcal W^\phi_\theta(s,E|x)=(\mathcal W^\phi_\theta(s,E))(x)$. A similar notation will be used for any maps whose values are also maps.}
\begin{align}
& \mathscr H^\phi \mathcal W^\phi_\theta(s,E) = (E+p^2) \mathcal W^\phi_\theta(s,E),\nonumber\\
& -i\partial_{x_3}\mathcal W^\phi_\theta(s,E|x) = p\mathcal W^\phi_\theta(s,E|x),\nonumber \\
& -i (x_1\partial_{x_2} - x_2\partial_{x_1})\mathcal W^\phi_\theta(s,E|x) = m\mathcal W^\phi_\theta(s,E|x)\nonumber
\end{align}
hold for every $s=(m,p)\in S$, $E\in\R$, and $x=(x_1,x_2,x_3)\in \mathscr O$.
\end{enumerate}

Furthermore, we shall define a positive Borel measure $M^\phi_\theta$ on $S\times \R$ such that the following statement is valid.

\begin{proposition}\label{propW}
Let $\phi\in\R$ and $\theta$ be a Borel real function on $A^\phi$. Then there is a unique unitary operator $W\colon L_2(\R^3)\to L_2(S\times\R,M^\phi_{\theta})$ such that
\begin{equation}\label{WWtheta12}
(W\Psi)(s,E) = \int \overline{\mathcal W^\phi_{\theta}(s,E|x)}\Psi(x)\,dx,\quad \Psi\in L_2^c(\R^3),
\end{equation}
for $M^\phi_{\theta}$-a.e. $(s,E)$, where the bar means complex conjugation and $L_2^c(\R^3)$ is the subspace of $L_2(\R^3)$ consisting of all its elements vanishing $\Lambda$-a.e. outside some compact subset of $\R^3$.
\end{proposition}

Condition~(b) means that $\mathcal W^\phi_\theta(m,s)$ is a joint generalized eigenfunction of operator~(\ref{diff_expr}) and $x_3$-components of momentum and angular momentum for every $s\in S$ and $E\in\R$. Proposition~\ref{propW} states that these eigenfunctions constitute a complete set and can be used to expand square-integrable functions on $\R^3$. We shall see (Theorem~\ref{main_theorem}) that all such eigenfunction expansions correspond to symmetry preserving self-adjoint extensions of $H^\phi$, but first we give precise definitions of $\mathcal W^\phi_\theta$ and $M^\phi_\theta$.

It follows from the second and third equalities in condition~(b) that
\begin{equation}\label{Wtheta12}
\mathcal W^\phi_\theta(s,E|x) = \frac{e^{ipx_3}}{2\pi\sqrt{r_x}}\left(\frac{x_1+ix_2}{r_x}\right)^m \mathcal J^\phi_\theta(s,E|r_x)
\end{equation}
for every $s=(m,p)\in S$, $E\in \R$, and $x\in\mathscr O$, where $r_x=\sqrt{x_1^2+x_2^2}$ and $\mathcal J^\phi_\theta(s,E)$ is some smooth function on $\R_+=(0,\infty)$ (note that $x_1\partial_{x_2} - x_2\partial_{x_1}$ and $r_x^{-m}(x_1+ix_2)^m$ correspond to $\partial_\varphi$ and $e^{im\varphi}$ respectively in the polar coordinates). For future convenience, we write the factor $(2\pi\sqrt{r_x})^{-1}$ separately rather than subsume it into $\mathcal J^\phi_\theta(s,E)$. The first equality in condition~(b) is fulfilled if and only if $\mathcal J^\phi_\theta(s,E)$ satisfies the one-dimensional Schr\"odinger equation
\begin{equation}\label{eqforJ}
-\partial^2_r\mathcal J^\phi_\theta(s,E|r) + \frac{(m+\phi)^2 -1/4}{r^2}\mathcal J^\phi_\theta(s,E|r) = E\mathcal J^\phi_\theta(s,E|r),\quad r\in \R_+,
\end{equation}
for every $s=(m,p)\in S$ and $E\in\R$.

To define $\mathcal J^\phi_\theta$ explicitly, we introduce some additional notation. For any $E,\kappa\in\R$, we define the function $u^\kappa(E)$ on $\R_+$ by the relation
\begin{equation}\label{ukappa}
u^\kappa(E|r) = r^{1/2+\kappa}\mathcal X_\kappa(r^2 E),\quad r\in\R_+,
\end{equation}
where the entire function $\mathcal X_\kappa$ is given by
\begin{equation}\label{Xkappa}
\mathcal X_\kappa(\zeta) = \frac{1}{2^\kappa}\sum_{n=0}^\infty \frac{(-1)^n\zeta^n}{\Gamma(\kappa+n+1)n!2^{2n}},\quad \zeta\in\C.
\end{equation}
For $-1<\kappa<1$, $\vartheta\in\R$, and $E\in\R$, we define the function $u^\kappa_\vartheta(E)$ on $\R_+$ by setting
\begin{equation}\label{wkappa}
u^\kappa_\vartheta(E) = \frac{u^{\kappa}(E)\sin(\vartheta+\vartheta_\kappa)-u^{-\kappa}(E)\sin(\vartheta-\vartheta_\kappa)}{\sin\pi\kappa},\quad 0<|\kappa|<1,
\end{equation}
and
\begin{equation}\label{w(z)}
u^0_\vartheta(E|r) = \lim_{\kappa\to 0} u^\kappa_\vartheta(E|r)=u^0(E|r)\cos\vartheta+\frac{2}{\pi}\left[\left(\ln\frac{r}{2} + \gamma\right)u^0(E|r) -\sqrt{r}\,\mathcal Y(r^2 E)\right]\sin\vartheta,\quad r\in\R_+,
\end{equation}
where
\begin{equation}\label{varthetakappa}
\vartheta_\kappa=\frac{\pi\kappa}{2},
\end{equation}
the entire function $\mathcal Y$ is given by
\begin{equation}\label{Y}
\mathcal Y(\zeta) = \sum_{n=1}^\infty \frac{(-1)^nc_n}{(n!)^2 2^{2n}} \zeta^n,\quad c_n = \sum_{j=1}^n \frac{1}{j},
\end{equation}
and $\gamma = \lim_{n\to\infty} (c_n -\ln n)=0,577\ldots$ is the Euler constant.\footnote{To compute the limit of $u^\kappa_\vartheta(E|r)$ as $\kappa\to 0$, one has to apply L'H\^{o}pital's rule and use the equality $\Gamma'(1+n)/\Gamma(1+n)=c_n-\gamma$ (see~Ref.~\onlinecite{Bateman}, Sec.~1.7.1, formula~(9)).} We now define $\mathcal J^\phi_\theta$ by setting
\begin{equation}\label{Jphitheta}
\mathcal J^\phi_\theta(s,E) = \left\{
\begin{matrix}
u^{|m+\phi|}(E),& s\in S\setminus A^\phi,\\
u^{m+\phi}_{\theta(s)}(E),& s\in A^\phi,
\end{matrix}
\right.
\end{equation}
for every $s=(m,p)\in S$ and $E\in \R$.

The function $\mathcal X_\kappa$ is closely related to Bessel functions: for $\zeta\neq 0$, we have
\begin{equation}\label{bessel}
\mathcal X_\kappa(\zeta) = \zeta^{-\kappa/2}J_\kappa(\zeta^{1/2}),
\end{equation}
where $J_\kappa$ is the Bessel function of the first kind of order $\kappa$. Since $J_\kappa$ satisfies the Bessel equation, it follows that
\begin{equation}\label{eqf}
-\partial^2_r u^{\pm\kappa}(E|r) + \frac{\kappa^2-1/4}{r^2}u^{\pm\kappa}(E|r) = E\,u^{\pm\kappa}(E|r),\quad r\in\R_+,
\end{equation}
for every $E\neq 0$. By continuity, this is true for $E=0$ as well. Moreover, (\ref{eqf}) is satisfied for $u^\kappa_\vartheta(E)$ in place of $u^{\pm\kappa}(E)$ for every $-1<\kappa<1$, $E\in\R$, and $\vartheta\in\R$ (this obviously follows from~(\ref{wkappa}) for $0<|\kappa|<1$; for $\kappa=0$,
 we can express $u^0_\vartheta(E|r)$ in terms of the Bessel functions $J_0$ and $Y_0$ by means of the equality\footnote{See~Ref.~\onlinecite{Bateman}, Sec.~7.2.4, formula~(33).} $\pi Y_0(\zeta) = 2\left(\gamma+\ln(\zeta/2)\right) J_0(\zeta) - 2\mathcal Y(\zeta^2)$ and make use of the Bessel equation).
In view of~(\ref{Jphitheta}), we conclude that (\ref{eqforJ}) holds for every $s\in S$ and $E\in \R$. It follows immediately from~(\ref{Wtheta12}), (\ref{ukappa}), (\ref{wkappa}), (\ref{w(z)}), and~(\ref{Jphitheta}) that $\mathcal W^\phi_\theta(s,E)$ is locally square-integrable on $\R^3$ for every $s\in S$ and $E\in \R$. Thus, both conditions~(a) and~(b) are fulfilled for the map $\mathcal W^\phi_\theta$ defined by~(\ref{Wtheta12}) and~(\ref{Jphitheta}).

It should be noted that equation~(\ref{eqforJ}) and condition~(a) determine $\mathcal J^\phi_\theta(s,E)$ uniquely up to a numerical factor for $s\in S\setminus A^\phi$. At the same time, if $s\in A^\phi$, then there are two linearly independent solutions of~(\ref{eqforJ}) compatible with~(a) for every $E\in\R$. According to~(\ref{Jphitheta}), different choices of the function $\theta$ on $A^\phi$ allow us to pick different solutions of~(\ref{eqforJ}) and, consequently, lead to different eigenfunction expansions for operator~(\ref{diff_expr}).

We now proceed to define the measure $M^\phi_\theta$. The above definition of $\mathcal W^\phi_\theta$ involves the generalized eigenfunctions $u^\kappa(E)$ and $u^\kappa_\vartheta(E)$ of the one-dimensional Schr\"odinger operator
\begin{equation}\label{onedimdiff}
-\partial^2_r + \frac{\kappa^2-1/4}{r^2}.
\end{equation}
Similarly, $M^\phi_\theta$ will be constructed from one-dimensional measures $\mathcal V_\kappa$ and $\mathcal V_{\kappa,\vartheta}$ entering eigenfunction expansions associated with differential expression~(\ref{onedimdiff}) (see Propositions~\ref{leig1} and~\ref{leig2}).
For $\kappa\in\R$, the positive Borel measure $\mathcal V_\kappa$ on $\R$ is given by
\begin{equation}\label{measVkappa}
d\mathcal V_\kappa(E) = \frac{1}{2}\Theta(E) E^{|\kappa|}\, dE,
\end{equation}
where $\Theta$ is the Heaviside function, i.e., $\Theta(E)=1$ for $E\geq 0$ and $\Theta(E)=0$ for $E<0$.
Given $-1<\kappa<1$ and $\vartheta\in\R$, we define the positive Borel measure $\mathcal V_{\kappa,\vartheta}$ on $\R$ as follows. If $0<|\kappa|<1$, we set
\begin{equation}\label{measVkappatheta}
\mathcal V_{\kappa,\vartheta} =  \left\{
\begin{matrix}
\tilde {\mathcal V}_{\kappa,\vartheta},& \vartheta \in [-|\vartheta_\kappa|,|\vartheta_\kappa|]+\pi\Z,\\
\frac{\pi\sin\pi\kappa|E_{\kappa,\vartheta}|}{2\kappa\sin(\vartheta+\vartheta_\kappa)\sin(\vartheta-\vartheta_\kappa)}\delta_{E_{\kappa,\vartheta}}+\tilde{\mathcal V}_{\kappa,\vartheta},& \vartheta\in (|\vartheta_\kappa|,\pi-|\vartheta_\kappa|)+\pi\Z,
\end{matrix}
\right.
\end{equation}
where $\vartheta_\kappa$ is defined by~(\ref{varthetakappa}),
the positive Borel measure $\tilde{\mathcal V}_{\kappa,\vartheta}$ on $\R$ is given by
\begin{equation}\label{tildeVkappatheta}
d\tilde{\mathcal V}_{\kappa,\vartheta}(E) = \frac{1}{2}\frac{\Theta(E) \sin^2\pi\kappa}{E^{-\kappa}\sin^2(\vartheta+\vartheta_\kappa) - 2\cos\pi\kappa\sin(\vartheta+\vartheta_\kappa)\sin(\vartheta-\vartheta_\kappa) + E^{\kappa}\sin^2(\vartheta-\vartheta_\kappa)}\,dE
\end{equation}
and $\delta_{E_{\kappa,\vartheta}}$ is the Dirac measure at the point
\begin{equation}\label{Ekappavartheta}
E_{\kappa,\vartheta} = -\left(\frac{\sin(\vartheta+\vartheta_\kappa)}{\sin(\vartheta-\vartheta_\kappa)}\right)^{1/\kappa}.
\end{equation}
For $\kappa=0$, the measure $\mathcal V_{\kappa,\vartheta}$ is defined by taking the limit $\kappa\to 0$ in formulas~(\ref{measVkappatheta}), (\ref{tildeVkappatheta}), and~(\ref{Ekappavartheta}). This yields
\begin{equation}\label{measV0theta}
\mathcal V_{0,\vartheta} =  \left\{
\begin{matrix}
\tilde {\mathcal V}_{0,\vartheta},& \vartheta \in \pi\Z,\\
\frac{\pi^2|E_{0,\vartheta}|}{2\sin^2\vartheta}\delta_{E_{0,\vartheta}}+\tilde{\mathcal V}_{0,\vartheta},& \vartheta\notin \pi\Z,
\end{matrix}
\right.
\end{equation}
where
\begin{equation}\label{E0vartheta}
E_{0,\vartheta} = -e^{\pi\cot \vartheta}
\end{equation}
and the positive Borel measure $\tilde{\mathcal V}_{0,\vartheta}$ on $\R$ is given by
\begin{equation}\label{tildeV0theta}
d\tilde{\mathcal V}_{0,\vartheta}(E) = \frac{1}{2}\frac{\Theta(E)}{ (\cos\vartheta -\ln E\sin\vartheta/\pi)^2 + \sin^2\vartheta} dE.
\end{equation}

The next lemma shows, in particular, that the measure $\mathcal V_{\kappa,\vartheta}$ depends continuously on both $\kappa$ and $\vartheta$.

\begin{lemma}\label{l_borel}
Let $F$ be a bounded (continuous) Borel function on $\R$ with compact support. Then $(\kappa,\vartheta)\to \int F(E)\,d\mathcal V_{\kappa,\vartheta}(E)$ is a Borel (resp., continuous) function on $(-1,1)\times\R$ that is bounded on $[-\alpha,\alpha]\times \R$ for every $0\leq\alpha<1$.
\end{lemma}
\begin{proof}
See Theorem~4 in Ref.~\onlinecite{Smirnov2015}.
\end{proof}

Given a Borel real function $\theta$ on $A^\phi$, let $\mu^\phi_{\theta}$ be the measure-valued map on $S$ such that
\begin{equation}\label{mutheta12}
\mu^\phi_{\theta}(s) = \left\{
\begin{matrix}
\mathcal V_{|m+\phi|},& s\in S\setminus A^\phi,\\
\mathcal V_{m+\phi,\theta(s)},& s\in A^\phi,
\end{matrix}
\right.
\end{equation}
for every $s=(m,p)\in S$.

The next statement follows immediately from Lemma~\ref{l_borel}.

\begin{lemma}\label{c_borel}
Let $\phi\in\R$ and $\theta$ be a Borel real function on $A^\phi$. For any compact set $K\subset \R$, $s\to \mu^\phi_{\theta}(s|K)$ is a Borel function on $S$ that is bounded on every compact subset of $S$.
\end{lemma}

Let $\varrho$ be the counting measure on $\Z$. We define the positive Borel measure $\nu_0$ on $S=\Z\times \R$ by setting $\nu_0=\varrho\times\lambda$, where $\lambda$ is the
Lebesgue measure on $\R$. For any $\nu_0$-integrable $f$, the function $p\to f(m,p)$ is integrable for every $m\in\Z$ and we have
\[
\int_{S} f(m,p)\,d\nu_0(m,p) = \sum_{m\in\Z} \int_{-\infty}^\infty f(m,p)\,dp.
\]

\begin{proposition}\label{prop1}
Let $\phi\in\R$ and $\theta$ be a Borel real function on $A^\phi$. Then there is a unique positive Borel measure $M$ on $S\times \R$ such that
\begin{equation}\label{M(K'timesK)}
M(K'\times K) = \int_{K'} \mu^\phi_\theta(s|K)\,d\nu_0(s)
\end{equation}
for every compact sets $K'\subset S$ and $K\subset \R$. If $f$ is an $M$-integrable complex function, then the function $E\to f(s,E)$ is $\mu^\phi_\theta(s)$-integrable for $\nu_0$-a.e. $s$, the function $s\to \int f(s,E)\,d\mu^\phi_\theta(s|E)$ is $\nu_0$-integrable, and
\begin{equation}\label{intf(s,E)}
\int f(s,E)\, dM(s,E) = \int d\nu_0(s)\int f(s,E)\,d\mu^\phi_\theta(s|E).
\end{equation}
\end{proposition}

Note that the right-hand side of~(\ref{M(K'timesK)}) is well-defined in view of Lemma~\ref{c_borel}. Given $\phi\in\R$ and a Borel real function $\theta$ on $A^\phi$, we let $M^\phi_\theta$ denote the measure $M$ satisfying the conditions of Proposition~\ref{prop1}. We thus have explicitly defined all entities entering Proposition~\ref{propW}.

It is worth mentioning that the structure of the set $A^\phi$ depends on whether $\phi\in\Z$ or not. If $\phi\in\Z$, then we have $A^\phi=\{-\phi\}\times\R$ and defining a Borel function on $A^\phi$ amounts to defining a single Borel function on $\R$. On the other hand, if $\phi\notin\Z$, then $A^\phi=\{m_\phi,m_\phi+1\}\times \R$, where $m_\phi\in \Z$ is such that $-\phi-1<m_\phi<-\phi$, and every Borel function on $A^\phi$ is determined by two Borel functions on $\R$. Using functions on $A^\phi$ instead of functions on $\R$ allows us to treat both cases uniformly.

Given $\phi\in\R$ and a Borel real function $\theta$ on $A^\phi$, the operator $W$ satisfying the conditions of Proposition~\ref{propW} will be denoted by $W^\phi_\theta$. Let the function $\mathfrak f$ on $S\times\R$ be defined by the relation
\begin{equation}\label{mathfrakf}
\mathfrak f(m,p\,;E) = p^2+E.
\end{equation}
We define the self-adjoint operator $H^\phi_{\theta}$ in $L_2(\R^3)$ by setting
\begin{equation}\label{Htheta12}
H^\phi_{\theta} = (W^\phi_{\theta})^{-1} \mathcal T^M_{\mathfrak f} W^\phi_{\theta},
\end{equation}
where $M=M^\phi_{\theta}$ and $\mathcal T^M_{\mathfrak f}$ is the operator of multiplication by $\mathfrak f$ in $L_2(S\times\R,M)$ (see Sec.~\ref{s_l2}).

Let $\Phi\in C_0^\infty(\mathscr O)$, $\Psi=[\Phi]_\Lambda$, and $h_\Psi(s,E)$ be the right-hand side of~(\ref{WWtheta12}). In view of~(\ref{checkH}) and  the first equality in condition~(b), integrating by parts yields
\begin{equation}\nonumber
h_{\check H^\phi\Psi}(s,E) = \int \overline{\mathcal W^\phi_{\theta}(s,E|x)}(\mathscr H^\phi\Phi)(x)\,dx=\int \overline{(\mathscr H^\phi\mathcal W^\phi_{\theta}(s,E))(x)}\Psi(x)\,dx = \mathfrak f(s,E)h_\Psi(s,E)
\end{equation}
for every $s\in S$ and $E\in\R$. It follows that $W^\phi_\theta \check H^\phi\Psi = \mathcal T^M_{\mathfrak f} W^\phi_\theta\Psi$ for every $\Psi\in D_{\check H^\phi}$, i.e., $H^\phi_\theta$ is a self-adjoint extension of $\check H^\phi$ and, hence, of $H^\phi$.

Further, given $\alpha,\beta\in \R$, let $G_{\alpha\beta}\in \mathcal G$ be defined by the relation
\begin{equation}\label{Galphabeta}
G_{\alpha\beta}x = (x_1\cos\alpha-x_2\sin\alpha,x_1\sin\alpha+x_2\cos\alpha,x_3+\beta),
\end{equation}
where $x=(x_1,x_2,x_3)\in\R^3$. It follows from~(\ref{Wtheta12}) that
\[
\mathcal W^\phi_\theta(s,E|G_{\alpha\beta}x) = e^{im\alpha+ip\beta}\mathcal W^\phi_\theta(s,E|x),\quad x\in \mathscr O,
\]
for every $s=(m,p)\in S$ and $E\in\R$. Let $\Psi\in L^c_2(\R^3)$. Setting $\tilde \Psi = T_{G_{\alpha\beta}}\Psi$, we obtain
\[
h_{\tilde\Psi}(s,E) = \int \overline{\mathcal W^\phi_{\theta}(s,E|x)}\Psi(G_{\alpha\beta}^{-1} x)\,dx = g_{\alpha\beta}(s,E)h_\Psi(s,E)
\]
for every $s\in S$ and $E\in\R$, where $g_{\alpha\beta}$ is the function $(m,p;E)\to e^{-im\alpha-ip\beta}$ on $S\times \R$. As $L^c_2(\R^3)$ is dense in $L_2(\R^3)$, this implies that
\[
T_{G_{\alpha\beta}} = (W^\phi_{\theta})^{-1} \mathcal T^M_{g_{\alpha\beta}} W^\phi_{\theta}.
\]
Since every element of $\mathcal G$ is equal to $G_{\alpha\beta}$ for some $\alpha,\beta\in\R$ and $\mathcal T^M_{\mathfrak f}$ commutes with $\mathcal T^M_{g_{\alpha\beta}}$, it follows from~(\ref{Htheta12}) that $H^\phi_\theta$ commutes with $T_G$ for all $G\in \mathcal G$.

Let $\theta$ and $\tilde \theta$ be Borel real functions on $A^\phi$ such that $\theta(s)-\tilde\theta(s)\in\pi\Z$ for $\nu_0$-a.e. $s\in A^\phi$. By~(\ref{measVkappatheta})--(\ref{tildeV0theta}), we have $\mathcal V_{\kappa,\vartheta+\pi n} = \mathcal V_{\kappa,\vartheta}$ for every $-1<\kappa<1$, $\vartheta\in\R$, and $n\in\Z$. In view of~(\ref{mutheta12}), it follows that $\mu^\phi_\theta(s)=\mu^\phi_{\tilde\theta}(s)$ for $\nu_0$-a.e. $s$. The uniqueness statement of Proposition~\ref{prop1} therefore implies that $M^\phi_\theta=M^\phi_{\tilde\theta}$. By~(\ref{wkappa}) and~(\ref{w(z)}), we have $u^\kappa_{\vartheta+\pi n}(z)=e^{i\pi n}u^\kappa_\vartheta(z)$ for every $-1<\kappa<1$, $\vartheta\in\R$, $z\in\C$, and $n\in\Z$. It follows from~(\ref{Wtheta12}) and~(\ref{Jphitheta}) that
\begin{equation}\label{tildethetatheta}
\mathcal W^\phi_{\tilde\theta}(s,E|x) = f_{\theta\tilde\theta}(s,E)\mathcal W^\phi_{\theta}(s,E|x)
\end{equation}
for every $s\in S\setminus N$, $E\in\R$, and $x\in\mathscr O$, where $N=\{s\in A^\phi: \theta(s)-\tilde \theta(s)\notin \pi\Z\}$ is a $\nu_0$-null set and the Borel function $f_{\theta\tilde\theta}$ on $S\times \R$ is given by
\[
f_{\theta\tilde\theta}(s,E) = \left\{
\begin{matrix}
1,& s\in S\setminus A^\phi,\\
e^{i(\tilde\theta(s)-\theta(s))},& s\in A^\phi.
\end{matrix}
\right.
\]
It easily follows from Proposition~\ref{prop1} that $N\times \R$ is an $M$-null set, where $M=M^\phi_\theta=M^\phi_{\tilde\theta}$, and, therefore, $f_{\theta\tilde\theta}(s,E) = \pm 1$ for $M$-a.e. $(s,E)$. Proposition~\ref{propW} and formula~(\ref{tildethetatheta}) hence imply that $W^\phi_{\tilde\theta} = \mathcal T^M_{f_{\theta\tilde\theta}}W^\phi_\theta$. In view of~(\ref{Htheta12}), we conclude that $H^\phi_{\theta}=H^\phi_{\tilde\theta}$.

Thus, assuming Propositions~\ref{prop1} and~\ref{propW}, we have proved the next statement in one direction.

\begin{theorem}\label{main_theorem}
Let $\phi\in\R$. For every Borel real function $\theta$ on $A^\phi$, the operator $H^\phi_{\theta}$ is a self-adjoint extension of $H^\phi$ commuting with $T_G$ for any $G\in \mathcal G$. Conversely, every self-adjoint extension of $H^\phi$ commuting with $T_G$ for any $G\in \mathcal G$ is equal to $H^\phi_{\theta}$ for some Borel real function $\theta$ on $A^\phi$. Given Borel real functions $\theta$ and $\tilde \theta$ on $A^\phi$, we have $H^\phi_{\theta} = H^\phi_{\tilde \theta}$ if and only if $\theta(s)-\tilde\theta(s)\in\pi\Z$ for $\nu_0$-a.e. $s\in A^\phi$.
\end{theorem}

\section{Commutation of operators and von~Neumann algebras}
\label{s3}

In this section, we give some background material for the treatment of diagonalizations in Sec.~\ref{s5}. It should be noted that, as far as diagonalizations are concerned, the conditions imposed on $\mathfrak X$ in Introduction are excessively restrictive. In fact, it suffices to assume (and we do so in this section and Sec.~\ref{s5}) that $\mathfrak X$ is an arbitrary set of closed densely defined operators rather than an involutive subset of $L(\mathfrak H)$. In Sec.~\ref{s6}, however, where the reduction by symmetries is treated, the assumptions on $\mathfrak X$ are the same as in Introduction.

\begin{lemma} \label{l1}
Let $\mathfrak H$ be a Hilbert space, $T\in L(\mathfrak H)$, and $R$ be an operator in $\mathfrak H$ commuting with $T$. If $R$ is densely defined, then $R^*$ commutes with $T^*$. If $R$ is
closable, then the closure $\bar R$ of $R$ commutes with $T$.
\end{lemma}
\begin{proof}
Suppose $R$ is densely defined. Let $\Psi\in D_{R^*}$ and $\Phi=R^*\Psi$. Then we have $\langle
R\Psi',\Psi\rangle=\langle \Psi',\Phi\rangle$ for any $\Psi'\in D_R$. Hence, we obtain
\[
\langle R\Psi',T^*\Psi\rangle=\langle RT\Psi',\Psi \rangle=\langle T\Psi',\Phi \rangle = \langle \Psi',T^*\Phi \rangle,\quad \Psi'\in D_R.
\]
This means that $T^*\Psi\in D_{R^*}$ and $R^* T^* \Psi=T^*R^*\Psi$, i.e., $R^*$ commutes with $T^*$.

Suppose now that $R$ is closable. Let $\Psi\in D_{\bar R}$. Then there is a sequence $\Psi_n\in D_R$ such that $\Psi_n\to \Psi$ and $R\Psi_n\to \bar R \Psi$ in $\mathfrak H$. Since $R$ commutes with
$T$, we have $T\Psi_n\in D_R$ for all $n$. The continuity of $T$ implies that $T\Psi_n\to T\Psi$ and $RT\Psi_n = TR\Psi_n\to T\bar R \Psi$. This means that $T\Psi\in
D_{\bar R}$ and $\bar R T\Psi= T\bar R\Psi$.
\end{proof}

Given a set $\mathfrak X$ of closed densely defined operators in a Hilbert space $\mathfrak H$, let $\mathfrak X'$ denote its commutant, i.e., the subalgebra of $L(\mathfrak H)$ consisting of all operators commuting with every element of $\mathfrak X$. Let $\mathfrak X^*$ be the set consisting of the adjoints of the elements of $\mathfrak X$. By Lemma~\ref{l1}, we have
\begin{equation}\label{adjcomm}
(\mathfrak X')^*=(\mathfrak X^*)'.
\end{equation}
The set $\mathfrak X$ is involutive if and only if $\mathfrak X^*=\mathfrak X$.

Recall~\cite{Dixmier} that a subalgebra $\mathcal M$ of $L(\mathfrak H)$ is
called a von~Neumann algebra if it is involutive and coincides with its bicommutant~$\mathcal M''$. By the well-known von~Neumann's bicommutant theorem (see,
e.~g., Ref.~\onlinecite{Dixmier}, Sec.~I.3.4, Corollaire~2), an involutive subalgebra $\mathcal M$ of $L(\mathfrak H)$ is a von~Neumann algebra if and only if
it contains the identity operator and is closed in the strong operator topology. It follows from~(\ref{adjcomm}) that $\mathfrak X'$ is an involutive subalgebra of $L(\mathfrak H)$ for any involutive set $\mathfrak X$ of closed densely defined operators in $\mathfrak H$. Moreover, it is easy to show (see~Lemma~1 in~Ref.~\onlinecite{JMAA2013}) that $\mathfrak X'$ is always strongly closed and, therefore, is a von~Neumann algebra for involutive $\mathfrak X$ by the bicommutant theorem.

A closed densely defined operator $T$ in $\mathfrak H$ is called affiliated with a von~Neumann algebra $\mathcal M$ if $T$ commutes with every element of $\mathcal M'$.
If $\mathfrak X$ is a set of closed densely defined operators in $\mathfrak H$, then every element of $\mathfrak X$ is obviously affiliated with the algebra $\mathcal A(\mathfrak X)=(\mathfrak X\cup\mathfrak X^*)''$. As shown by the next lemma, $\mathcal A(\mathfrak X)$ is actually the smallest von~Neumann algebra with this property.

\begin{lemma}\label{l1a}
Let $\mathfrak X$ be a set of closed densely defined operators in a Hilbert space $\mathfrak H$ and $\mathcal M$ be a von~Neumann algebra in $\mathfrak H$. Then $\mathcal A(\mathfrak X)\subset \mathcal M$ if and only if every operator in $\mathfrak X$ is affiliated with $\mathcal M$.
\end{lemma}
\begin{proof}
If every element of $\mathfrak X$ is affiliated with $\mathcal M$, then $\mathcal M'\subset\mathfrak  X'$, whence $\mathcal M'\subset (\mathfrak X^*)'$ by~(\ref{adjcomm}). It follows that $\mathcal M'\subset \mathfrak X'\cap (\mathfrak X^*)'=(\mathfrak X\cup\mathfrak X^*)'$ and, hence, $\mathcal A(\mathfrak X)\subset \mathcal M''=\mathcal M$. Conversely, if $\mathcal A(\mathfrak X)\subset \mathcal M$, then $\mathcal M'\subset (\mathfrak X\cup\mathfrak X^*)'\subset \mathfrak X'$ and, hence, every element of $\mathfrak X$ is affiliated with $\mathcal M$.
\end{proof}

The algebra $\mathcal A(\mathfrak X)$ will be called the von~Neumann algebra generated by $\mathfrak X$, and $\mathfrak X$ will be referred to as a set of generators of $\mathcal A(\mathfrak X)$. If $\mathfrak X\subset L(\mathfrak H)$, then $\mathcal A(\mathfrak X)$ is just the smallest von~Neumann algebra containing $\mathfrak X$. By Lemma~\ref{l1a}, a closed densely defined operator $T$ is affiliated with a von~Neumann algebra $\mathcal M$ if and only if $\mathcal A(T)\subset \mathcal M$ (here and subsequently, we write $\mathcal A(T)$ instead of $\mathcal A(\{T\})$, where $\{T\}$ is the one-element set containing $T$).

Let $\mathfrak X\subset L(\mathfrak H)$ be an involutive set of pairwise commuting operators. Then $\mathfrak X\subset \mathfrak X'$ and, hence, $\mathcal A(\mathfrak X)\subset \mathfrak X'$ because $\mathfrak X'$ is a von~Neumann algebra. As $\mathfrak X'=\mathcal A(\mathfrak X)'$, it follows that $\mathcal A(\mathfrak X)$ is Abelian.

We say that two sets $\mathfrak X$ and $\mathfrak Y$ of closed densely defined operators in $\mathfrak H$ are equivalent if $\mathcal A(\mathfrak X) =
\mathcal A(\mathfrak Y)$. We say that $\mathfrak X$ is equivalent to a closed densely defined operator $T$ if $\mathfrak X$ is equivalent to the
one-element set $\{T\}$.

\begin{remark}
If $\mathfrak X\subset L(\mathfrak H)$ is an involutive set, then the smallest strongly closed algebra containing $\mathfrak X$ and the identity operator in $\mathfrak H$ is obviously a von~Neumann algebra and, hence, coincides with $\mathcal A(\mathfrak X)$. The above definition of $\mathcal A(\mathfrak X)$ therefore complies with that used in Introduction.
\end{remark}

Given a spectral measure $\mathcal E$ (see Sec.~\ref{s_spec_meas}), we denote by $\mathcal P_{\mathcal E}$ the set of all operators $\mathcal E(A)$, where $A$ is an $\mathcal E$-measurable set. As the elements of $\mathcal P_{\mathcal E}$ pairwise commute, the algebra $\mathcal A(\mathcal P_{\mathcal E})$ is Abelian. For an $\mathcal E$-measurable complex function $g$, we let $J^{\mathcal E}_g$ denote the integral of $g$ with respect to $\mathcal E$ (see Sec.~\ref{s_spec_meas}).
\begin{lemma}\label{l4d}
Let $\mathfrak H$ be a separable Hilbert space and $\mathcal E$ be a spectral measure in $\mathfrak H$. Then the following statements hold:
\begin{enumerate}
\item[$1.$] The algebra $\mathcal A(\mathcal P_{\mathcal E})$ coincides with the set of all $J^{\mathcal E}_g$, where $g$ is an $\mathcal E$-measurable $\mathcal E$-essentially bounded complex function.

\item[$2.$] A closed densely defined operator $T$ in $\mathfrak H$ is affiliated with $\mathcal A(\mathcal P_{\mathcal E})$ if and only if $T=J^{\mathcal E}_g$ for an $\mathcal E$-measurable complex function $g$.
\end{enumerate}
\end{lemma}
\begin{proof}
See Lemma~10 in~Ref.\onlinecite{JMAA2013}.
\end{proof}

A family of maps $\{g_\iota\}_{\iota\in I}$ is said to separate points of a set $S$ if $S\subset D_{g_\iota}$ for all $\iota\in I$ and for any two distinct elements $s_1$ and $s_2$ of $S$, there is
$\iota\in I$ such that $g_\iota(s_1)\neq g_\iota(s_2)$.

To cover both positive and spectral measures, the next definition is formulated in terms of a general $\mathfrak A$-valued measure (see Sec.~\ref{measures_sec}).
\begin{definition}\label{d_exact}
Let $\mathfrak A$ be a topological Abelian group and $\nu$ be a $\sigma$-finite $\mathfrak A$-valued measure. A family $\{g_\iota\}_{\iota\in I}$ of maps is said to be $\nu$-separating if $I$ is countable and $\{g_\iota\}_{\iota\in I}$ separates points of $\s_\nu\setminus N$ for some $\nu$-null set $N$.\footnote{Here and subsequently, $S_\nu$ denotes the largest $\nu$-measurable set, see Sec.~\ref{measures_sec}.}
\end{definition}

The next result gives a complete description of systems of generators for $\mathcal A(\mathcal P_{\mathcal E})$.

\begin{proposition}\label{t0}
Let $\mathfrak H$ be a separable Hilbert space, $\mathcal E$ be a standard$\,\footnote{See Sec.~\ref{s_standard}.}$ spectral measure in $\mathfrak H$, and $\mathfrak X$
be a set of closed densely defined operators in $\mathfrak H$. Then $\mathcal A(\mathfrak X)=\mathcal A(\mathcal P_{\mathcal E})$ if and only if the following
conditions hold
\begin{enumerate}
\item[$1.$] $\mathcal A(\mathfrak X)\subset \mathcal A(\mathcal P_{\mathcal E})$

\item[$2.$] There is an $\mathcal E$-separating family $\{g_\iota\}_{\iota\in I}$ of $\mathcal E$-measurable complex functions such that $J^{\mathcal E}_{g_\iota}\in\mathfrak X$
for all $\iota\in I$.
\end{enumerate}
\end{proposition}
\begin{proof}
By Lemma~\ref{l1a} and statement~2 of Lemma~\ref{l4d}, condition~1 holds if and only if every element of $\mathfrak X$ is equal to $J^{\mathcal E}_g$ for some $\mathcal E$-measurable complex function $g$. Hence the proposition follows from Theorem~3 in~Ref.~\onlinecite{JMAA2013}.
\end{proof}

\begin{example}\label{e_1}
Let $T$ be a normal\footnote{Recall that a closed densely defined linear operator $T$
in a Hilbert space is called normal if the operators $TT^*$ and $T^*T$ have the same domain of definition and coincide thereon. In particular, self-adjoint and unitary operators are normal.} operator in a separable Hilbert space and $\mathcal E_T$ be its spectral measure (see Sec.~\ref{s_spec_meas}). By Lemma~\ref{l1a} and statement~2 of Lemma~\ref{l4d}, condition~1 of Proposition~\ref{t0} is fulfilled for every set $\mathfrak X$ whose elements are functions of $T$. Let $g$ be the identical function on $\C$: $g(z)=z$, $z\in\C$. Then the family containing the single function $g$ separates points of $\C$, and Proposition~\ref{t0} implies that the operator $T=g(T)$ is equivalent to the set $\mathcal P_{\mathcal E_T}$ of its spectral projections. Let $\zeta\in \C$ and $h_\zeta$ be the function on $\C\setminus\{\zeta\}$ defined by the relation $h_\zeta(z) = (z-\zeta)^{-1}$. If $\zeta$ is not an eigenvalue of $T$, then $\mathcal E_T(\{\zeta\})=0$ and the family containing the single function $h_\zeta$ is $\mathcal E_T$-separating. It follows from Proposition~\ref{t0} that the operator $(T-\zeta)^{-1}=h_\zeta(T)$ is equivalent to $\mathcal P_{\mathcal E_T}$ (and, hence, to $T$). Let $A\subset\C$ be a set having an accumulation point in $\C$ and let $f_\zeta(z)=e^{\zeta z}$ for $\zeta\in A$ and $z\in\C$. It is easy to show (see example~7 in~Ref.~\onlinecite{JMAA2013} for details) that the family $\{f_\zeta\}_{\zeta\in A}$ contains a countable subfamily separating the points of $\C$. By Proposition~\ref{t0}, we conclude that the set of all operators $e^{\zeta T}$ with $\zeta\in A$ is equivalent to $\mathcal P_{\mathcal E_T}$.
\end{example}

\section{Diagonalizations}
\label{s5}

Given a positive $\sigma$-finite measure $\nu$, we say that a $\nu$-a.e. defined family $\mathfrak S$ of Hilbert spaces is $\nu$-nondegenerate if $\mathfrak S(s)\neq \{0\}$ for $\nu$-a.e. $s$.

\begin{definition}\label{df-diag}
Let $\mathfrak X$ be a set of closed densely defined operators in a Hilbert space $\mathfrak H$. A triple $(\nu,\mathfrak S,V)$, where $\nu$ is a positive $\sigma$-finite measure, $\mathfrak S$ is a $\nu$-nondegenerate $\nu$-measurable family of Hilbert spaces, and $V$ is a unitary operator from $\mathfrak H$ to $\int^\oplus \mathfrak S(s)\,d\nu(s)$, is called a diagonalization for $\mathfrak X$ if every $T\in\mathfrak X$ is equal to $V^{-1}\mathcal T_g^{\nu,\mathfrak S} V$ for some complex $\nu$-measurable function $g$. A diagonalization $(\nu,\mathfrak S,V)$ for $\mathfrak X$ is called exact if condition~(\ref{exactcond}) holds for any complex $\nu$-measurable $\nu$-essentially bounded function $g$.
\end{definition}

The next theorem gives an exactness criterion for diagonalizations, similar to condition~(E) discussed in Introduction in the context of direct sum decompositions.

\begin{theorem}\label{t0a}
Let $\mathfrak H$ be a separable Hilbert space and $\mathfrak X$ be a set of closed densely defined operators in $\mathfrak H$.
A diagonalization $(\nu,\mathfrak S,V)$ for $\mathfrak X$, where $\nu$ is standard, is exact if and only if
there is a $\nu$-separating family $\{g_\iota\}_{\iota\in I}$ of $\nu$-measurable complex functions such that
\begin{equation}\label{diag}
V^{-1}\mathcal T^{\nu,\mathfrak S}_{g_\iota} V\in \mathfrak X,\quad \iota\in I.
\end{equation}
\end{theorem}
Sometimes the analysis of diagonalizations for $\mathfrak X$ simplifies if we replace $\mathfrak X$ with an equivalent set $\mathfrak Y$. The next result shows that passing to equivalent sets is always possible.
\begin{proposition}\label{t0b}
Let $\mathfrak H$ be a separable Hilbert space and $\mathfrak X$ and $\mathfrak Y$ be equivalent sets of closed densely defined operators in $\mathfrak H$. Then every (exact) diagonalization for $\mathfrak X$ is an (exact) diagonalization for $\mathfrak Y$.
\end{proposition}

We say that normal operators $T_1$ and $T_2$ in $\mathfrak H$ commute if their spectral projections commute (if $T_2\in L(\mathfrak H)$, then this definition agrees with Definition~\ref{d0000}, see~Ref.~\onlinecite{Fuglede}). If a set $\mathfrak X$ of closed densely defined operators admits a diagonalization, then $\mathfrak X$ consists of normal pairwise commuting operators because operators of multiplication by functions are normal and commute with each other. The converse statement is provided by the next proposition.
\begin{proposition}\label{t_diag}
For every set of normal pairwise commuting operators in a separable Hilbert space, there exists a diagonalization $(\nu,\mathfrak S,V)$, where $\nu$ is standard.
\end{proposition}

Before proceeding with the proofs of the above results, we give some simple examples of diagonalizations.

\begin{example}\label{e_2}
Let $Z$ be the set of all absolutely continuous square-integrable complex functions on $\R$ having square-integrable derivatives. Let $P$ be the one-dimensional operator of momentum, i.e., the operator in $L_2(\R)$ with the domain $D_P = \{[f]_\lambda : f\in Z\}$ satisfying the relation
\[
(P[f]_\lambda)(x) = -if'(x),\quad f\in Z,
\]
for $\lambda$-a.e. $x$ (as in Sec.~\ref{sec_results}, $\lambda$ is the Lebesgue measure on $\R$). Let $\mathcal F\colon L_2(\R)\to L_2(\R)$ be the operator of the Fourier transformation: $(\mathcal F f)(p) = (2\pi)^{-1/2}\int f(x)e^{-ipx}\,dx$. Then we have $\mathcal F P \mathcal F^{-1}=\mathcal T^\lambda_g$, where $g$ is the identical function on $\R$: $g(p)=p$, $p\in\R$. By Theorem~\ref{t0a}, we conclude that $(\lambda, \mathcal I_{\C,\lambda},\mathcal F)$ is an exact diagonalization for $P$ (here, $\mathcal I_{\C,\lambda}$ is a constant family of Hilbert spaces, see Sec.~\ref{s_const_fam}; by~(\ref{l2dirint}), we have $L_2(\R)=\int^\oplus \mathcal I_{\C,\lambda}(x)\,dx$). In view of Example~\ref{e_1} and Proposition~\ref{t0b}, $(\lambda, \mathcal I_{\C,\lambda},\mathcal F)$ is also an exact diagonalization for the set $\{e^{iaP}\}_{a\in\R}$ of translations in $L_2(\R)$.
\end{example}

\begin{example}\label{e_3}
Let $P$, $\mathcal F$, and $g$ be as in Example~\ref{e_2} and let $H=P^2$ be the Hamiltonian of the one-dimensional free particle. Then $\mathcal F H \mathcal F^{-1}=\mathcal T^\lambda_{g^2}$ and, hence, $(\lambda, \mathcal I_{\C,\lambda},\mathcal F)$ is a diagonalization for $H$. Since the family containing the single function $g^2$ is not $\lambda$-separating, this diagonalization is not exact. Let $\R_+=(0,\infty)$ and $V\colon L_2(\R)\to L_2(\R_+,\C^2,\lambda)$ be the unitary operator such that $(Vf)(p) = ((\mathcal Ff)(p),(\mathcal Ff)(-p))$, $f\in L_2(\R)$, for $\lambda$-a.e. $p\in \R_+$. Then $V H V^{-1}$ is the operator of multiplication by $g^2$ in $L_2(\R_+,\C^2,\lambda)$. In view of~(\ref{l2dirint}), it follows that $V H V^{-1}= \mathcal T^{\lambda_+,\mathfrak S}_{g^2}$, where $\lambda_+ = \lambda|_{\R_+}$ and $\mathfrak S=\mathcal I_{\C^2,\lambda_+}$. Since $g^2$ separates points of $\R_+$, it follows from Theorem~\ref{t0a} that $(\lambda_+,\mathcal I_{\C^2,\lambda_+}, V)$ is an exact diagonalization for $H$.
\end{example}

We now turn to the proofs of Theorem~\ref{t0a} and Propositions~\ref{t0b} and~\ref{t_diag}. In the rest of this section, we assume that the Hilbert space $\mathfrak H$ is separable. For brevity, we say that $(\nu,\mathfrak S,V)$ is an $\mathfrak H$-triple if $\nu$ is a positive $\sigma$-finite measure, $\mathfrak S$ is a $\nu$-nondegenerate $\nu$-measurable family of Hilbert spaces, and $V$ is a unitary operator from $\mathfrak H$ to $\int^\oplus \mathfrak S(s)\,d\nu(s)$.

Given an $\mathfrak H$-triple $t=(\nu,\mathfrak S,V)$, we define the map $\mathcal E^t$ on $\sigma(D_\nu)$ by the relation
\[
\mathcal E^t(A) = V^{-1} \mathcal T^{\nu,\mathfrak S}_{\chi_A} V,\quad A\in\sigma(D_\nu),
\]
where $\chi_A$ is equal to unity on $A$ and vanishes on $\s_\nu\setminus A$.

\begin{lemma}\label{l_spec_m}
Let $t=(\nu,\mathfrak S,V)$ be an $\mathfrak H$-triple. Then $\mathcal E^t$ is a spectral measure such that $\mathcal E^t$-measurable and $\mathcal E^t$-null sets coincide with $\nu$-measurable and $\nu$-null sets respectively. For any $\nu$-measurable complex function $g$, we have $J^{\mathcal E^t}_g = V^{-1}\mathcal T^{\nu,\mathfrak S}_g V$.
\end{lemma}
\begin{proof}
Clearly, $\mathcal E^t$ is an $L(\mathfrak H)$-valued $\sigma$-additive function satisfying condition~(a) of Sec.~\ref{measures_sec}. Clearly, $\mathcal N_\nu\subset \mathcal N_{\mathcal E^t}$. If $N\in \mathcal N_{\mathcal E^t}$, then $\mathcal T^{\nu,\mathfrak S}_{\chi_N}=0$ and, hence, $\nu(N)=0$ because $\mathfrak S(s)\neq 0$ for $\nu$-a.e. $s$. It follows that $\mathcal N_\nu = \mathcal N_{\mathcal E^t}$ and, therefore, condition~(b) of Sec.~\ref{measures_sec} is also fulfilled. Thus, $\mathcal E^t$ is a spectral measure having the same null sets as $\nu$. Since $D_{\mathcal E^t}=\sigma(D_\nu)$, we have $\sigma(D_{\mathcal E^t})=\sigma(D_\nu)$, i.e., $\nu$-measurable sets coincide with $\mathcal E^t$-measurable sets. For any $A\in D_{\mathcal E^t}$ and $\Psi\in \mathfrak H$, we have $\mathcal E^t_\Psi(A) = \int_A \|(V\Psi)(s)\|^2\,d\nu(s)$ (see~Sec.~\ref{s_spec_meas}). Hence, a $\nu$-measurable function $g$ is $\mathcal E^t_\Psi$-integrable if and only if $s\to \|(V\Psi)(s)\|^2f(s)$ is a $\nu$-integrable function, in which case we have
\begin{equation}\label{equality_spec_int}
\int g(s)\,d\mathcal E^t_\Psi(s) = \int \|(V\Psi)(s)\|^2g(s)\,d\nu(s).
\end{equation}
Let $g$ be a $\nu$-measurable function and $\Psi\in \mathfrak H$. By~(\ref{dom}), we have
\begin{equation}\nonumber
\Psi\in D_{J^{\mathcal E^t}_g} \Longleftrightarrow \mbox{$|g|^2$ is $\mathcal E^t_\Psi$-integrable}\Longleftrightarrow  \mbox{$s\to \|g(s)(V\Psi)(s)\|^2$ is  $\nu$-integrable} \Longleftrightarrow V\Psi\in D_{\mathcal T^{\nu,\mathfrak S}_g}
\end{equation}
and, therefore, the domains of $J^{\mathcal E^t}_g$ and $V^{-1}\mathcal T^{\nu,\mathfrak S}_g V$ coincide. Now (\ref{spectral_integral}) and~(\ref{equality_spec_int}) imply that
\[
\langle\Psi, J^{\mathcal E^t}_g \Psi\rangle = \int g(s)\,d\mathcal E^t_\Psi(s) = \langle V\Psi, \mathcal T^{\nu,\mathfrak S}_g V\Psi\rangle
\]
for any $\Psi\in D_{J^{\mathcal E^t}_g}$. Hence, $J^{\mathcal E^t}_g = V^{-1}\mathcal T^{\nu,\mathfrak S}_g V$.
\end{proof}

Proposition~\ref{t0b} follows immediately from the next lemma.

\begin{lemma}\label{l_diag}
Let $\mathfrak X$ be a set of closed densely defined operators in $\mathfrak H$. An $\mathfrak H$-triple $t$ is an (exact) diagonalization for $\mathfrak X$ if and only if $\mathcal A(\mathfrak X)\subset \mathcal A(\mathcal P_{\mathcal E^t})$ (resp., $\mathcal A(\mathfrak X)= \mathcal A(\mathcal P_{\mathcal E^t})$).
\end{lemma}
\begin{proof}
Let $t$ be an $\mathfrak H$-triple. Lemma~\ref{l_spec_m}, statement~2 of Lemma~\ref{l4d}, and Lemma~\ref{l1a} imply that
\begin{multline}\nonumber
\mbox{$t$ is a diagonalization for $\mathfrak X$} \Longleftrightarrow\\ \mbox{every element of $\mathfrak X$ is equal to $J^{\mathcal E^t}_g$ for some $\mathcal E^t$-measurable $g$} \Longleftrightarrow \\  \mbox{every element of $\mathfrak X$ is affiliated with $\mathcal A(\mathcal P_{\mathcal E^t})$} \Longleftrightarrow \mathcal A(\mathfrak X)\subset \mathcal A(\mathcal P_{\mathcal E^t}).
\end{multline}
In view of statement~1 of Lemma~\ref{l4d}, an $\mathfrak H$-triple $t$ is an exact diagonalization for $\mathfrak X$ if and only if it is a diagonalization for $\mathfrak X$ and $\mathcal A(\mathfrak X)\supset \mathcal A(\mathcal P_{\mathcal E^t})$. By the above, these two conditions are equivalent to the equality $\mathcal A(\mathfrak X)= \mathcal A(\mathcal P_{\mathcal E^t})$.
\end{proof}

\begin{proof}[Proof of Theorem~$\mathrm{\ref{t0a}}$] Let $t=(\nu,\mathfrak S,V)$ be a diagonalization for $\mathfrak X$, where $\nu$ is standard. By Lemma~\ref{l_diag} and Proposition~\ref{t0}, $t$ is exact if and only if there is an $\mathcal E^t$-separating family $\{g_\iota\}_{\iota\in I}$ of $\mathcal E^t$-measurable complex functions such that $J^{\mathcal E^t}_{g_\iota}\in\mathfrak X$
for all $\iota\in I$. Hence, the required statement follows from Lemma~\ref{l_spec_m}.
\end{proof}

\begin{proof}[Proof of Proposition~$\mathrm{\ref{t_diag}}$]
Let $\mathfrak X$ be a set of pairwise commuting normal operators in $\mathfrak H$ and $\mathfrak Y = \bigcup_{T\in \mathfrak X} \mathcal P_{\mathcal E_T}$, where $\mathcal E_T$ is the spectral measure of $T$ (see Sec.~\ref{s_spec_meas}). Then $\mathfrak Y$ is an involutive subset of $L(\mathfrak H)$ whose elements pairwise commute and, therefore, $\mathcal M=A(\mathfrak Y)$ is an Abelian von~Neumann algebra. Since $\mathcal A(T) = \mathcal A(\mathcal P_{\mathcal E_T})\subset \mathcal M$ for any $T\in \mathfrak X$ (see Example~\ref{e_1}), every element of $\mathfrak X$ is affiliated with $\mathcal M$. It follows from Lemma~\ref{l1a} that $\mathcal A(\mathfrak X)\subset \mathcal M$. Hence, the algebra $\mathcal A(\mathfrak X)$ is Abelian. By Th\'eor\`eme~2 of Sec.~II.6.2 in~Ref.~\onlinecite{Dixmier}, there are a finite Borel measure $\nu$ on a compact metrizable
space, a $\nu$-measurable family $\mathfrak S$ of Hilbert spaces, and a unitary operator $V\colon \mathfrak H \to \int^\oplus \mathfrak
S(s)\,d\nu(s)$ such that $\mathcal A(\mathfrak X)$ coincides with the set of all operators $V^{-1}\mathcal T^{\nu,\mathfrak S}_g V$, where $g$ is a $\nu$-measurable
$\nu$-essentially bounded complex function. This means that $(\nu,\mathfrak S,V)$ is an exact diagonalization for $\mathcal A(\mathfrak X)$. Since $\mathfrak X$ is equivalent to $\mathcal A(\mathfrak X)$, Proposition~\ref{t0b} implies that $(\nu,\mathfrak S,V)$ is also an exact diagonalization for $\mathfrak X$.
\end{proof}

\section{Reduction by symmetries}
\label{s6}

In this section, we assume that
\begin{itemize}
\item[(A)] $\mathfrak H$ is a Hilbert space, $\mathfrak X$ is an involutive subset of $L(\mathfrak H)$, and $(\nu,\mathfrak S,V)$ is an exact digonalization for $\mathfrak X$.
\end{itemize}

Given a $\nu$-measurable family $\mathcal H$ of closed operators in $\mathfrak S$, we define the operator $Q_{\mathcal H}$ in $\mathfrak H$ by setting
\[
Q_{\mathcal H} = V^{-1} \int^\oplus \mathcal H(s) \,d\nu(s)\,V.
\]

The structure of closed operators in $\mathfrak H$ commuting with operators in $\mathfrak X$ and of their closed (in particular, self-adjoint) symmetry preserving extensions is described by the next theorem.

\begin{theorem} \label{t1}
Let $(\mathrm A)$ be satisfied. Then the following statements hold:
\begin{enumerate}
\item[$1.$] $H$ is a closed operator in $\mathfrak H$ commuting with all elements of $\mathfrak X$ if and only if $H=Q_{\mathcal H}$ for some $\nu$-measurable family $\mathcal H$ of closed operators in $\mathfrak S$.
\item[$2.$] Let $\mathcal H$ and $\tilde{\mathcal H}$ be $\nu$-measurable families of closed operators in $\mathfrak S$. Then $Q_{\tilde{\mathcal H}}$ is an extension of $Q_{\mathcal H}$ if and only if $\tilde{\mathcal H}(s)$ is an extension of $\mathcal H(s)$ for $\nu$-a.e.~$s$. In particular, $Q_{\tilde{\mathcal H}}=Q_{\mathcal H}$ if and only if $\tilde{\mathcal H}(s)=\mathcal H(s)$ for $\nu$-a.e. $s$.
\item[$3.$] Let $\mathcal H$ be a $\nu$-measurable family of closed operators in $\mathfrak S$. Then $Q_{\mathcal H}$ is self-adjoint if and only if $\mathcal H(s)$ is self-adjoint for $\nu$-a.e. $s$.
\end{enumerate}
\end{theorem}
\begin{proof}
Let $\mathcal H$ be a $\nu$-measurable family of closed operators in $\mathfrak S$. Since $\mathfrak X\subset L(\mathfrak H)$, every element of $\mathfrak X$ is representable in the form $V^{-1} \mathcal T^{\nu,\mathfrak S}_g V$, where $g$ is $\nu$-essentially bounded. By Proposition~\ref{p_di}, $\int^\oplus \mathcal H(s) \,d\nu(s)$ is a closed operator commuting with all $\mathcal T^{\nu,\mathfrak S}_g$ and, therefore, $Q_{\mathcal H}$ is a closed operator commuting with all elements of $\mathfrak X$. Conversely, let $H$ be a closed operator in $\mathfrak H$ commuting with all elements of $\mathfrak X$. Let $\mathcal M$ denote the subalgebra of $L(\mathfrak H)$ consisting of all operators commuting with $H$. By Lemma~1 in~Ref.~\onlinecite{JMAA2013}, $\mathcal M$ is strongly closed. Since $\mathfrak X$ is involutive, $\mathcal A(\mathfrak X)$ coincides with the smallest strongly closed subalgebra of $L(\mathfrak H)$ containing $\mathfrak X$ and the identity operator in $\mathfrak H$. We hence have $\mathcal A(\mathfrak X)\subset \mathcal M$, i.e., $H$ commutes with all operators in $\mathcal A(\mathfrak X)$. As the diagonalization $(\nu,\mathfrak S,V)$ is exact, it follows from~(\ref{exactcond}) that $VH V^{-1}$ commutes with all operators $\mathcal T^{\nu,\mathfrak S}_g$, where $g$ is a $\nu$-measurable $\nu$-essentially bounded function. By Proposition~\ref{p_di}, we have $VH V^{-1} = \int^\oplus \mathcal H(s) \,d\nu(s)$ for some $\nu$-measurable family $\mathcal H$ of closed operators in $\mathfrak S$. This means that $H=Q_{\mathcal H}$ and statement~1 is proved. Statements~2 and~3 follow immediately from Proposition~~\ref{l_di4b} and Corollary~\ref{c_di} respectively.
\end{proof}

In particular, Statements~1 and~2 of Theorem~\ref{t1} imply the existence and uniqueness (up to $\nu$-equivalence) of decomposition~(\ref{dint_repr}) for any closed operator $H$ commuting with all elements of $\mathfrak X$.

\begin{corollary}\label{c_red}
Let $(\mathrm A)$ be satisfied and $\mathcal H$ be a $\nu$-measurable family of closed operators in $\mathfrak S$. Then the closed (resp., self-adjoint) extensions of $Q_{\mathcal H}$ commuting with all elements of $\mathfrak X$ are precisely the operators $Q_{\tilde{\mathcal H}}$, where $\tilde{\mathcal H}$ is a $\nu$-measurable family of operators in $\mathfrak S$ such that $\tilde{\mathcal H}(s)$ is a closed (resp., self-adjoint) extension of $\mathcal H(s)$ for $\nu$-a.e. $s$.
\end{corollary}

In concrete examples, $H$ usually comes as the closure of some non-closed operator $\check H$. It is often possible to establish by a direct computation that $Q_{\mathcal H}$ is an extension of $\check H$ for some $\nu$-a.e. defined family $\mathcal H$ of closed operators. In this case, the next proposition may be used to prove that actually equality~$(\ref{dint_repr})$ holds.

\begin{proposition}\label{p_reduction}
Let $(\mathrm A)$ be satisfied, $\check H$ be an operator in $\mathfrak H$, and $\mathcal H$ be a $\nu$-a.e. defined family of closed operators in $\mathfrak S$ such that $Q_{\mathcal H}$ is an extension of $\check H$. Suppose $D_{\check H}$ is taken to itself by all operators in $\mathfrak X$ and there is a sequence $\xi_1,\xi_2,\ldots$ of elements of $V(D_{\check H})$ such that the linear span of $(\xi_j(s),\mathcal H(s)\xi_j(s))$ is dense in the graph $G_{\mathcal H(s)}$ of $\mathcal H(s)$ for $\nu$-a.e. $s$. Then $\mathcal H$ is $\nu$-measurable, $\check H$ is closable, and $\overline{\check H}=Q_{\mathcal H}$.
\end{proposition}
\begin{proof}
Since $\xi_j\in D_{V\check HV^{-1}}$ for all $j=1,2,\ldots$ and $\int^\oplus \mathcal H(s)\,d\nu(s)$ is an extension of $V\check HV^{-1}$, we have $\mathcal H(s)\xi_j(s) = (V\check HV^{-1}\xi_j)(s)$ for $\nu$-a.e. $s$ and, hence, $s\to \mathcal H(s)\xi_j(s)$ are $\nu$-measurable sections of $\mathfrak S$. It follows that the family $\mathcal H$ is $\nu$-measurable. Statement~1 of Theorem~\ref{t1} implies that $Q_{\mathcal H}$ is a closed operator commuting with all elements of $\mathfrak X$. Since $Q_{\mathcal H}$ is an extension of $\check H$ and $D_{\check H}$ is taken to itself by all operators in $\mathfrak X$, it follows that $\check H$ is closable and commutes with all elements of $\mathfrak X$. By Lemma~\ref{l1}, we conclude that $H=\overline{\check H}$ also commutes with all operators in $\mathfrak X$. By statement~1 of Theorem~\ref{t1}, there is a $\nu$-measurable family $\mathcal H_0$ of closed operators in $\mathfrak S$ such that $H = Q_{\mathcal H_0}$. As $Q_{\mathcal H}$ is a closed extension of $\check H$, it is also an extension of $H$. By statement~2 of Theorem~\ref{t1}, it follows that $\mathcal H(s)$ is an extension of $\mathcal H_0(s)$ for $\nu$-a.e. $s$. Since $V(D_{\check H})\subset V(D_{H})$ and $V(D_H)$ coincides with the domain of $\int^\oplus \mathcal H_0(s)\,d\nu(s)$, we have $\xi_j(s)\in D_{\mathcal H_0(s)}$ and $\mathcal H_0(s)\xi_j(s)=\mathcal H(s)\xi_j(s)$ for all $j$ and $\nu$-a.e. $s$. Thus, the linear span of $(\xi_j(s),\mathcal H(s)\xi_j(s))$ is contained in $G_{\mathcal H_0(s)}$ and, hence, $G_{\mathcal H_0(s)}$ is dense in $G_{\mathcal H(s)}$ for $\nu$-a.e. $s$. In view of the closedness of $\mathcal H_0(s)$, this implies that $\mathcal H_0(s)=\mathcal H(s)$ for $\nu$-a.e. $s$.
\end{proof}

\section{One-dimensional Schr\"odinger operators}
\label{s_one}

In this section, we recall basic facts (see, e.g., Refs.~\onlinecite{Naimark,Teschl,Weidmann}) concerning one-dimensional Schr\"odinger operators and then briefly discuss self-adjoint realizations of differential expression~(\ref{onedimdiff}) and their eigenfunction expansions. In Sec.~\ref{s8}, we shall use the results of Sec.~\ref{s6} to represent self-adjoint extensions of the Aharonov-Bohm Hamiltonian as direct integrals of one-dimensional Schr\"odinger operators of type~(\ref{onedimdiff}). The proof of Theorem~\ref{main_theorem} given in Sec.~\ref{s11} is based on combining such a representation with the analysis of the one-dimensional problem given in this section.

Let $-\infty \leq a < b \leq \infty$ and $\lambda_{a,b}$ be the restriction to $(a,b)$ of the Lebesgue measure $\lambda$ on $\R$.
We denote by $\mathcal D$ the space of all complex continuously differentiable functions on $(a,b)$ whose derivative is absolutely continuous on $(a,b)$ (i.e., absolutely continuous on every segment $[c,d]$ with $a<c\leq d<b$).
Given a locally integrable real function $q$ on $(a,b)$, we denote by $l_q$ the linear operator from $\mathcal D$ to the space of complex $\lambda_{a,b}$-equivalence classes such that
\begin{equation}\label{lq}
l_q f(r) = -f''(r)+ q(r) f(r)
\end{equation}
for $\lambda$-a.e. $r\in (a,b)$.
For every $c\in (a,b)$ and complex numbers $z_1$ and $z_2$, there is a unique solution $f$ of the equation $l_{q}f=0$ such that $f(c)=z_1$ and $f'(c)=z_2$. This implies that solutions of $l_{q}f=0$ constitute a two-dimensional subspace of $\mathcal D$.
For any functions $f,g\in \mathcal D$, their Wronskian $W_r(f,g)$ at point $r\in (a,b)$ is defined by the
relation
\begin{equation}\nonumber
W_r(f,g) = f(r)g'(r) - f'(r)g(r).
\end{equation}
Clearly, $r\to W_r(f,g)$ is an absolutely continuous function on $(a,b)$. If $f$ and $g$ are such that $r\to W_r(f,g)$ is a constant function on $(a,b)$ (this is the case, in particular, when $f$ and $g$ are solutions of $l_{q}f=l_{q}g=0$), its value will be denoted by $W(f,g)$.
Let
\begin{equation}\label{D_q}
\mathcal D_q = \{f\in \mathcal D : f \mbox{ and } l_q f \mbox{ are both square-integrable on } (a,b)\}.
\end{equation}
A $\lambda_{a,b}$-measurable complex function $f$ is said to be left (right) square-integrable on $(a,b)$ if $\int_a^c |f(r)|^2\,dr <\infty$ (resp., $\int_c^b
|f(r)|^2\,dx <\infty$) for any $c\in (a,b)$. The subspace of $\mathcal D$ consisting of left (right) square-integrable on $(a,b)$ functions $f$ such that $l_qf$ is also left (resp., right) square-integrable on $(a,b)$ will be denoted by $\mathcal D_q^l$ (resp., $\mathcal D_q^r$). We obviously have $\mathcal D_q = \mathcal D_q^l\cap \mathcal D_q^r$.
It follows from~(\ref{lq}) by integrating by parts that
\[
\int_c^d ((l_q f)(r)g(r) - f(r) (l_q g)(r))\,dr = W_d(f,g) - W_c(f,g)
\]
for every $f,g\in \mathcal D$ and $c,d\in (a,b)$. This implies the existence of limits $W_a(f,g) = \lim_{r\downarrow a} W_r(f,g)$ and $W_b(f,g) = \lim_{r\uparrow b} W_r(f,g)$ for $f,g\in \mathcal D_q^l$ and $f,g\in \mathcal D_q^r$ respectively. Moreover, it follows that\footnote{Throughout this section and Sec.~\ref{s_meas}, all equivalence classes are taken with respect to $\lambda_{a,b}$. We shall drop the subscript
and write $[f]$ instead of $[f]_{\lambda_{a,b}}$.}
\begin{equation}\label{anti}
\langle l_q f, [g]\rangle - \langle [f], l_q g\rangle = W_{b}(\bar f,g) - W_{a}(\bar f,g)
\end{equation}
for any $f,g\in \mathcal D_q$, where $\langle\cdot,\cdot\rangle$ is the scalar product in $L_2(a,b)$.

For any linear subspace $Z$ of $\mathcal D_q$, let $L_q(Z)$ be the linear operator in $L_2(a,b)$ defined by the relations
\begin{align}
& D_{L_q(Z)} = \{[f]: f\in Z\},\nonumber\\
& L_q(Z) [f] = l_q f,\quad f\in Z.\nonumber
\end{align}
We define the minimal operator $L_q$ by setting
\begin{equation}\label{L_q}
L_q = L_q(\mathcal D_q^0),
\end{equation}
where
\begin{equation}\label{D^0_q}
\mathcal D_q^0 = \{f\in \mathcal D_q: W_a(f,g)=W_b(f,g)=0\mbox{ for any }g\in \mathcal D_q\}.
\end{equation}
By~(\ref{anti}), the operator $L_q(Z)$ is symmetric if and only if $W_{a}(\bar f,g) = W_{b}(\bar f,g)$ for any $f,g\in Z$. In particular, $L_q$ is a symmetric operator.
Moreover, $L_q$ is closed and densely defined and its adjoint $L_q^*$ is given by
\begin{equation}\label{lq*}
L_q^* = L_q(\mathcal D_q)
\end{equation}
(see~Ref.~\onlinecite{Teschl}, Lemma~9.4).
If $T$ is a symmetric extension of $L_q$, then $L^*_q$ is an extension of $T^*$ and, hence, of $T$. In view of~(\ref{lq*}), this implies that $T$ is of the form $L_q(Z)$ for some subspace $Z$ of~$\mathcal D_q$.

If $W_{a}(f,g)=0$ for any $f,g\in \mathcal D_q^l$, then $q$ is said to be in the limit point case (l.p.c.) at $a$. Otherwise $q$ is said to be in the limit circle case (l.c.c.) at $a$. Similarly, $q$ is said to be in the l.p.c. at $b$ if $W_{b}(f,g)=0$ for any $f,g\in \mathcal D_q^r$ and to be in l.c.c. at $b$ otherwise. According to the well-known Weyl alternative (see, e.g., Ref~\onlinecite{Teschl}, Theorem~9.9), $q$ is in l.c.c. at $a$ if and only if all solutions of $l_{q}f = 0$ are left square-integrable on $(a,b)$ (and, hence, belong to $\mathcal D_q^l$).

If $q$ is in l.p.c. at both $a$ and $b$, then (\ref{lq*}) implies that $L_q^*$ is symmetric and, therefore, $L_q$ is self-adjoint.

Suppose now that $q$ is in l.c.c. at $a$.
Let $f_1,f_2\in \mathcal D$ be linearly independent functions such that $l_qf_1 = l_qf_2 = 0$. Then $f_1$ and $f_2$ are left square-integrable on $(a,b)$ and, therefore, the function
\begin{equation}\label{zeta}
\delta^g_q(r) = \frac{1}{W(f_1,f_2)}\left[ f_1(r)\int\limits_a^r (l_qg)(\rho)f_2(\rho)\,d\rho - f_2(r)\int\limits_a^r (l_qg)(\rho)f_1(\rho)\,d\rho\right]
\end{equation}
is well-defined and belongs to $\mathcal D$ for any $g\in \mathcal D_q$. It is straightforward to check that $\delta_q^g$ does not depend on the choice of the solutions $f_1$ and $f_2$ and $l_q \delta^g_q = l_q g$. Hence,
the function
\begin{equation}\label{sigma}
\gamma_q^g = g -\delta_q^g
\end{equation}
satisfies the equation
\begin{equation}\label{lqgamma}
l_q \gamma_q^g = 0.
\end{equation}

Let $q$ be in l.c.c. at $a$ and in l.p.c. at $b$. In this case, $L_q$ has deficiency indices $(1,1)$ and the self-adjoint extensions of $L_q$ are precisely the operators (see~Ref.~\onlinecite{Smirnov2015}, Lemma~11)
\begin{equation}\label{Lqf}
L_q^f = L_q(Z_q^f),
\end{equation}
where $f$ is a nontrivial real solution of $l_q f = 0$ and the subspace $Z_q^f$ of $\mathcal D_q$ is given by
\begin{equation}\label{Zqf}
Z_q^f = \{g\in \mathcal D_q: W_a(f,g) = 0\}.
\end{equation}
The operator $L_q^f$ determines $f$ uniquely up to a nonzero real coefficient. Since the deficiency indices of $L_q$ are both equal to $1$, the orthogonal complement $G_{L_q^f}\ominus G_{L_q}$ of the graph $G_{L_q}$
of $L_q$ in the graph $G_{L_q^f}$ of $L_q^f$ is one-dimensional.

\begin{lemma}\label{l_ext}
Suppose $q$ is in l.c.c. at $a$ and in l.p.c. at $b$. Let $T$ be a self-adjoint extension of $L_q$ and $g$ be a real function in $\mathcal D_q$ such that $[g]\in D_{T}\setminus
D_{L_q}$. Then $\gamma_q^g$ is a real nontrivial solution of~$(\ref{lqgamma})$ and $T = L_q^{\gamma_q^g}$.
\end{lemma}
\begin{proof}
The reality of $\gamma_q^g$ follows from~(\ref{zeta}) and~(\ref{sigma}) because we can choose $f_1$ and $f_2$ in (\ref{zeta}) to be real. It follows easily from~(\ref{zeta}) that
\begin{equation}
W_r(\delta_q^g,h) = \frac{1}{W(f_1,f_2)}\left[ W_r(f_1,h)\int_a^r (l_qg)(\rho)f_2(\rho)\,d\rho - W_r(f_2,h)\int_a^r
(l_qg)(\rho)f_1(\rho)\,d\rho\right],\quad r\in (a,b),\nonumber
\end{equation}
for any $h\in \mathcal D$, where $f_1,f_2$ are linearly independent solutions of $l_q f_{1,2} = 0$. This implies that
\begin{equation}\label{wr}
W_a(\delta_q^g,h) =0
\end{equation}
for any $h\in \mathcal D_q^l$ and, therefore, $W_a(g,h) = W_a(\gamma_q^g,h)$.
If $\gamma_q^g$ were trivial, we would have $W_a(g,h) = W_b(g,h) = 0$ for any $h\in \mathcal D_q$ (recall that $q$ is in l.p.c. at
$b$) and, hence, $[g]\in D_{L_q}$ by~(\ref{D^0_q}) and~(\ref{L_q}). Thus, $\gamma_q^g$ is nontrivial.
Let $f$ be a nontrivial real solution of $l_q f = 0$ such that
$T=L^f_q$. Then we have $W_a(f,g)=0$, and it follows from~(\ref{sigma}) and~(\ref{wr}) that $W(f,\gamma_q^g)=0$. This means that $\gamma_q^g = Cf$ for some real $C\neq 0$
and, therefore, $T=L^{\gamma_q^g}_q$.
\end{proof}

If $q$ is locally square-integrable on $(a,b)$, then $\mathcal D^0_q\supset C^\infty_0(a,b)$, where $C^\infty_0(a,b)$ is the space of smooth functions on $(a,b)$ with compact support. In view of~(\ref{L_q}), this implies that $([f],l_qf)\in G_{L_q}$ for any $f\in C^\infty_0(a,b)$.

\begin{lemma}\label{l_density}
There exists a countable set $A\subset C^\infty_0(a,b)$ such that the elements $([f],l_q f)$ with $f\in A$ are dense in the graph $G_{L_q}$ of $L_q$ for any locally square-integrable real function $q$ on $(a,b)$.
\end{lemma}
\begin{proof}
Given $a<\alpha\leq\beta<b$, we denote by $Y_{\alpha,\beta}$ the linear subspace of the space $C^\infty_0(a,b)$ consisting of all its elements vanishing outside of $[\alpha,\beta]$. We make $Y_{\alpha,\beta}$ a normed space by setting
\[
\|f\| = \sup_{\alpha\leq r\leq\beta} (|f(r)|+|f''(r)|),\quad f\in Y_{\alpha,\beta}.
\]
Let $C[\alpha,\beta]$ be the space of all continuous functions on the segment $[\alpha,\beta]$. Then $f\to (f,f'')$ is an isometric embedding of $Y_{\alpha,\beta}$ into $C[\alpha,\beta]^2$ if the latter space is endowed with the norm
\begin{equation}\nonumber
\|(f,g)\|=\sup_{\alpha\leq r\leq\beta} (|f(r)|+|g(r)|),\quad f,g\in C[\alpha,\beta].
\end{equation}
Note that the space $C[\alpha,\beta]^2$ with this norm is separable because $C[\alpha,\beta]$ endowed with its ordinary supremum norm is separable. Since every subspace of a separable metric space is separable, it follows that $Y_{\alpha,\beta}$ is separable. We now pick sequences $\alpha_1,\alpha_2,\ldots$ and $\beta_1,\beta_2,\ldots$ such that $a<\alpha_j\leq\beta_j<b$ for all $j$ and $\alpha_j\to a$ and $\beta_j\to b$ as $j\to\infty$. For each $j=1,2,\ldots$, we choose a countable dense subset $A_j$ of $Y_{\alpha_j,\beta_j}$ and define the set $A$ by the relation $A = \bigcup_{j=1}^\infty A_j$. If $q$ is a locally square-integrable real function on $(a,b)$, then $L_q$ is the closure of $L_q(C^\infty_0(a,b))$ by Lemma~17 of Ref.~\onlinecite{Smirnov2015}. This means that elements of the form $([f],l_q f)$ with $f\in C^\infty_0(a,b)$ are dense in the graph of $L_q$, and it suffices to prove that every such element can be approximated by $([g],l_qg)$ with $g\in A$. For this, we choose a $j$ such that $f\in Y_{\alpha_j,\beta_j}$ and find a sequence $g_1,g_2,\ldots$ of elements of $A_j$ converging to $f$ in $Y_{\alpha_j,\beta_j}$. As $l_q$ obviously induces a continuous map from $Y_{\alpha_j,\beta_j}$ to $L_2(a,b)$ and $f\to [f]$ is a continuous embedding of $Y_{\alpha_j,\beta_j}$ into $L_2(a,b)$, we conclude that $[g_k]\to [f]$ and $l_qg_k\to l_q f$ in $L_2(a,b)$ as $k\to\infty$.
\end{proof}

We now turn to the Schr\"odinger operators corresponding to differential expression~(\ref{onedimdiff}) that arises as a result of the separation of variables in the Aharonov-Bohm model. We therefore assume $a=0$ and $b=\infty$ and define the potential $q_\kappa$ on $\R_+$ by the relation
\begin{equation}\label{qkappa}
q_\kappa(r) = \frac{\kappa^2-1/4}{r^2}
\end{equation}
for every $\kappa\in\R$. Let
\begin{equation}\label{hkappa}
h_\kappa = L_{q_\kappa}.
\end{equation}
Our aim is to describe all self-adjoint extensions of $h_\kappa$ and their eigenfunction expansions.

The equation $l_{q_\kappa} f = 0$ has linearly independent solutions $r^{1/2\pm\kappa}$ for $\kappa\neq 0$ and $r^{1/2}$ and $r^{1/2}\ln r$ for $\kappa=0$.
This implies that
\begin{itemize}
\item[(i)] $q_\kappa$ is in l.p.c. at both $0$ and $\infty$ for $|\kappa|\geq 1$,
\item[(ii)] $q_\kappa$ is in l.p.c. at $\infty$ and in l.c.c. at $0$ for $|\kappa|<1$.
\end{itemize}

For every $E\in\R$, let the function $u^\kappa(E)$ on $\R_+$ be given by~(\ref{ukappa}). By~(\ref{eqf}), $u^\kappa(E)$ are generalized eigenfunctions for differential expression~(\ref{onedimdiff}). In particular, for $E=0$, we have
\begin{equation}\label{lukappa}
l_{q_\kappa}u^{\kappa}(0) = 0.
\end{equation}
Let $L_2^c(\R_+)$ denote the subspace of $L_2(\R_+)$ consisting of all its elements vanishing $\lambda$-a.e. outside some compact subset of $\R_+$.

If $|\kappa|\geq 1$, then the operator $h_\kappa$ is self-adjoint by~(i). It is well known\cite{Naimark,GesztesyZinchenko,GTV2010,KST} that this operator can be diagonalized by the Hankel transformation. In terms of the functions $u^\kappa(E)$, this result can be formulated as follows (see Theorem~1 and formula~(36) in~Ref.~\onlinecite{Smirnov2015}).

\begin{proposition}\label{leig1}
Let $|\kappa|\geq 1$ and the positive measure $\mathcal V_\kappa$ on $\R$ be defined by~$(\ref{measVkappa})$. Then there is a unique unitary operator $U_\kappa\colon L_2(\R_+)\to L_2(\R,\mathcal V_\kappa)$ such that
\begin{equation}\label{Ukappa}
(U_\kappa\psi)(E) = \int_0^\infty u^{|\kappa|}(E|r)\psi(r)\,dr,\quad \psi\in L_2^c(\R_+),
\end{equation}
for $\mathcal V_\kappa$-a.e. $E$. We have $U_\kappa h_\kappa U_\kappa^{-1} = \mathcal T^{\mathcal V_\kappa}_\iota$, where $\iota$ is the identity function on $\R$ (i.e., $\iota(E)=E$ for all $E\in\R$).
\end{proposition}

Suppose now that $-1<\kappa<1$. For every $\vartheta,E\in \R$, let the function $u^\kappa_\vartheta(E)$ on $\R_+$ be defined by~(\ref{wkappa}) and~(\ref{w(z)}).
As was shown in Sec.~\ref{sec_results}, equality~(\ref{eqf}) remains valid for $u^\kappa_\vartheta(E)$ in place of $u^{\pm\kappa}(E)$ and, hence, $u^\kappa_\vartheta(E)$ are generalized eigenfunctions for differential expression~(\ref{onedimdiff}). For $E=0$, this means that
\begin{equation}\label{lukappavartheta}
l_{q_\kappa}u^\kappa_\vartheta(0) = 0.
\end{equation}
It follows immediately from~(\ref{wkappa}) and~(\ref{w(z)}) that
\begin{equation}\label{ukappavartheta}
u^\kappa_\vartheta(E) = u^\kappa(E)\cos(\vartheta-\vartheta_\kappa) + w^\kappa(E)\sin(\vartheta-\vartheta_\kappa)
\end{equation}
for all $\vartheta,E\in\R$, where $\vartheta_\kappa$ is defined by~(\ref{varthetakappa}) and
\begin{equation}\label{wkappa(z)}
w^\kappa(E) = u^\kappa_{\pi/2+\vartheta_\kappa}(E),\quad E\in\R.
\end{equation}
By~(\ref{lukappavartheta}) and~(\ref{wkappa(z)}), we have
\begin{equation}\label{lu0}
l_{q_\kappa}w^\kappa(0)=0.
\end{equation}
It follows immediately from (\ref{ukappa}), (\ref{wkappa}), (\ref{w(z)}), and~(\ref{wkappa(z)}) that
\begin{equation}\label{wronskian}
W(u^\kappa(0),w^\kappa(0)) =\frac{2}{\pi}
\end{equation}
and, therefore, $u^\kappa(0)$ and $w^\kappa(0)$ are real linearly independent elements of $\mathcal D$ for every $-1<\kappa<1$. By~(\ref{lukappa}), (\ref{lukappavartheta}), (\ref{ukappavartheta}), and~(\ref{lu0}), we conclude that the set of all real nontrivial $f\in \mathcal D$ such that $l_{q_\kappa}f=0$ coincides with the set of all elements of the form $c\,u^\kappa_\vartheta(0)$, where $c,\vartheta\in\R$ and $c\neq0$. In view of~(\ref{hkappa}) and~(ii), this implies that the self-adjoint extensions of $h_\kappa$ are precisely the operators
\begin{equation}\label{htheta}
h_{\kappa,\vartheta} = L_{q_\kappa}^{u^\kappa_\vartheta(0)},
\end{equation}
where $\vartheta\in\R$. Given $\vartheta\in \R$, let the positive measure $\mathcal V_{\kappa,\vartheta}$ on $\R$ be defined by~$(\ref{measVkappatheta})$ and~$(\ref{measV0theta})$ for $0<|\kappa|<1$ and $\kappa=0$ respectively. The next proposition gives eigenfunction expansions for $h_{\kappa,\vartheta}$.

\begin{proposition}\label{leig2}
For every $-1<\kappa<1$ and $\vartheta\in \R$, there is a unique unitary operator $U_{\kappa,\vartheta}\colon L_2(\R_+)\to L_2(\R,\mathcal V_{\kappa,\vartheta})$ such that
\begin{equation}\label{Ukappatheta}
(U_{\kappa,\vartheta}\psi)(E) = \int_0^\infty u^\kappa_\vartheta(E|r)\psi(r)\,dr,\quad \psi\in L_2^c(\R_+),
\end{equation}
for $\mathcal V_{\kappa,\vartheta}$-a.e. $E$. We have $U_{\kappa,\vartheta} h_{\kappa,\vartheta} U_{\kappa,\vartheta}^{-1} = \mathcal T^{\mathcal V_{\kappa,\vartheta}}_\iota$,
where $\iota$ is the identity function on $\R$.
\end{proposition}
\begin{proof}
See Theorem~3 and formulas~(36) and~(57) in~Ref.~\onlinecite{Smirnov2015}.
\end{proof}

\begin{remark}
Another treatment of self-adjoint realizations of~(\ref{onedimdiff}) and their eigenfunction expansions can be found in~Ref.~\onlinecite{GTV2010}. Our consideration differs from that in~Ref.~\onlinecite{GTV2010} by the choice of parametrization of self-adjoint extensions. The advantage of our choice is that the generalized eigenfunctions $u^\kappa_\vartheta(E|r)$ and the measures $\mathcal V_{\kappa,\vartheta}$ are continuous in $\kappa$ at $\kappa=0$.
\end{remark}

\section{Measurable families of one-dimensional Schr\"odinger operators}
\label{s_meas}

In this section, we derive a criterion for the measurability of families of one-dimensional Schr\"odinger operators. In Sec.~\ref{s8}, we shall use this criterion together with Corollary~\ref{c_red} to obtain direct integral representations for self-adjoint extensions of $H^\phi$ commuting with $T_G$ for all $G\in \mathcal G$.

In what follows, we fix $-\infty \leq a < b \leq \infty$ and set $\mathfrak h = L_2(a,b)$.

Let $\nu$ be a $\sigma$-finite positive measure. A $\nu$-a.e. defined map $\xi$ is said to be a $\nu$-measurable family of functions on $(a,b)$ if $\xi(s)$ is a locally integrable complex function on $(a,b)$ for $\nu$-a.e. $s$ and $s\to \int_\alpha^\beta \xi(s|r)\,dr$ is a $\nu$-measurable complex function for any $a<\alpha\leq \beta<b$.

\begin{lemma}\label{aux_meas}
Let $\nu$ be a $\sigma$-finite positive measure. Then the following statements hold:
\begin{enumerate}
\item[$1.$] Let $\xi$ and $\eta$ be $\nu$-measurable families of locally square-integrable functions on $(a,b)$. Then $s\to \xi(s)\eta(s)$ is a $\nu$-measurable family of functions on $(a,b)$.

\item[$2.$] $\xi$ is an $\mathfrak h$-valued $\nu$-measurable map if and only if $\xi$ is a $\nu$-measurable family of functions on $(a,b)$ such that $\xi(s)\in \mathfrak h$ for $\nu$-a.e. $s$.

\item[$3.$] Let $\xi$ be a $\nu$-a.e. defined map such that $\xi(s)$ is a continuous function on $(a,b)$ for $\nu$-a.e. $s$. Then $\xi$ is a $\nu$-measurable family of functions on $(a,b)$ if and only if $s\to \xi(s|r)$ is a $\nu$-measurable complex function for any $r\in (a,b)$.

\item[$4.$] Let $r_0\in (a,b)$, $\xi$ be a $\nu$-measurable family of functions on $(a,b)$, and $\eta$ be a $\nu$-a.e. defined map such that, for $\nu$-a.e. $s$, $\eta(s)$ is a complex function on $(a,b)$ satisfying the equality $\eta(s|r) = \int_{r_0}^r \xi(s|r)\,dr$ for all $r\in (a,b)$. Then $\eta$ is a $\nu$-measurable family of functions on $(a,b)$.

\item[$5.$] Let $\xi$ be a $\nu$-measurable family of functions on $(a,b)$ such that $\xi(s)$ is absolutely continuous on $(a,b)$ for $\nu$-a.e. $s$. Then $s\to \xi(s)'$ is a $\nu$-measurable family of functions on $(a,b)$.
\end{enumerate}
\end{lemma}
\begin{proof}
1. For any $a<\gamma \leq \delta < b$, let $\chi_{\gamma,\delta}$ be the function on $(a,b)$ that is equal to unity on $[\gamma,\delta]$ and vanishes outside this segment. Then the linear span of all $[\chi_{\gamma,\delta}]$ with rational $\gamma$ and $\delta$ is dense in $\mathfrak h$. Let $e_1,e_2,\ldots$ be a basis in $\mathfrak h$ obtained by orthogonalization of this system. Given $a<\alpha\leq \beta<b$, we have\footnote{As in Sec.~\ref{s_one}, we write $[f]$ in place of $[f]_{\lambda_{a,b}}$, where $\lambda_{a,b}$ is the restriction to $(a,b)$ of the Lebesgue measure $\lambda$ on $\R$.}
\begin{multline}
\int_\alpha^\beta \xi(s|r)\eta(s|r)\,dr = \langle [\chi_{\alpha,\beta}\overline{\xi(s)}],[\chi_{\alpha,\beta}\eta(s)]\rangle = \\ = \sum_{i=1}^\infty \langle [\chi_{\alpha,\beta}\overline{\xi(s)}],e_i\rangle \langle e_i,[\chi_{\alpha,\beta}\eta(s)]\rangle = \sum_{i=1}^\infty \int_\alpha^\beta \xi(s|r) e_i(r)\,dr
\int_\alpha^\beta e_i(r)\eta(s|r)\,dr \nonumber
\end{multline}
for $\nu$-a.e. $s$, where $\langle\cdot,\cdot\rangle$ is the scalar product in $\mathfrak h$. Since each $e_i$ is $\lambda_{a,b}$-equivalent to a linear combination of $\chi_{\gamma,\delta}$, the right-hand side is a $\nu$-measurable complex function of $s$.

2. Let $\chi_{\alpha,\beta}$ and $e_i$ be as in the proof of~1. If $\xi$ is an $\mathfrak h$-valued $\nu$-measurable map, then $s\to \langle \psi,\xi(s)\rangle$ is a $\nu$-measurable function for any $\psi\in \mathfrak h$. In particular, $s\to \langle [\chi_{\alpha,\beta}],\xi(s)\rangle = \int_\alpha^\beta \xi(s|r)\,dr$ is a $\nu$-measurable function for any $a<\alpha\leq \beta<b$, i.e., $\xi$ is a $\nu$-measurable family of functions on $(a,b)$. Conversely, if $\xi$ is a $\nu$-measurable family of functions on $(a,b)$, then $s\to  \int_a^b e_i(r)\xi(s|r)\,dr$ is a $\nu$-measurable function for every $i=1,2,\ldots$. If, in addition, $\xi(s)\in \mathfrak h$ for $\nu$-a.e. $s$, then $\int_a^b e_i(r)\xi(s|r)\,dr = \langle e_i,\xi(s)\rangle$ and, therefore, $\xi$ is an $\mathfrak h$-valued $\nu$-measurable map.

3. Let $\xi$ be a $\nu$-measurable family of functions on $(a,b)$ and $r\in (a,b)$. As $\xi(s)$ is continuous at $r$ for $\nu$-a.e. $s$, we have
$\xi(s|r) = \lim_{n\to\infty}n\int^{r+1/2n}_{r-1/2n} \xi(s|r)\,dr$ for $\nu$-a.e. $s$. This implies that $s\to \xi(s|r)$ is a $\nu$-measurable function because $s\to n\int^{r+1/2n}_{r-1/2n} \xi(s|r)\,dr$ are $\nu$-measurable functions. Conversely, let $s\to \xi(s|r)$ be a $\nu$-measurable function for every $r\in (a,b)$ and let $a<\alpha<\beta<b$. Given $n=1,2,\ldots$ and a function $f$ on $(a,b)$, we denote by $S_n(f)$ the Riemann sum $\frac{\beta-\alpha}{n}\sum_{k=1}^n f(r_n^k)$ for $f$ on $[\alpha,\beta]$, where $r_n^k = \alpha+(\beta-\alpha)k/n$. As $\xi(s)$ is continuous for $\nu$-a.e. $s$, $\int_\alpha^\beta\xi(s|r)\,dr$ is the limit of $S_n(\xi(s))$ for $\nu$-a.e. $s$. This means that $s\to \int_\alpha^\beta\xi(s|r)\,dr$ is a $\nu$-measurable function because $s\to S_n(\xi(s))$ is a $\nu$-measurable function for every $n$.

4. Clearly, $\eta(s)$ is a continuous function on $(a,b)$ for $\nu$-a.e. $s$ and $s\to \eta(s|r)$ is a $\nu$-measurable function for any $r\in (a,b)$. Hence, the statement follows from~3.

5. For $\nu$-a.e. $s$, the function $\xi(s)'$ is locally integrable and\footnote{We denote by $\xi'(s|r)$ the derivative of $\xi(s)$ at point $r$: $\xi'(s|r)=(\xi(s))'(r)$.} $\int_\alpha^\beta \xi'(s|r)\,dr = \xi(s|\beta)-\xi(s|\alpha)$. Hence, the statement follows from~3.
\end{proof}

\begin{lemma}\label{l_v}
Let $\nu$ be a $\sigma$-finite positive measure and $v$ be a $\nu$-measurable family of real locally square-integrable functions on $(a,b)$. Then $s\to L_{v(s)}$ is a $\nu$-measurable family of operators in $\mathfrak h$.
\end{lemma}
\begin{proof}
By Lemma~\ref{l_density}, there is a sequence $\varphi_1,\varphi_2,\ldots$ of functions in $C_0^\infty(a,b)$ such that the vectors $([\varphi_j],l_q\varphi_j)$ are dense in $G_{L_q}$ for any locally square-integrable real function $q$ on $(a,b)$. For each $j=1,2,\ldots$, let $\xi_j$ and $\eta_j$ be $\nu$-a.e. defined maps such that $\xi_j(s) = [\varphi_j]$ and $\eta_j(s) = l_{v(s)}\varphi_j$ for $\nu$-a.e. $s$. Since $v(s)$ is locally square integrable for $\nu$-a.e. $s$, the vectors $(\xi_j(s),\eta_j(s))$ are dense in $G_{L_{v(s)}}$ for $\nu$-a.e. $s$. By statements~1 and~5 of Lemma~\ref{aux_meas}, $\eta_j$ is a $\nu$-measurable family of functions on $(a,b)$, and statement~2 of Lemma~\ref{aux_meas} implies that $\eta_j$ is an $\mathfrak h$-valued $\nu$-measurable map for any $j=1,2,\ldots$. As $\xi_j$ are obviously $\mathfrak h$-valued $\nu$-measurable maps for all $j$, we conclude that $s\to L_{v(s)}$ is a $\nu$-measurable family of operators in $\mathfrak h$.
\end{proof}

\begin{lemma}\label{l_sol}
Let $\nu$ and $v$ be as in Lemma~$\mathrm{\ref{l_v}}$ and $\xi$ be a $\mathcal D$-valued $\nu$-a.e. defined map such that
\begin{equation}\label{main_eq}
l_{v(s)}\xi(s) =0
\end{equation}
for $\nu$-a.e. $s$. Suppose there
is $r_0\in (a,b)$ such that the functions $s\to \xi(s|r_0)$ and $s\to \xi'(s|r_0)$ are $\nu$-measurable. Then $\xi$ is
a $\nu$-measurable family of functions on $(a,b)$.
\end{lemma}
\begin{proof}
Clearly, it suffices to show that $\xi$ is
a $\nu$-measurable family of functions on $(\alpha,\beta)$ for any real numbers $\alpha$ and $\beta$ such that $a<\alpha<r_0<\beta<b$. We fix such $\alpha$ and $\beta$ and set
\begin{equation}
A_N = \left\{s\in \s_\nu :  \int\nolimits_\alpha^\beta v(s|r)^2\,dr < N^2/(\beta-\alpha)\right\}\nonumber
\end{equation}
for each $N=1,2,\ldots$.
By statement~1 of Lemma~\ref{aux_meas}, $A_N$ is a $\nu$-measurable set. By the Cauchy--Bunya\-kovsky inequality, we have
\begin{equation}\label{qn}
\int_\alpha^\beta |v(s|r)|\,dr < N,\quad s\in A_N.
\end{equation}
Let $\nu_N = \nu|_{A_N}$. To prove our statement, it suffices to show that the set
\[
Q_N = \{r\in (\alpha,\beta) : \mbox{$s\to \xi(s|r)$ and $s\to \xi'(s|r)$ are $\nu_N$-measurable functions}\} \nonumber
\]
coincides with $(\alpha,\beta)$ for all $N=1,2,\ldots$. Indeed, this condition implies that $s\to \xi(s|r)$ is a $\nu_N$-measurable function for all $N$ and $r\in (\alpha,\beta)$. Because $S_\nu\setminus \bigcup_{N=1}^\infty A_N$ is a $\nu$-null set, this means that $s\to \xi(s|r)$ is a $\nu$-measurable function for every $r\in (\alpha,\beta)$, and statement~3 of Lemma~\ref{aux_meas} ensures that $\xi$ is a $\nu$-measurable family of functions on $(\alpha,\beta)$. Note that $r_0\in Q_N$ for all $N$. Hence, to prove the equality $Q_N=(\alpha,\beta)$, it suffices to verify that the set $R_{N,r} = \{\rho\in(\alpha,\beta) : |\rho-r|<1/2N\}$ is contained in $Q_N$ for every $N=1,2,\ldots$ and $r\in Q_N$. Fix $N$ and $r\in Q_N$ and let $\xi_0,\xi_1,\ldots$ be $\mathcal D$-valued $\nu$-a.e. defined maps such that, for $\nu$-a.e. $s$, the relations
\begin{align}
& \xi_0(s|\rho) = \xi(s|r) + \xi'(s|r)(\rho-r),\nonumber \\
& \xi_n(s|\rho) = \xi_0(s|\rho) + \int_{r}^\rho d\rho'\int_{r}^{\rho'} v(s|t) \xi_{n-1}(s|t)\,dt,\quad n=1,2,\ldots, \label{eqn}
\end{align}
hold for all $\rho\in (a,b)$. As $r\in Q_N$, $\xi_0$ is a $\nu_N$-measurable family of functions on $(a,b)$, and it follows from statements~1 and~4 of Lemma~\ref{aux_meas} that $\xi_n$ is a $\nu_N$-measurable family of functions on $(a,b)$ for every $n=0,1,\ldots$.
Since $\xi(s)$ satisfies~(\ref{main_eq}) for $\nu$-a.e. $s$, it follows that, for $\nu$-a.e. $s$, the equality
\begin{equation}\label{int_eq}
\xi(s|\rho) = \xi_0(s|\rho) + \int_{r}^\rho d\rho'\int_{r}^{\rho'} v(s|t) \xi(s|t)\,dt
\end{equation}
holds for any $\rho\in (a,b)$. Let $\rho\in R_{N,r}$. It follows from~(\ref{qn}), (\ref{eqn}) and~(\ref{int_eq}) that
\[
|\xi(s|\rho)-\xi_n(s|\rho)|\leq  |\xi(s|\rho)-\xi_{n-1}(s|\rho)|/2
\]
for $\nu_N$-a.e. $s$ and all $n=0,1,\ldots$. We hence have
\[
|\xi(s|\rho)-\xi_n(s|\rho)| \leq |\xi(s|\rho)-\xi_0(s|\rho)|/2^n
\]
for $\nu_N$-a.e. $s$. This means that, for any $\rho\in R_{N,r}$, the sequence $\xi_n(s|\rho)$ converges to $\xi(s|\rho)$ for $\nu_N$-a.e. $s$ and, therefore, $s\to \xi(s|\rho)$ is a $\nu_N$-measurable function for every $\rho\in R_{N,r}$. As $\xi'(s|\rho)$ is the limit of $k(\xi(s|\rho+1/k)-\xi(s|\rho))$ as $k\to \infty$ for $\nu$-a.e. $s$, it also follows that $s\to \xi'(s|\rho)$ is a $\nu_N$-measurable function for every $\rho\in R_{N,r}$. Hence $R_{N,r}\subset Q_N$ and the lemma is proved.
\end{proof}

\begin{corollary}\label{cor_sol}
Let $\nu$ and $v$ be as in Lemma~$\mathrm{\ref{l_v}}$. Then there are $\nu$-measurable families $\xi_1$ and $\xi_2$ of functions on $(a,b)$ such that
$\xi_1(s)$ and $\xi_2(s)$ are linearly independent real elements of $\mathcal D$ satisfying equation ~$(\ref{main_eq})$ for $\nu$-a.e. $s$.
\end{corollary}
\begin{proof}
Choose $r_0\in (a,b)$. Let $\xi_1$ and $\xi_2$ be $\nu$-a.e. defined $\mathcal D$-valued maps such that $\xi_1(s)$ and $\xi_2(s)$ are solutions of~(\ref{main_eq}) satisfying the conditions
\[
\xi_1(s|r_0) = \xi_2'(s|r_0)=1,\quad \xi_1'(s|r_0) = \xi_2(s|r_0)=0
\]
for $\nu$-a.e. $s$. Obviously, $\xi_1(s)$ and $\xi_2(s)$ are linearly independent for $\nu$-a.e. $s$, and Lemma~\ref{l_sol} implies that $\xi_1$ and
$\xi_2$ are $\nu$-measurable families of functions on $(a,b)$.
\end{proof}

\begin{proposition}\label{l_ext1}
Let $\nu$ and $v$ be as in Lemma~$\mathrm{\ref{l_v}}$. Suppose $v(s)$ is in l.c.c. at $a$ and in l.p.c. at $b$ for $\nu$-a.e. $s$. If $\xi$ is a $\nu$-measurable family of elements of $\mathcal D$ such that $\xi(s)$ is a
nontrivial real solution of~$(\ref{main_eq})$ for $\nu$-a.e. $s$, then $s\to L^{\xi(s)}_{v(s)}$ is a $\nu$-measurable family of
self-adjoint operators in $\mathfrak h$. If $\mathcal R$ is a $\nu$-measurable family of operators in $\mathfrak h$ such that $\mathcal R(s)$ is a self-adjoint extension of $L_{v(s)}$ for $\nu$-a.e. $s$, then there exists a $\nu$-measurable family $\xi$ of elements of $\mathcal D$ such that $\xi(s)$ is a
nontrivial real solution of~$(\ref{main_eq})$ and $\mathcal R(s) = L^{\xi(s)}_{v(s)}$ for $\nu$-a.e.~$s$.
\end{proposition}
\begin{proof}
By Lemma~\ref{l_v}, $s\to L_{v(s)}$ is a $\nu$-measurable family of operators in $\mathfrak h$. This means that there is a sequence
$\zeta_1,\zeta_2,\ldots$ of $\mathfrak h\oplus \mathfrak h$-valued $\nu$-measurable maps such that the linear span of $\zeta_1(s),\zeta_2(s),\ldots$
is dense in the graph $G_{L_{v(s)}}$ of $L_{v(s)}$ for $\nu$-a.e.~$s$.

Let $\xi$ be a $\nu$-measurable family of elements of $\mathcal D$ such that $\xi(s)$ is a
nontrivial real solution of~$(\ref{main_eq})$ for $\nu$-a.e. $s$ and let $\tau$ be a smooth function on $(a,b)$ that is equal to unity in a neighborhood of $a$ and vanishes in a neighborhood of $b$. Let $\nu$-a.e. defined maps $g$ and $h$ be such that
\begin{align}
& g(s) = [\tau \xi(s)], \quad h(s) = -[\tau'' \xi(s) + 2\tau' \xi'(s)]\nonumber
\end{align}
for $\nu$-a.e. $s$. We obviously have $g(s)\in \mathfrak h$ and $h(s)\in \mathfrak h$ for $\nu$-a.e. $s$ and, therefore,
statements~1, 2, and~5 of Lemma~\ref{aux_meas} imply that $g$ and $h$ are $\mathfrak h$-valued $\nu$-measurable maps. Let $\zeta$ be an $\mathfrak h\oplus\mathfrak h$-valued $\nu$-measurable map such that $\zeta(s) = (g(s),h(s))$ for $\nu$-a.e.~$s$. Note that
\begin{equation}\label{lv(s)tauxi}
l_{v(s)} (\tau\xi(s)) = h(s)
\end{equation}
for $\nu$-a.e.~$s$ and, therefore, $\tau\xi(s)\in \mathcal D_{v(s)}$ for $\nu$-a.e. $s$ (see~(\ref{D_q})). Since $\xi(s)$ is a nontrivial real solution of~$(\ref{main_eq})$ and $\tau\xi(s)$ coincides with $\xi(s)$ in a neighborhood of $a$, it follows from~(\ref{D^0_q}) and~(\ref{Zqf}) that $\tau\xi(s)\in Z_{v(s)}^{\xi(s)}$ and $\tau\xi(s)\notin \mathcal D^0_{v(s)}$ for $\nu$-a.e. $s$. In view of~(\ref{L_q}), (\ref{Lqf}), and~(\ref{lv(s)tauxi}), we conclude that $\zeta(s)\in G_{L^{\xi(s)}_{v(s)}}\setminus G_{L_{v(s)}}$ for $\nu$-a.e.~$s$. As $G_{L^{\xi(s)}_{v(s)}}\ominus G_{L_{v(s)}}$ is one-dimensional, this implies that the linear span of the sequence $\zeta(s),\zeta_1(s),\zeta_2(s),\ldots$ is
dense in the graph of $L^{\xi(s)}_{v(s)}$ for $\nu$-a.e. $s$. This means that $s\to L^{\xi(s)}_{v(s)}$ is a $\nu$-measurable family of
operators in $\mathfrak h$.

Conversely, let $\mathcal R$ be a $\nu$-measurable family of operators in $\mathfrak h$ such that $\mathcal R(s)$ is a self-adjoint extension of $L_{v(s)}$ for $\nu$-a.e. $s$. Then both $s\to G_{\mathcal R(s)}$ and $s\to G_{L_{v(s)}}$ are
$\nu$-measurable families of subspaces of $\mathfrak h\oplus \mathfrak h$. By statement~1 of Lemma~\ref{l_BBB2}, $s\to G_{\mathcal R(s)}\ominus G_{L_{v(s)}}$ is also a $\nu$-measurable
family of subspaces of $\mathfrak h\oplus \mathfrak h$. Since $G_{\mathcal R(s)}\ominus G_{L_{v(s)}}$ is nontrivial for $\nu$-a.e. $s$, there is an $\mathfrak h\oplus\mathfrak h$-valued
$\nu$-measurable map $\eta$ such that $\eta(s)$ is a nonzero element of $G_{\mathcal R(s)}\ominus G_{L_{v(s)}}$ for $\nu$-a.e. $s$.
Let\footnote{Given $\psi=(\psi_1,\psi_2)\in \mathfrak h\oplus\mathfrak h$, we set $\bar\psi=(\bar\psi_1,\bar\psi_2)$, where $\bar\psi_1$ and $\bar \psi_2$ are elements of $\mathfrak h$ such that
$\bar \psi_{1,2}(r) = \overline{\psi_{1,2}(r)}$ for $\lambda$-a.e. $r\in (a,b)$.} $Q = \{s\in S_\nu: \eta(s) = -\overline{\eta(s)}\}$ and $\zeta$ be a $\nu$-a.e. defined map such that $\zeta(s) = i\eta(s)$ for $\nu$-a.e. $s\in Q$ and $\zeta(s) = \eta(s)+\overline{\eta(s)}$ for $\nu$-a.e. $s\in \s_\nu\setminus Q$. Then $\zeta$ is an $\mathfrak h\oplus\mathfrak h$-valued
$\nu$-measurable map such that
\begin{equation}\label{conj_zeta}
\zeta(s) = \overline{\zeta(s)}
\end{equation}
and $\zeta(s)\neq 0$ for $\nu$-a.e. $s$. Moreover, since $G_{\mathcal R(s)}$ and $G_{L_{v(s)}}$ are both invariant under complex conjugation, we have $\zeta(s)\in G_{\mathcal R(s)}\ominus G_{L_{v(s)}}$ for $\nu$-a.e. $s$. Let $g$ and $h$ be $\mathfrak h$-valued $\nu$-measurable maps such that $\zeta(s) = (g(s),h(s))$ for $\nu$-a.e. $s$. For $\nu$-a.e. $s$, we have
\begin{equation}\label{member}
g(s)\in D_{\mathcal R(s)}\setminus D_{L_{v(s)}}.
\end{equation}
As $g(s)\in D_{\mathcal R(s)}$ for $\nu$-a.e. $s$, there exists a $\nu$-a.e. defined map $\tilde g$ such that $\tilde g(s)\in \mathcal D_{v(s)}$ and $g(s) = [\tilde g(s)]$ for $\nu$-a.e. $s$. It follows from~(\ref{conj_zeta}) that $\tilde g(s)$ is real for $\nu$-a.e. $s$. By statement~2 of Lemma~\ref{aux_meas}, $\tilde g$ and $h$ are $\nu$-measurable families of functions on $(a,b)$.
Let $\xi$ be a $\nu$-a.e. defined map such that $\xi(s) = \gamma^{\tilde g(s)}_{v(s)}$ for $\nu$-a.e. $s$. In view of~(\ref{member}), it follows from Lemma~\ref{l_ext} that $\xi(s)$ is a nontrivial real solution of~$(\ref{main_eq})$ and $\mathcal R(s) = L^{\xi(s)}_{v(s)}$ for $\nu$-a.e. $s$. Let $\xi_1$ and $\xi_2$ be as in Corollary~\ref{cor_sol}.
Since $h(s) = \mathcal R(s)g(s)$, we have $h(s) = l_{v(s)}\tilde g(s)$ for $\nu$-a.e. $s$.
Equations~(\ref{zeta}) and~(\ref{sigma}) imply that, for $\nu$-a.e. $s$, equality
\begin{equation}
\xi(s|r) = \tilde g(s|r) - \frac{1}{W(\xi_1(s),\xi_2(s))}\left[ \xi_1(s|r)\int\limits_a^r h(s|\rho)\xi_2(s|\rho)\,d\rho - \xi_2(s|r)\int\limits_a^r h(s|\rho)\xi_1(s|\rho)\,d\rho\right]
\end{equation}
holds for all $r\in (a,b)$. By statements~3 and~5 of Lemma~\ref{aux_meas}, $s\to W(\xi_1(s),\xi_2(s))$ is a $\nu$-measurable function. It therefore follows from statements~1 and~4 of Lemma~\ref{aux_meas} that $\xi$ is a $\nu$-measurable family of functions on $(a,b)$.
\end{proof}

\section{Self-adjoint extensions of the three-dimensional \\ Aharonov--Bohm Hamiltonian}
\label{s8}

In this section, we use the results of Secs.~\ref{s6} and~\ref{s_meas} to represent the self-adjoint extensions of $H^\phi$ commuting with all operators $T_G$, $G\in \mathcal G$, as direct integrals of one-dimensional Schr\"odinger operators. Let $\mathfrak h = L^2(\R_+)$ and the positive Borel measure $\nu_0$ on $S=\Z\times \R$ be as in Sec.~\ref{sec_results}. We begin by constructing a suitable unitary operator $V\colon L_2(\R^3)\to L_2(S,\mathfrak h,\nu_0)$ and then prove that $(\nu_0,\mathcal I_{\mathfrak h,\nu_0},V)$ is actually an exact diagonalization for the set of the operators $T_G$. After that, we use Proposition~\ref{p_reduction} to obtain a representation of form~(\ref{dint_repr}) for $H^\phi$. Finally, we combine Corollary~\ref{c_red} and Proposition~\ref{l_ext1} to get an explicit description of the self-adjoint extensions of $H^\phi$ commuting with $T_G$.

We denote by $\lambda_+$ the restriction to $\R_+$ of the Lebesgue measure $\lambda$ on $\R$.
For $\Phi\in C_0^\infty(\mathscr O)$, let the map $\tilde\Phi$ from $S$ to $C^\infty_0(\R_+)$ be defined by the relation
\begin{equation}\label{tilde}
\tilde\Phi(s|r) = \frac{\sqrt{r}}{2\pi}\int_{-\infty}^\infty dx_3\int_{0}^{2\pi}d\varphi\, \Phi(r\cos\varphi,r\sin\varphi,x_3)e^{-ipx_3-im\varphi}
\end{equation}
for any $s=(m,p)\in S$ and $r>0$.

\begin{lemma}\label{l_unitary}
There is a unique unitary operator $V\colon L_2(\R^3)\to L_2(S,\mathfrak h,\nu_0)$ such that the equality
\begin{equation}\label{def_v}
(V[\Phi]_\Lambda)(s) = [\tilde\Phi(s)]_{\lambda_+}
\end{equation}
holds for $\nu_0$-a.e. $s$ for any $\Phi\in C_0^\infty(\mathscr O)$.
\end{lemma}
\begin{proof}
Let $\Phi\in C_0^\infty(\mathscr O)$. Then $s\to \int_c^d \tilde\Phi(s|r)\,dr$ is a continuous and, hence, $\nu_0$-measurable function for any $c,d\in\R_+$. This means that $\tilde \Phi$ is a $\nu_0$-measurable family of functions on $(0,\infty)$, and statement~2 of Lemma~\ref{aux_meas} implies that $s\to [\tilde\Phi(s)]_{\lambda_+}$ is an $\mathfrak h$-valued $\nu_0$-measurable map.
Let $f$ be the function on $\R^3$ defined by the formula
\[
f(x_1,x_2,p) = \int_{-\infty}^\infty \Phi(x_1,x_2,x_3) e^{-ipx_3}\,dx_3.
\]
Clearly, $f$ is continuous and, hence, Lebesgue measurable. It follows from the Fubini theorem and the Parseval identity for the one-dimensional Fourier transformation that $f$ is square-integrable and
\begin{equation}\label{parseval}
\int |f(x_1,x_2,p)|^2\,dx_1dx_2dp = 2\pi\int |\Phi(x_1,x_2,x_3)|^2 \,dx_1dx_2dx_3.
\end{equation}
Let $g$ be the continuous function on $\R\times\R\times\R_+$ defined by the relation $g(\varphi,p,r) = f(r\cos\varphi,r\sin\varphi,p)$. By the Fubini theorem, we have
\[
\tilde\Phi(s|r) = \frac{\sqrt{r}}{2\pi}\int_0^{2\pi} g(\varphi,p,r) e^{-im\varphi}\,d\varphi
\]
for any $s=(m,p)\in S$ and $r\in \R_+$. Further, by the Fubini theorem and Parseval identity for the Fourier series expansion, we have
\begin{multline}
\int d\nu_0(s)\int_0^\infty |\tilde\Phi(s|r)|^2\,dr = \int_0^\infty dr\int_{-\infty}^\infty dp \sum_{m\in\Z} |\tilde\Phi(m,p|r)|^2 = \\ =
\frac{1}{2\pi}\int_0^\infty r\,dr\int_{-\infty}^\infty dp\int_0^{2\pi}|g(\varphi,p,r)|^2\,d\varphi = \frac{1}{2\pi} \int |f(x_1,x_2,p)|^2\,dx_1dx_2dp\nonumber
\end{multline}
and in view of~(\ref{parseval}), we conclude that $s\to [\tilde\Phi(s)]_{\lambda_+}$ is a square-integrable map and
\[
\int \|[\tilde\Phi(s)]_{\lambda_+}\|^2\,d\nu_0(s) = \int d\nu_0(s)\int_0^\infty |\tilde\Phi(s|r)|^2\,dr = \int |\Phi(x_1,x_2,x_3)|^2 \,dx_1dx_2dx_3.
\]
Since $[\Phi]_\Lambda$ with $\Phi\in C_0^\infty(\mathscr O)$ are dense in $L_2(\R^3)$, it follows that there exists a unique isometric operator $V\colon L_2(\R^3)\to L_2(S,\mathfrak h,\nu_0)$ satisfying~(\ref{def_v}). It remains to check that the image of $V$ is the entire space $L_2(S,\mathfrak h,\nu_0)$. For any $\xi \in L_2(S,\mathfrak h,\nu_0)$ and $h\in \mathfrak h$, $s\to \langle \xi(s),h\rangle$ is a $\nu_0$-square integrable function because $|\langle \xi(s),h\rangle|\leq \|\xi(s)\|\|h\|$ for $\nu_0$-a.e. $s$. Given $\psi \in C_0^\infty(\R_+)$, $\xi \in L_2(S,\mathfrak h,\nu_0)$, and $m\in\Z$, we denote by $F_{\xi,\psi,m}$ the element of $L_2(\R)$ such that $F_{\xi,\psi,m}(p) = \langle \xi(m,p),[\psi]_{\lambda_+}\rangle$ for $\lambda$-a.e. $p$. For $\psi\in C_0^\infty(\R_+)$, $\chi\in C_0^\infty(\R)$, and $m\in \Z$, let $\Phi_{\psi,\chi,m}\in C_0^\infty(\mathscr O)$ be such that
\begin{equation}\label{Phi_psi}
\Phi_{\psi,\chi,m}(r\cos\varphi,r\sin\varphi,x_3) = \frac{1}{\sqrt{r}}\psi(r)\chi(x_3)e^{im\varphi}
\end{equation}
for any $r>0$ and $x_3,\varphi\in\R$. Then we have
\begin{equation}\label{tildePhi}
\tilde \Phi_{\psi,\chi,m}(k,p|r) = \delta_{km} \hat\chi(p)\psi(r),
\end{equation}
where $\delta_{km}=0$ for $k\neq m$ and $\delta_{km}=1$ for $k=m$ and $\hat\chi(p) = \int \chi(x_3)e^{-ipx_3}\,dx_3$ is the Fourier transform of $\chi$.
It follows that
\begin{equation}\label{scalar_pr}
\langle \xi, V[\Phi_{\psi,\chi,m}]_\Lambda\rangle_{L_2(S,\mathfrak h,\nu_0)} = \int \hat\chi(p)\langle\xi(m,p),[\psi]_{\lambda_+}\rangle_{\mathfrak h}\,dp = \langle \bar F_{\xi,\psi,m},[\hat\chi]_\lambda\rangle_{L_2(\R)}
\end{equation}
for every $\xi\in L_2(S,\mathfrak h,\nu_0)$. Suppose now that $\xi$ is orthogonal to every element of $\mathrm{Im}\,V$. Since $[\hat\chi]_\lambda$ with $\chi\in C_0^\infty(\R)$ are dense in $L_2(\R)$, equation~(\ref{scalar_pr}) implies that $F_{\xi,\psi,m} = 0$, i.e., $\langle \xi(s),[\psi]_{\lambda_+}\rangle = 0$ for $\nu_0$-a.e. $s$ for any $\psi\in C_0^\infty(\R_+)$. As $[\psi]_{\lambda_+}$ with $\psi\in C_0^\infty(\R_+)$ are dense in $\mathfrak h$, it follows that $\xi=0$ and, therefore, $\mathrm{Im}\,V = L_2(S,\mathfrak h,\nu_0)$.
\end{proof}

\begin{lemma}\label{l_d}
Let $V$ be as in Lemma~$\mathrm{\ref{l_unitary}}$. Then $(\nu_0,\mathcal I_{\mathfrak h,\nu_0},V)$ is an exact diagonalization for the set of the operators $T_G$ with $G\in\mathcal G$.
\end{lemma}
\begin{proof}
For $\alpha,\beta\in \R$, let $G_{\alpha\beta}\in \mathcal G$ be defined by~(\ref{Galphabeta}). Clearly, each element of $\mathcal G$ is equal to $G_{\alpha\beta}$ for some $\alpha,\beta\in \R$. We have
\begin{equation}\label{widetildePhicirc}
\widetilde{\Phi\circ G_{\alpha\beta}^{-1}}(m,p) = e^{-i\alpha m - i\beta p}\tilde\Phi(m,p),\quad (m,p)\in S,
\end{equation}
for any $\Phi\in C_0^\infty(\mathscr O)$ and, therefore,
\[
VT_{G_{\alpha\beta}} V^{-1} = \mathcal T^{\nu_0,\mathfrak S}_{\gamma_{\alpha\beta}},
\]
where $\mathfrak S = \mathcal I_{\mathfrak h,\nu_0}$ and the function $\gamma_{\alpha\beta}$ on $S$ is given by
\begin{equation}\nonumber
\gamma_{\alpha\beta}(m,p) = e^{-i\alpha m - i\beta p}.
\end{equation}
Suppose $(m,p)$ and $(m',p')$ are such that $\gamma_{\alpha\beta}(m,p)=\gamma_{\alpha\beta}(m',p')$ for all $(\alpha,\beta)\in \mathbb Q^2$, where $\mathbb Q$ is the set of rational numbers. Then we
have
\[
e^{i\alpha(m-m') +i\beta(p-p')}=1
\]
for all $(\alpha,\beta)\in \mathbb Q^2$. Since $\mathbb Q^2$ is dense in $\R^2$, it follows that $m=m'$ and $p=p'$. This means that the family
$\{\gamma_{\alpha,\beta}\}_{(\alpha,\beta)\in\mathbb Q^2}$ separates the points of $S$. The desired statement now follows from Theorem~\ref{t0a}.
\end{proof}

The next lemma gives a representation of form~(\ref{dint_repr}) for $H^\phi$.

\begin{lemma}\label{l100}
Let $\phi\in\R$, $V$ be as in Lemma~$\mathrm{\ref{l_unitary}}$, and $\mathcal H^\phi$ be the map on $S$ such that
\begin{equation}\label{mathcalHphi}
\mathcal H^\phi(m,p) = h_{m+\phi}+p^2 1_{\mathfrak h},\quad (m,p)\in S,
\end{equation}
where $1_{\mathfrak h}$ is the identity operator in $\mathfrak h$ and the operator $h_\kappa$, $\kappa\in\R$, is given by~$(\ref{hkappa})$. Then $\mathcal H^\phi$ is a $\nu_0$-measurable family of operators in $\mathfrak h$ and we have
\[
H^\phi = V^{-1}\int^\oplus \mathcal H^\phi(s)\,d\nu_0(s)\, V.
\]
\end{lemma}
\begin{proof}
It easily follows from (\ref{checkH}) that
\begin{equation}\label{polar}
(\mathscr H^\phi\Phi)(r\cos \varphi,r\sin\varphi,x_3) = \left(-\partial_{x_3}^2-
\partial_r^2-\frac{1}{r}\partial_r - \frac{1}{r^2}(\partial_\varphi^2+2i\phi\partial_\varphi-\phi^2)\right)
F_\Phi(r,\varphi,x_3)
\end{equation}
for any $\Phi\in C^\infty_0(\mathscr O)$, where the operator $\mathscr H^\phi$ in $C^\infty_0(\mathscr O)$ is given by~(\ref{checkH}) and $F_\Phi$ is the smooth function on $\R\times\R\times\R_+$ which represents $\Phi$ in the
cylindrical coordinates,
\[
F_\Phi(r,\varphi,x_3) = \Phi(r\cos \varphi,r\sin\varphi,x_3).
\]
Substituting (\ref{polar}) in~(\ref{tilde}) and integrating by parts yields
\begin{equation}\label{widetildecheck}
\widetilde{\mathscr H^\phi\Phi}(s) = -\tilde\Phi(s)'' +
q_{m+\phi}\tilde \Phi(s) + p^2\tilde\Phi(s)
\end{equation}
for every $s=(m,p)\in S$, where $q_\kappa$, $\kappa\in\R$, is given by~(\ref{qkappa}). Since $\tilde \Phi(s)\in C_0^\infty(\R_+)$ for all $s\in S$, it follows from~(\ref{def_v}) and~(\ref{hkappa}) that $(V[\Phi]_\Lambda)(s)\in D_{\mathcal H^\phi(s)}$ and
\begin{equation}\label{amp}
\mathcal H^\phi(s) (V[\Phi]_\Lambda)(s) = l_{q_{m+\phi}}\tilde\Phi(s)+p^2[\tilde\Phi(s)]_{\lambda_+} =  [\widetilde{\mathscr H^\phi\Phi}(s)]_{\lambda_+} = (V[\mathscr H^\phi\Phi]_\Lambda)(s) = (V\check H^\phi[\Phi]_\Lambda)(s)
\end{equation}
for $\nu_0$-a.e. $s=(m,p)$. Since every element of $D_{\check H^\phi}$ is equal to $[\Phi]_\Lambda$ for some $\Phi\in C^\infty_0(\mathscr O)$, this implies that $V^{-1}\int^\oplus \mathcal H^\phi(s)\,d\nu_0(s)\,V$ is an extension of $\check H^\phi$. In view of~(\ref{hkappa}), Lemma~\ref{l_density} implies the existence of a sequence $f_1,f_2,\ldots$ of elements of $C_0^\infty(\R_+)$ such that $([f_j]_{\lambda_+},l_{q_\kappa}f_j)$ are dense in the graph of $h_\kappa$ for any $\kappa\in\R$. Since $p^2 1_{\mathfrak h}$ is a bounded everywhere defined operator in $\mathfrak h$, it follows that the elements $([f_j]_{\lambda_+},\mathcal H^\phi(s)[f_j]_{\lambda_+})$ are dense in the graph of $\mathcal H^\phi(s)$ for all $s\in S$. Let $\chi$ be a nonzero element of $C_0^\infty(\R)$ and $\hat\chi(p) = \int_{-\infty}^\infty \chi(x)e^{-ipx}\,dx$ be its Fourier transform. For every  $j=1,2,\ldots$ and $m\in\Z$, we define $\xi_{jm}$ to be the $\nu_0$-equivalence class such that
\[
\xi_{jm}(k,p) = \delta_{km}\hat\chi(p)[f_j]_{\lambda_+}
\]
for $\nu_0$-a.e. $(k,p)$, where $\delta_{km}=0$ for $k\neq m$ and $\delta_{km}=1$ for $k = m$. Let $\Phi_{jm}\in C_0^\infty(\mathscr O)$ be defined by the relation $\Phi_{jm} = \Phi_{f_j,\chi,m}$, where $\Phi_{\psi,\chi,m}$ for $\psi\in C_0^\infty(\R_+)$ is given by~(\ref{Phi_psi}). It follows from~(\ref{def_v}) and~(\ref{tildePhi}) that $(V[\Phi_{jm}]_\Lambda)(s) = \xi_{jm}(s)$ for $\nu_0$-a.e. $s$ and, hence, $V[\Phi_{jm}]_\Lambda = \xi_{jm}$. Thus, $\xi_{jm}\in V(D_{\check H^\phi})$ for all $j=1,2,\ldots$ and $m\in \Z$. Since $\hat\chi$ is the restriction to $\R$ of a nontrivial entire function, the set of its zeros is at most countable. This implies that the elements $(\xi_{jm}(s),\mathcal H^\phi(s)\xi_{jm}(s))$ are dense in the graph of $\mathcal H^\phi(s)$ for $\nu_0$-a.e. $s$. The statement of the lemma therefore follows from~(\ref{Hphidef}) and Proposition~\ref{p_reduction}.
\end{proof}

Let $\phi\in\R$ and $\theta$ be a Borel real function on $A^\phi$, where $A^\phi$ is given by~(\ref{Aphi}).
We denote by $\mathcal H^\phi_{\theta}$ the operator-valued map on $S$ such that
\begin{equation}\label{htau12}
\mathcal H^\phi_{\theta}(s) = \left\{
\begin{matrix}
\mathcal H^\phi(s),& s\in S\setminus A^\phi,\\
h_{m+\phi,\theta(s)}+p^2 1_{\mathfrak h},& s\in A^\phi,
\end{matrix}
\right.
\end{equation}
for every $s=(m,p)\in S$, where $\mathcal H^\phi$ is defined by~(\ref{mathcalHphi}) and the operator $h_{\kappa,\vartheta}$ for $-1<\kappa<1$ and $\vartheta\in\R$ is given by~(\ref{htheta}).

\begin{lemma}\label{l_famsaext}
Let $\phi\in \R$. For any Borel real function $\theta$ on $A^\phi$, $\mathcal H^\phi_{\theta}$ is a $\nu_0$-measurable family of self-adjoint operators in $\mathfrak h$ such that $\mathcal H^\phi_{\theta}(s)$ is an extension of $\mathcal H^\phi(s)$ for every $s\in S$. If $\tilde{\mathcal H}$ is a $\nu_0$-measurable family of operators in $\mathfrak h$ such that $\tilde{\mathcal H}(s)$ is a self-adjoint extension of $\mathcal H^\phi(s)$ for $\nu_0$-a.e. $s$, then $\tilde{\mathcal H}$ is $\nu_0$-equivalent to $\mathcal H^\phi_\theta$ for some Borel real function $\theta$ on $A^\phi$.
\end{lemma}
\begin{proof}
Let $v$ be the map on $S$ defined by by the relation
\begin{equation}\label{vsr}
v(m,p) = q_{m+\phi},\quad (m,p)\in S,
\end{equation}
where $q_\kappa$, $\kappa\in\R$, is given by~(\ref{qkappa}). Then $v$ is a $\nu_0$-measurable family of functions on $\R_+$, and (\ref{hkappa}) implies that
\begin{equation}\label{lvs}
L_{v(s)} = h_{m+\phi}
\end{equation}
for any $s=(m,p)\in S$. Clearly, $v(s)$ is in l.c.c. at $0$ if and only if $s\in A^\phi$. Given a Borel real function  $\theta$ on $A^\phi$, we define the map $\xi_\theta$ on $A^\phi$ by setting
\begin{equation}\label{xitheta}
\xi_\theta(s) = u^{m+\phi}(0)\cos\left(\theta(s)-\vartheta_{m+\phi}\right) + w^{m+\phi}(0)\sin\left(\theta(s)-\vartheta_{m+\phi}\right)
\end{equation}
for every $s=(m,p)\in A^\phi$, where $u^\kappa$, $w^\kappa$, and $\vartheta_\kappa$ are given by~(\ref{ukappa}), (\ref{wkappa(z)}), and~(\ref{varthetakappa}) respectively. By~(\ref{lukappa}), (\ref{lu0}), and~(\ref{vsr}), we have $l_{v(s)}\xi_\theta(s)=0$ for every $s\in A^\phi$. It follows from~(\ref{ukappavartheta}), (\ref{htheta}), (\ref{vsr}), and~(\ref{xitheta}) that
\begin{equation}\label{mathcalRphi}
h_{m+\phi,\theta(s)} = L_{v(s)}^{\xi_\theta(s)},\quad s=(m,p)\in A^\phi.
\end{equation}
Since $\xi_\theta$ is obviously a $\nu_0|_{A^\phi}$-measurable family of functions on $\R_+$, it follows from Proposition~\ref{l_ext1} that $s\to h_{m+\phi,\theta(s)}$ is a $\nu_0$-measurable family of operators in $\mathfrak h$ on $A^\phi$. Since $(m,p)\to p^2 1_{\mathfrak h}$ is a $\nu_0$-measurable family of operators in $\mathfrak h$ by Lemma~\ref{l_BBB}, the $\nu_0$-measurability of $\mathcal H^\phi_{\theta}$ on $A^\phi$ follows from~(\ref{htau12}) and Lemma~\ref{l_BBB4a}. As $\mathcal H_\theta^\phi$ is $\nu_0$-measurable on $S\setminus A^\phi$ by~(\ref{htau12}) and Lemma~\ref{l100}, we see that $\mathcal H^\phi_\theta$ is a $\nu_0$-measurable family of operators in $\mathfrak h$. It follows immediately from~(\ref{mathcalHphi}) and (\ref{htau12}) that $\mathcal H^\phi_{\theta}(s)$ is a self-adjoint extension of $\mathcal H^\phi(s)$ for every $s\in S$.

Let $\tilde{\mathcal H}$ be a $\nu_0$-measurable family of operators in $\mathfrak h$ such that $\tilde{\mathcal H}(s)$ is a self-adjoint extension of $\mathcal H^\phi(s)$ for $\nu_0$-a.e. $s$. Let $\mathcal R$ be a $\nu_0$-a.e. defined family of operators in $\mathfrak h$ such that $\mathcal R(s)=\tilde{\mathcal H}(s)-p^2 1_{\mathfrak h}$ for $\nu_0$-a.e. $s=(m,p)$. By Lemma~\ref{l_BBB}, $\mathcal R$ is $\nu_0$-measurable. In view of~(\ref{mathcalHphi}) and~(\ref{lvs}), $\mathcal R(s)$ is a self-adjoint extension of $L_{v(s)}$ for $\nu_0$-a.e. $s$. By Proposition~\ref{l_ext1}, there exists a $\nu_0$-a.e. defined map $\xi$ on $A^\phi$ such that $\xi$ is a $\nu_0|_{A^\phi}$-measurable family of functions on $\R_+$, $\xi(s)$ is a nontrivial real element of $\mathcal D$ satisfying $l_{v(s)}\xi(s)=0$ for $\nu_0$-a.e. $s\in A^\phi$, and $\mathcal R(s)=L^{\xi(s)}_{v(s)}$ for $\nu_0$-a.e. $s\in A^\phi$. By~(\ref{lukappa}), (\ref{lu0}), and (\ref{wronskian}), $u^{m+\phi}(0)$ and $w^{m+\phi}(0)$ are real linearly independent solutions of $l_{v(s)}u^{m+\phi}(0)=l_{v(s)}w^{m+\phi}(0)=0$ for every $s=(m,p)\in A^\phi$. Hence, there exist $\nu_0$-a.e. defined real functions $C_1$ and $C_2$ on $A^\phi$ such that
\[
\xi(s) =C_1(s)u^{m+\phi}(0) + C_2(s) w^{m+\phi}(0)
\]
for $\nu_0$-a.e. $s=(m,p)\in A^\phi$. In view of~(\ref{wronskian}), we have
\[
C_1(s) = \frac{\pi}{2}W(\xi(s),w^{m+\phi}(0)),\quad C_2(s) = -\frac{\pi}{2}W(\xi(s),u^{m+\phi}(0))
\]
for $\nu_0$-a.e. $s=(m,p)\in A^\phi$. It follows from statements~3 and~5 of Lemma~\ref{aux_meas} that $C_1$ and $C_2$ are $\nu_0$-measurable functions on $A^\phi$. Since $\xi(s)\neq 0$, we have $C_1(s)^2+C_2(s)^2\neq 0$ for $\nu_0$-a.e. $s\in A^\phi$. Let $\theta$ be a Borel function on $A^\phi$ such that
\[
\theta(s) = \vartheta_{m+\phi}+\tau(C_2(s))\arccos \frac{C_1(s)}{C_1(s)^2+C_2(s)^2}
\]
for $\nu_0$-a.e. $s=(m,p)\in A^\phi$, where $\tau(y)=1$ for $y\geq 0$ and $\tau(y)=-1$ for $y<0$. We then have
\[
\cos(\theta(s)-\vartheta_{m+\phi})= \frac{C_1(s)}{C_1(s)^2+C_2(s)^2},\quad \sin (\theta(s)-\vartheta_{m+\phi})= \frac{C_2(s)}{C_1(s)^2+C_2(s)^2}
\]
for $\nu_0$-a.e. $s=(m,p)\in A^\phi$. This means that $\xi(s)$ is proportional to the function $\xi_\theta(s)$ given by~(\ref{xitheta}) and it follows from~(\ref{mathcalRphi}) that $\mathcal R(s)=h_{m+\phi,\theta(s)}$ for $\nu_0$-a.e. $s=(m,p)\in A^\phi$. Hence, $\tilde{\mathcal H}(s)=\mathcal H^\phi_\theta(s)$ for $\nu_0$-a.e. $s\in A^\phi$. As $\mathcal H^\phi(s)$ is self-adjoint for all $s\in S\setminus A^\phi$, we have $\tilde{\mathcal H}(s)=\mathcal H^\phi(s)=\mathcal H^\phi_\theta(s)$ for $\nu_0$-a.e. $s\in S\setminus A^\phi$. We therefore have $\tilde{\mathcal H}(s)=\mathcal H^\phi_\theta(s)$ for $\nu_0$-a.e.~$s$.
\end{proof}

Given $\phi\in\R$ and a Borel real function $\theta$ on $A^\phi$, we define the operator $R^\phi_\theta$ in $L_2(\R^3)$ by the relation
\begin{equation}\label{301}
R^\phi_{\theta} = V^{-1}\int^\oplus \mathcal H^\phi_{\theta}(s)\,d\nu_0(s)\,\, V,
\end{equation}
where the unitary operator $V\colon L_2(\R^3)\to L_2(S,\mathfrak h,\nu_0)$ is as in Lemma~\ref{l_unitary}.

\begin{proposition}\label{t2}
Let $\phi\in \R$. For any Borel real function $\theta$ on $A^\phi$, $R^\phi_\theta$
is a self-adjoint extension of $H^\phi$ commuting with $T_G$ for all $G\in \mathcal G$. Every self-adjoint extension of $H^\phi$ commuting with $T_G$ for all $G\in \mathcal G$ is equal to $R^\phi_{\theta}$  for some Borel real function $\theta$ on $A^\phi$. Given Borel real functions $\theta$ and $\tilde\theta$ on $A^\phi$, we have $R^\phi_{\theta}=R^\phi_{\tilde\theta}$ if and only if $\theta(s)-\tilde\theta(s)\in\pi\Z$ for $\nu_0$-a.e. $s\in A^\phi$.
\end{proposition}

\begin{proof}
If $\theta$ is a Borel real function on $A^\phi$, then it follows immediately from Corollary~\ref{c_red} and Lemmas~\ref{l_d}, \ref{l100}, and~\ref{l_famsaext} that $R^\phi_{\theta}$ is a self-adjoint extension of $H^\phi$ commuting with $T_G$ for all $G\in \mathcal G$.

Conversely, let $\tilde H$ be a self-adjoint extension of $H^\phi$ commuting with $T_G$ for all $G\in \mathcal G$. By Corollary~\ref{c_red} and Lemmas~\ref{l_d} and~\ref{l100}, there is a $\nu_0$-measurable family $\tilde{\mathcal H}$ of operators in $\mathfrak h$ such that $\tilde{\mathcal H}(s)$ is a self-adjoint extension of $\mathcal H^\phi(s)$ for $\nu_0$-a.e. $s$ and
\begin{equation}\label{tildeHdint}
\tilde H = V^{-1} \int^\oplus \tilde{\mathcal H}(s)\,d\nu_0(s)\,V.
\end{equation}
By Lemma~\ref{l_famsaext}, there is a Borel real function $\theta$ on $A^\phi$ such that $\tilde{\mathcal H}(s)=\mathcal H^\phi_\theta(s)$ for $\nu_0$-a.e. $s$.
In view of~(\ref{301}) and~(\ref{tildeHdint}), it follows that $\tilde H = R^\phi_\theta$.

Let $\theta$ and $\tilde\theta$ be Borel real functions on $A^\phi$. By Proposition~\ref{l_di4b}, the equality $R^\phi_{\theta}=R^\phi_{\tilde\theta}$ holds if and only if $\mathcal H^\phi_{\theta}(s)=\mathcal H^\phi_{\tilde\theta}(s)$ for $\nu_0$-a.e. $s$. By~(\ref{htheta}) and (\ref{htau12}), the latter condition is fulfilled if and only if $u^{m+\phi}_{\theta(s)}(0)$ is proportional to $u^{m+\phi}_{\tilde\theta(s)}(0)$ for $\nu_0$-a.e. $s=(m,p)\in A^\phi$. In view of~(\ref{ukappavartheta}), this is true if and only if $\theta(s)-\tilde\theta(s)\in\pi\Z$ for $\nu_0$-a.e. $s\in A^\phi$.
\end{proof}

\section{Direct integrals of measures}
\label{s9}

Let $\nu$ be a $\sigma$-finite positive measure and $\mu$ be a map such that $\mu(s)$ is a $\sigma$-finite positive measure for $\nu$-a.e. $s$. Suppose $\Sigma$ is a $\sigma$-algebra whose every element is a $\mu(s)$-measurable set for $\nu$-a.e. $s$. Under suitable measurability conditions on $\mu$, we shall define a new measure $M$, called the direct integral of $\mu$ and denoted by $(\Sigma)\mbox{-}\!\int^\oplus\mu(s)\,d\nu(s)$, satisfying $\s_M = \s_\nu\times G$, where $G$ is the largest element of $\Sigma$. This measure has a Fubini-type property (Proposition~\ref{p_fubini}) and can be viewed as a generalization of the direct product of measures. The main result of this section (Proposition~\ref{puo}) states that the space $L_2(\s_M,M)$ can be identified with $\int^\oplus L_2(G,\mu(s))\,d\nu(s)$. We shall see in Sec.~\ref{s11} that the measure $(\mathcal B_\R)\mbox{-}\!\int^\oplus\mu^\phi_\theta(s)\,d\nu_0(s)$, where $\mu^\phi_\theta$ is given by~(\ref{mutheta12}) and $\mathcal B_\R$ is the Borel $\sigma$-algebra on $\R$, satisfies the conditions of Proposition~\ref{prop1} and, therefore, coincides with $M^\phi_\theta$. This will allow us to obtain a direct integral representation of the space $L_2(S\times \R,M^\phi_\theta)$ entering Proposition~\ref{propW} and relate it to the results of Sec.~\ref{s8}.

\begin{definition}\label{d_meas}
Let $\nu$ be a $\sigma$-finite positive measure, $\Sigma$ be a $\sigma$-algebra, and a map $\mu$ be such that $\mu(s)$ is a $\sigma$-finite positive measure for $\nu$-a.e. $s$. Let the $\delta$-ring $\mathcal Q^\Sigma_{\mu,\nu}$ be defined by the relation
\[
\mathcal Q^\Sigma_{\mu,\nu} = \{A\in\Sigma: \mbox{$A\in D_{\mu(s)}$ for $\nu$-a.e. $s$}\}
\]
We say that $\mu$ is a $(\nu,\Sigma)$-measurable family of measures if $\Sigma = \sigma(\mathcal Q^\Sigma_{\mu,\nu})$ and $s\to \mu(s|A)$ is a $\nu$-measurable function for any $A\in \mathcal Q^\Sigma_{\mu,\nu}$.
\end{definition}

\begin{lemma}\label{lmm}
Let $\nu$, $\Sigma$, $\mu$, and $\mathcal Q^\Sigma_{\mu,\nu}$ be as in Definition~$\mathrm{\ref{d_meas}}$.
Let $\mathcal K\subset \mathcal Q^\Sigma_{\mu,\nu}$ be closed under finite intersections and satisfy $\sigma(\mathcal K) = \Sigma$.
Let a map $\xi$ be such that, for every $A\in\mathcal Q^\Sigma_{\mu,\nu}$, $\xi(s)$ is a $\mu(s)$-integrable complex function on $A$ for $\nu$-a.e. $s$.
If $s\to \int_A \xi(s|E)\,d\mu(s|E)$ is a $\nu$-measurable function for any $A\in \mathcal K$, then the same is true for every $A\in \mathcal Q^\Sigma_{\mu,\nu}$.
\end{lemma}

If $\xi(s)$ is $\mu(s)$-equivalent to unity for $\nu$-a.e. $s$, this lemma reduces to the next result.

\begin{corollary}\label{cormm}
Let $\nu$, $\Sigma$, $\mu$, and $\mathcal K$ be as in Lemma~$\mathrm{\ref{lmm}}$. If $s\to \mu(s|A)$ is a $\nu$-measurable function for any $A\in \mathcal K$, then $\mu$ is a $(\nu,\Sigma)$-measurable family of measures.
\end{corollary}

For the proof of Lemma~\ref{lmm}, we need the next definition.

\begin{definition}\label{d_alpha}
We say that a nonempty set of sets $\mathcal Q$ is an $\alpha$-class if it satisfies the following conditions
\begin{enumerate}
\item[(1)] If $A,B\in \mathcal Q$ and $A\cap B=\varnothing$, then $A\cup B\in \mathcal Q$.
\item[(2)] If $A,B\in \mathcal Q$ and $B\subset A$, then $A\setminus B\in \mathcal Q$.
\item[(3)] If $A_1\supset A_2\supset\ldots$ is a nonincreasing sequence of elements of $\mathcal Q$, then $\bigcap_{i=1}^\infty A_i\in \mathcal Q$.
\end{enumerate}
\end{definition}

\begin{remark}
In the context of $\delta$-rings, the role of $\alpha$-classes is essentially the same as that of Dynkin systems (see, e.g., Sec.~2 in~Ref.~\onlinecite{Bauer}) with respect to $\sigma$-algebras. In particular, Lemma~\ref{l_alpha} below is an analogue of Theorem~2.4 in~Ref.~\onlinecite{Bauer}.
\end{remark}

\begin{lemma}\label{l_alpha}
Let $\mathcal Q$ be an $\alpha$-class. If there exists a set $\mathcal K\subset \mathcal Q$ that is closed under finite intersections and satisfies $\mathcal Q\subset \sigma(\mathcal K)$, then $\mathcal Q$ is a $\delta$-ring.
\end{lemma}

\begin{proof}[Proof of Lemma~$\mathrm{\ref{lmm}}$]
It is straightforward to check that the set
\[
\mathcal L = \{A\in \mathcal Q^\Sigma_{\mu,\nu}: \mbox{$s\to \int_A\xi(s|E)\,d\mu(s|E)$ is a $\nu$-measurable function}\}.
\]
is an $\alpha$-class containing $\mathcal K$. As $\sigma(\mathcal K)=\Sigma$, we have $\mathcal L\subset \sigma(\mathcal K)$, and Lemma~\ref{l_alpha} implies that $\mathcal L$ is a $\delta$-ring. As $\sigma(\mathcal L)=\Sigma$, every $A\in \mathcal Q^\Sigma_{\mu,\nu}$ is representable as the union of a sequence $A_1\subset A_2\subset \ldots$ of elements of $\mathcal L$. Hence,
\[
\int_A\xi(s|E)\,d\mu(s|E)=\lim_{k\to\infty}\int_{A_k}\xi(s|E)\,d\mu(s|E)
\]
for $\nu$-a.e. $s$ and, therefore, $s\to\int_A\xi(s|E)\,d\mu(s|E)$ is $\nu$-measurable. Thus, $\mathcal Q^\Sigma_{\mu,\nu}=\mathcal L$ and the lemma is proved.
\end{proof}

\begin{proof}[Proof of Lemma~$\mathrm{\ref{l_alpha}}$]
Let $\mathcal L$ be the smallest $\alpha$-class containing $\mathcal K$.\footnote{It is easy to see that the intersection of any set of $\alpha$-classes is again an $\alpha$-class (note that such an intersection is always nonempty because $\varnothing$ is an element of every $\alpha$-class). Hence, given a set of sets $\mathcal K$, there exists the smallest $\alpha$-class containing $\mathcal K$.} We first prove that $\mathcal L$ is a $\delta$-ring. To this end, it suffices to show that $\mathcal L$ is closed under finite intersections because every $\alpha$-class with this property is a $\delta$-ring. For $A\in \mathcal L$, we set
\[
\mathcal L^A = \{B\in \mathcal L: A\cap B\in \mathcal L\}.
\]
Let $A\in \mathcal L$ and $B,C\in \mathcal L^A$ be such that $B\cap C=\varnothing$. As $(A\cap B)\cap (A\cap C)=\varnothing$ and both sets $A\cap B$ and $A\cap C$ belong to $\mathcal L$, the set $A\cap (B\cup C)=(A\cap B)\cup (A\cap C)$ also belongs to $\mathcal L$ by condition~1 of Definition~\ref{d_alpha}. This means that $B\cup C\in \mathcal L^A$.

If $B,C\in \mathcal L^A$ are such that $C\subset B$, then $A\cap (B\setminus C) = (A\cap B)\setminus (A\cap C)$. Since both $A\cap B$ and $A\cap C$ belong to $\mathcal L$ and $A\cap C\subset A\cap B$, we have $A\cap (B\setminus C)\in \mathcal L$ by condition~2 of Definition~\ref{d_alpha} and, hence, $B\setminus C\in \mathcal L^A$.

Now let $B_1\supset B_2\supset\ldots$ be a nonincreasing sequence of elements of $\mathcal L^A$ and let $B=\bigcap_{i=1}^\infty B_i$. Since $A\cap B = \bigcap_{i=1}^\infty (A\cap B_i)$ and the sets $A\cap B_1,A\cap B_2,\ldots$ constitute a nonincreasing sequence of elements of $\mathcal L$, we have $A\cap B\in \mathcal L$ by condition~3 of Definition~\ref{d_alpha} and, therefore, $B\in \mathcal L^A$.

It follows from the above that $\mathcal L^A$ is an $\alpha$-class for any $A\in \mathcal L$. If $B\in \mathcal K$, then we obviously have $\mathcal K\subset \mathcal L^B$ and, hence, $\mathcal L^B=\mathcal L$. It follows that $A\cap B\in \mathcal L$ for arbitrary $A\in \mathcal L$ and $B\in \mathcal K$ and, therefore, $B\in \mathcal L^A$. This means that $\mathcal K\subset \mathcal L^A$ for every $A\in \mathcal L$. Hence, $\mathcal L^A=\mathcal L$ for every $A\in \mathcal L$ and, consequently, $A\cap B\in \mathcal L$ for every $A,B\in \mathcal L$. Thus, $\mathcal L$ is a $\delta$-ring.

We now show that $\mathcal Q$ is a $\delta$-ring. As for $\mathcal L$, it suffices to show that $A\cap B\in \mathcal Q$ for any $A,B\in \mathcal Q$. As $\mathcal L$ is a $\delta$-ring, every element of $\sigma(\mathcal K)=\sigma(\mathcal L)$ is a countable union of elements of $\mathcal L$. Let $A_1\subset A_2\subset\ldots$ and $B_1\subset B_2\subset\ldots$ be nondecreasing sequences of elements of $\mathcal L$ such that $A=\bigcup_{i=1}^\infty A_i$ and $B=\bigcup_{i=1}^\infty B_i$. As $\mathcal L$ is a ring, $C_i = A_i\cap B_i$ belongs to $\mathcal L$ (and, hence, to $\mathcal Q$) for every $i=1,2,\ldots$. It follows from condition~2 of Definition~\ref{d_alpha} that $C'_i=A\setminus C_i$ belongs to $\mathcal Q$ for every $i=1,2,\ldots$. As $C'_1\supset C'_2\supset\ldots$, condition~3 of Definition~\ref{d_alpha} implies that $A\setminus B=\bigcap_{i=1}^\infty C'_i$ belongs to $\mathcal Q$. Applying condition~2 of Definition~\ref{d_alpha} again, we conclude that $A\cap B = A\setminus (A\setminus B)$ belongs to $\mathcal Q$.
\end{proof}

Given sets of sets $\mathcal Q_1$ and $\mathcal Q_2$, we denote by $\mathcal Q_1\boxtimes \mathcal Q_2$ the set of all sets $A_1\times A_2$, where $A_1\in \mathcal Q_1$ and $A_2\in \mathcal Q_2$. Let $A$ be a set. For any $s$, we define its section $A_s$ by the relation $A_s =\{E: (s,E)\in A\}$.

Let $\nu$ be a $\sigma$-finite positive measure, $\Sigma$ be a $\sigma$-algebra, and $\mu$ be a $(\nu,\Sigma)$-measurable family of measures. We set $\Xi^\Sigma_{\mu,\nu} = \sigma(D_\nu\boxtimes\Sigma)$ and denote by $\Delta^\Sigma_{\mu,\nu}$ the set of all $A\in\Xi^\Sigma_{\mu,\nu}$ such that $A_s\in D_{\mu(s)}$ for $\nu$-a.e. $s$ and $s\to \mu(s|A_s)$ is a $\nu$-integrable function.

\begin{proposition}\label{p_prod}
Let $\nu$ be a $\sigma$-finite positive measure, $\Sigma$ be a $\sigma$-algebra, and $\mu$ be a $(\nu,\Sigma)$-measurable family of measures. Then there is a unique $\Xi^\Sigma_{\mu,\nu}$-compatible measure $M$ such that $\Delta^\Sigma_{\mu,\nu}\subset D_M$ and
\begin{equation}\label{fub}
M(A) = \int \mu(s|A_s)\,d\nu(s)
\end{equation}
for any $A\in \Delta^\Sigma_{\mu,\nu}$. Moreover, the equality $\Xi^\Sigma_{\mu,\nu}\cap D_M =\Delta^\Sigma_{\mu,\nu}$ is fulfilled. If $\nu$ is $\tilde\Sigma$-compatible for some $\sigma$-algebra $\tilde\Sigma$, then $M$ is $\sigma(\tilde\Sigma\boxtimes\Sigma)$-compatible.
\end{proposition}
\begin{proof}
Let $\mathcal Q=\mathcal Q^\Sigma_{\mu,\nu}$, $\Xi=\Xi^\Sigma_{\mu,\nu}$, $\Delta=\Delta^\Sigma_{\mu,\nu}$, and $\mathcal K = \Delta\cap (D_\nu\boxtimes\mathcal Q)$. Since $\mathcal Q$ is a $\delta$-ring, $\mathcal K$ is closed under finite intersections. Given $B\in D_\nu$, $C\in \mathcal Q$, and $n=1,2,\ldots$, we set $B_n = \{s\in B: \mu(s|C)\leq n\}$. Clearly, $B = N\cup\bigcup_n B_n$, where $N$ is a $\nu$-null set. Since $B_n\times C\in \mathcal K$ for all $n$ and $N\times C\in \mathcal K$, we conclude that $B\times C\in \sigma(\mathcal K)$. Thus, $D_\nu\boxtimes\mathcal Q\subset \sigma(\mathcal K)$. As $\sigma(\mathcal Q)=\Sigma$, it follows that $\sigma(\mathcal K)=\Xi$. As  $\Delta$ is obviously an $\alpha$-class, Lemma~\ref{l_alpha} implies that $\Delta$ is a $\delta$-ring. Let the function $m\colon \Delta\to \R$ be such that $m(A)$ is equal to the right-hand side of~(\ref{fub}) for any $A\in\Delta$. It follows from the monotone convergence theorem that $m$ is a $\sigma$-additive function satisfying condition~(a) of Sec.~\ref{measures_sec}. Hence, the existence and uniqueness of the measure $M$, as well as the equality $\Xi\cap D_M =\Delta$, are ensured by Lemma~\ref{compl}.
Now suppose $\nu$ is $\tilde\Sigma$-compatible and let $\mathcal L$ denote the set of all $A\in\Xi$ such that there exist $B_1,B_2\in\sigma(\tilde \Sigma\boxtimes\Sigma)$ satisfying the conditions $B_1\subset A\subset B_2$ and $M(B_2\setminus B_1)=0$. Then $\mathcal L$ is a $\sigma$-algebra that contains $D_\nu\boxtimes\Sigma$ and, hence, coincides with $\Xi$. In view of the $\Xi$-compatibility of $M$, this implies that $M$ is $\sigma(\tilde \Sigma\boxtimes\Sigma)$-compatible.
\end{proof}

\begin{definition}\label{d_dirm}
Under the conditions of Proposition~\ref{p_prod}, we call the measure $M$ the direct integral of $\mu$ with respect to $(\nu,\Sigma)$ and denote it by $(\Sigma)\mbox{-}\!\int^\oplus\mu(s)\,d\nu(s)$.
\end{definition}

\begin{lemma}\label{cor_prod}
Let $\nu$ be a $\sigma$-finite positive measure, $\Sigma$ be a $\sigma$-algebra, $\mu$ be a $(\nu,\Sigma)$-measurable family of measures, and $M=(\Sigma)\mbox{-}\!\int^\oplus\mu(s)\,d\nu(s)$. An $M$-measurable set $N$ is an $M$-null set if and only if $N_s$ is a $\mu(s)$-null set for $\nu$-a.e. $s$.
\end{lemma}
\begin{proof}
Let $\Xi=\Xi^\Sigma_{\mu,\nu}$ and $N$ be an $M$-measurable set. By the $\Xi$-compatibility of $M$, there is a set $N'\in \Xi$ such that $N'\supset N$ and $N'\setminus N$ is an $M$-null set. If $N$ is an $M$-null set, then $N'$ is also an $M$-null set, and it follows from Proposition~\ref{p_prod} that $N'_s\in D_{\mu(s)}$ for $\nu$-a.e. $s$ and $\int \mu(s|N_s)\,d\nu(s) = 0$, i.e., $N'_s$ is a $\mu(s)$-null set for $\nu$-a.e. $s$. Since $N_s\subset N'_s$, this implies that $N_s$ is a $\mu(s)$-null set for $\nu$-a.e. $s$. Conversely, suppose $N_s$ is a $\mu(s)$-null set for $\nu$-a.e. $s$. By the above $(N'\setminus N)_s$ is a $\mu(s)$-null set and, hence, $N'_s$ is a $\mu(s)$-null set for $\nu$-a.e. $s$. Proposition~\ref{p_prod} now implies that $M(N')=0$ and, therefore, $N$ is an $M$-null set.
\end{proof}

\begin{proposition}\label{p_fubini}
Let $\nu$, $\Sigma$, $\mu$, and $M$ be as in Lemma~$\mathrm{\ref{cor_prod}}$. For any $M$-integrable complex function $f$, the function $E\to f(s,E)$ is $\mu(s)$-integrable for $\nu$-a.e. $s$ and we have
\begin{equation}\label{fubini}
\int f(s,E)\,dM(s,E) = \int d\nu(s)\int f(s,E)\,d\mu(s|E).
\end{equation}
\end{proposition}
\begin{proof}
Let $\Xi=\Xi^\Sigma_{\mu,\nu}$. As $M$ is $\Xi$-compatible, every $A\in D_M$ is representable in the form $A=B\cup N$, where $B\in D_M\cap \Xi$ and $N$ is an $M$-null set. By Proposition~\ref{p_prod}, we have $B\in\Delta^\Sigma_{\mu,\nu}$. Hence, $B_s\in D_{\mu(s)}$ for $\nu$-a.e. $s$, the function $s\to \mu(s|B_s)$ is $\nu$-integrable, and $M(B) = \int \mu(s|B_s)\,d\nu(s)$. Since $A_s = B_s\cup N_s$, it follows from Lemma~\ref{cor_prod} that $A_s\in D_{\mu(s)}$ and $\mu(s|A_s) = \mu(s|B_s)$ for $\nu$-a.e. $s$. Hence, the function $s\to \mu(s|A_s)$ is $\nu$-integrable, and $M(A) = \int \mu(s|A_s)\,d\nu(s)$. This proves the proposition for $f=X_A$, where $X_A$ is the function on $\s_M$ that is equal to unity on $A$ and vanishes on $\s_M\setminus A$. We say that a function $f$ on $\s_M$ is simple if it can be represented in the form $f = \sum_{i=1}^n c_i X_{A_i}$ with $c_i\in\C$ and $A_i\in D_M$. Clearly, the proposition is true if $f$ is a simple function. If $f$ is a nonnegative $M$-integrable function, then there is a nondecreasing sequence $f_k$ of nonnegative simple functions such that $f_k(s,E)\leq f(s,E)$ and $f_k(s,E)\to f(s,E)$ as $k\to\infty$ for $M$-a.e. $(s,E)$. Let $g_k$ be a $\nu$-integrable function such that $g_k(s) = \int f_k(s,E)\,d\mu(s|E)$ for $\nu$-a.e. $s$. The dominated convergence theorem implies that
\[
\int f(s,E) \,dM(s,E) = \lim_{k\to\infty} \int f_k(s,E) \,dM(s,E) = \lim_{k\to\infty} \int g_k(s) \,d\nu(s).
\]
By the monotone convergence theorem, there exists a $\nu$-integrable function $g$ such that $g(s)=\lim_{k\to\infty} g_k(s)$ for $\nu$-a.e. $s$, and we have
\begin{equation}\label{limint}
\int f(s,E) \,dM(s,E) = \int g(s)\,d\nu(s).
\end{equation}
By Lemma~\ref{cor_prod}, for $\nu$-a.e. $s$, the relations $f_k(s,E)\leq f(s,E)$ and $f_k(s,E)\to f(s,E)$ hold for $\mu(s)$-a.e. $E$. In view of the existence of $\lim_{k\to\infty}g_k(s)$, the monotone convergence theorem implies that the function $E\to f(s,E)$ is $\mu(s)$-integrable and $\int f(s,E)\,d\mu(s|E)=g(s)$ for $\nu$-a.e. $s$. Substituting this equality in~(\ref{limint}) yields~(\ref{fubini}). To complete the proof, it remains to note that every $M$-integrable function is a linear combination of nonnegative $M$-integrable functions.
\end{proof}

\begin{corollary}\label{cor_fubini}
Let $\nu$, $\Sigma$, $\mu$, and $M$ be as in Lemma~$\mathrm{\ref{cor_prod}}$ and $f$ be an $M$-measurable complex function. Then $E\to f(s,E)$ is a $\mu(s)$-measurable complex function for $\nu$-a.e. $s$. If $E\to f(s,E)$ is $\mu(s)$-integrable for $\nu$-a.e. $s$, then $s\to\int f(s,E)\,d\mu(s|E)$ is a $\nu$-measurable function. If, in addition, $f$ is nonnegative and $s\to\int f(s,E)\,d\mu(s|E)$ is a $\nu$-integrable function, then $f$ is $M$-integrable.
\end{corollary}
\begin{proof}
Let $A_1\subset A_2\subset\ldots$ be a sequence of elements of $D_M$ such that $\s_M = \bigcup_{k=1}^\infty A_k$. Let $B_k = A_k\cap\{(s,E)\in \s_M: |f(s,E)|\leq k\}$ and $f_k =X_{B_k}f$, where $X_{B_k}$ is as in the proof of Proposition~\ref{p_fubini}. Then $\s_M$ coincides with $\bigcup_{k=1}^\infty B_k$ up to an $M$-null set and, therefore, $f_k(s,E)\to f(s,E)$ as $k\to\infty$ for $M$-a.e. $(s,E)$. Since $f_k$ are $M$-integrable for all $k$, Lemma~\ref{cor_prod} and Proposition~\ref{p_fubini} imply that, for $\nu$-a.e. $s$, the functions $E\to f_k(s,E)$ are $\mu(s)$-integrable and converge $\mu(s)$-a.e. to the function $E\to f(s,E)$. Hence, the latter is $\mu(s)$-measurable for $\nu$-a.e. $s$. If it is also $\mu(s)$-integrable for $\nu$-a.e. $s$, then it follows from the dominated convergence theorem that $\int f(s,E)\,d\mu(s|E) = \lim_{k\to\infty} \int f_k(s,E)\,d\mu(s|E)$ for $\nu$-a.e. $s$. Since $s\to \int f_k(s,E)\,d\mu(s|E)$ are $\nu$-measurable functions by Proposition~\ref{p_fubini}, we conclude that $s\to \int f(s,E)\,d\mu(s|E)$ is $\nu$-measurable. If $f$ is nonnegative and the function $s\to\int f(s,E)\,d\mu(s|E)$ is $\nu$-integrable, then Proposition~\ref{p_fubini} implies that
\[
\int f_k(s,E)\,dM(s,E) \leq \int d\nu(s)\int f(s,E)\,d\mu(s|E)
\]
for all $k$. The $M$-integrability of $f$ therefore follows from the monotone convergence theorem.
\end{proof}

\begin{lemma}\label{l_vanish}
Let $\mathcal K$ be a set of sets closed under finite intersections and $\nu$ be a positive $\sigma$-finite $\sigma(\mathcal K)$-compatible measure. Suppose $f$ is a $\nu$-measurable complex function that is $\nu$-integrable on every set in $\mathcal K$. If $\int_A f(s)\,d\nu(s) = 0$ for every $A\in \mathcal K$, then $f(s) = 0$ for $\nu$-a.e. $s$.
\end{lemma}
\begin{proof}
If $f$ satisfies the conditions of the lemma, then the same is true for its real and imaginary parts. So we can assume that $f$ is real. Let
\[
\mathcal Q = \left\{ A\in \sigma(\mathcal K): \mbox{$f$ is $\nu$-integrable on $A$ and $\int_A f(s)\,d\nu(s) = 0$}\right\}.
\]
Clearly, $\mathcal Q$ is an $\alpha$-class satisfying $\mathcal K\subset \mathcal Q\subset \sigma(\mathcal K)$ and, therefore, is a $\delta$-ring by Lemma~\ref{l_alpha}. As $\nu$ is $\sigma(\mathcal K)$-compatible and $f$ is real, every $A\in \mathcal Q$ can be represented in the form $A=A_+\cup A_-$, where $A_\pm\in \sigma(\mathcal K)$ are such that $f(s)\geq 0$ for $\nu$-a.e. $s\in A_+$ and $f(s)\leq 0$ for $\nu$-a.e. $s\in A_-$. Since $A_\pm$ are elements of $\sigma(\mathcal Q)$ contained in an element of $\mathcal Q$, we conclude that $A_\pm\in \mathcal Q$ and, therefore, $\int_{A_\pm} f(s)\,d\nu(s) = 0$. This implies that $f(s) = 0$ for $\nu$-a.e. $s\in A$. Since $\s_\nu$ is, up to a $\nu$-null set, a countable union of elements of $\mathcal Q$, we have $f(s)=0$ for $\nu$-a.e. $s$.
\end{proof}

\begin{proposition}\label{propmm}
Let $\nu$ be a $\sigma$-finite positive measure, $\Sigma$ be a countably generated $\sigma$-algebra, and $\mu$ be a $(\nu,\Sigma)$-measurable family of measures such that $\mu(s)$ is $\Sigma$-compatible for $\nu$-a.e. $s$.
Then there is a unique (up to $\nu$-identity) $\nu$-measurable family $\mathfrak S$ of Hilbert spaces satisfying the conditions:
\begin{enumerate}
\item[$(1)$] $\mathfrak S(s) = L_2(G,\mu(s))$ for $\nu$-a.e. $s$, where $G$ is the largest set in $\Sigma$.
\item[$(2)$] A $\nu$-a.e. defined section $\xi$ of $\mathfrak S$ is $\nu$-measurable if and only if \\$s\to \int_A \xi(s|E)\,d\mu(s|E)$ is a $\nu$-measurable function for any $A\in \mathcal Q_{\mu,\nu}^\Sigma$.
\end{enumerate}
Let $\mathcal K\subset \mathcal Q^\Sigma_{\mu,\nu}$ be closed under finite intersections and satisfy $\sigma(\mathcal K) = \Sigma$. If $\xi$ is a $\nu$-a.e. defined section of $\mathfrak S$ such that $s\to \int_A \xi(s|E)\,d\mu(s|E)$ is a $\nu$-measurable function for any $A\in \mathcal K$, then $\xi$ is $\nu$-measurable.
\end{proposition}
\begin{proof}
Let $\mathcal L_0$ be a countable set generating $\Sigma$. Without loss of generality, we can assume that $\mathcal L_0$ is closed under finite intersections. Let $B_1\subset B_2\subset \ldots$ be a sequence of elements of $\mathcal Q^\Sigma_{\mu,\nu}$ such that $G=\bigcup_{i=1}^\infty B_i$ and let $\mathcal L$ be the set of all sets of the form $A\cap B_i$ for some $A\in \mathcal L_0$ and $i=1,2,\ldots$. Then $\mathcal L$ is a countable subset of $\mathcal Q^\Sigma_{\mu,\nu}$ that is closed under finite intersections and satisfies $\Sigma =\sigma(\mathcal L)$. Let $\mathfrak S$ be $\nu$-a.e. defined family of Hilbert spaces satisfying~(1). We choose a numbering $A_1,A_2,\ldots$ of the set $\mathcal L$ and endow $\mathfrak S$ with a $\nu$-measurable structure by setting
\[
\xi_i^{\mathfrak S}(s) = [\chi_{A_i}]_{\mu(s)},\quad i=1,2,\ldots,
\]
for $\nu$-a.e. $s$, where the function $\chi_{A_i}$ is equal to unity on $A_i$ and vanishes on $G\setminus A_i$. Since $\langle \xi_i^{\mathfrak S}(s), \xi_j^{\mathfrak S}(s)\rangle = \mu(s|A_i\cap A_j)$ and $\mathcal Q^\Sigma_{\mu,\nu}$ is a ring, $s\to \langle \xi_i^{\mathfrak S}(s), \xi_j^{\mathfrak S}(s)\rangle$ is a $\nu$-measurable function for any $i,j=1,2,\ldots$. If $\psi\in \mathfrak S(s)$ is orthogonal to $\xi_i^{\mathfrak S}(s)$ for all $i$, then Lemma~\ref{l_vanish} implies that $[\psi]_{\mu(s)} = 0$. This means that the linear span of the sequence $\xi_1^{\mathfrak S}(s),\xi_2^{\mathfrak S}(s),\ldots$ is dense in $\mathfrak S(s)$ for $\nu$-a.e. $s$. Thus, $\mathfrak S$ endowed with the sequence $\xi_1^{\mathfrak S},\xi_2^{\mathfrak S},\ldots$ is indeed a $\nu$-measurable family of Hilbert spaces. If $\xi$ is a $\nu$-measurable section of $\mathfrak S$, then $s\to \langle\xi(s),[\chi_A]_{\mu(s)}\rangle=\int_A\xi(s|E)\,d\mu(s|E)$ is a $\nu$-measurable function for any $A\in \mathcal L$, and Lemma~\ref{lmm} implies that this is also true for every $A\in \mathcal Q^\Sigma_{\mu,\nu}$. Thus, $\mathfrak S$ satisfies~(2). Suppose now that $\mathcal K\subset \mathcal Q^\Sigma_{\mu,\nu}$ is closed under finite intersections and satisfies $\sigma(\mathcal K) = \Sigma$. By Lemma~\ref{lmm}, if $\xi$ is a $\nu$-a.e. defined section of $\mathfrak S$ such that $s\to \int_A \xi(s|E)\,d\mu(s|E)$ is a $\nu$-measurable function for any $A\in \mathcal K$, then this is also true for any $A\in \mathcal Q^\Sigma_{\mu,\nu}$ and it follows from~(2) that $\xi$ is $\nu$-measurable.
\end{proof}

\begin{proposition}\label{puo}
Let $\nu$, $\mu$, $\Sigma$, and $\mathfrak S$ be as in Proposition~$\mathrm{\ref{propmm}}$. Let $\mathfrak H = \int^\oplus \mathfrak S(s)\,d\nu(s)$ and $M = (\Sigma)\mbox{-}\!\int^\oplus\mu(s)\,d\nu(s)$. Given an $M$-measurable function $f$, let $\hat f$ denote the $\nu$-equivalence class such that $\hat f(s)$ is the $\mu(s)$-equivalence class of the map $E\to f(s,E)$ for $\nu$-a.e. $s$. Then the following statements hold:
\begin{enumerate}
\item[$1.$] $M$-measurable functions $f_1$ and $f_2$ are $M$-equivalent if and only if $\hat f_1=\hat f_2$.

\item[$2.$] If $f\in L_2(\s_M,M)$, then $\hat f \in \mathfrak H$ and the operator $Q\colon L_2(\s_M,M)\to \mathfrak H$ taking $f$ to $\hat f$ is unitary.

\item[$3.$] If $g$ is an $M$-measurable complex function, then $\hat g(s)$ is $\mu(s)$-measurable for $\nu$-a.e. $s$, $s\to \mathcal T^{\mu(s)}_{\hat g(s)}$ is a $\nu$-measurable family of operators in $\mathfrak S$, and
\[
Q\mathcal T^M_g Q^{-1} = \int^\oplus \mathcal T^{\mu(s)}_{\hat g(s)}\,d\nu(s).
\]
\end{enumerate}
\end{proposition}
\begin{proof} 1. Since $f_1$ and $f_2$ are $M$-measurable, the set $N = \{(s,E)\in \s_M: f_1(s,E)\neq f_2(s,E)\}$ is $M$-measurable. By Lemma~\ref{cor_prod}, $f_1$ and $f_2$ are $M$-equivalent if and only if $N_s$ is a $\mu(s)$-null set for $\nu$-a.e. $s$. Clearly, this condition holds if and only if $\hat f_1(s) = \hat f_2(s)$ for $\nu$-a.e. $s$ and, hence, $\hat f_1 = \hat f_2$.

\par\medskip\noindent 2.
Let $f\in L_2(\s_M,M)$. By Corollary~\ref{cor_fubini}, $\hat f(s)$ is a $\mu(s)$-measurable function for $\nu$-a.e. $s$. Since $|f|^2$ is $M$-integrable, Proposition~\ref{p_fubini} implies that $|\hat f(s)|^2$ is $\mu(s)$-integrable and, therefore, $\hat f(s)\in \mathfrak S(s)$ for $\nu$-a.e. $s$. For every $A\in \mathcal Q^\Sigma_{\mu,\nu}$, the function $\chi_A\hat f(s)$, where $\chi_A$ is as in the proof of Proposition~\ref{propmm}, is a product of two $\mu(s)$-square-integrable functions and, hence, is $\mu(s)$-integrable for $\nu$-a.e. $s$. Applying Corollary~\ref{cor_fubini} to the function $(s,E)\to \chi_A(E)f(s,E)$, we conclude that $s\to \int_A \hat f(s|E)\,d\mu(s|E)$ is a $\nu$-measurable function for any $A\in \mathcal Q^\Sigma_{\mu,\nu}$. In view of Proposition~\ref{propmm}, this means that $\hat f$ is a $\nu$-measurable section of $\mathfrak S$. By Proposition~\ref{p_fubini}, $s\to \|\hat f(s)\|^2$ is a $\nu$-integrable function and we have
$\|f\|^2 = \int \|\hat f(s)\|^2\,d\nu(s)$. This means that $\hat f\in \mathfrak H$ and the operator $Q$ is isometric. We now prove that $Q$ is unitary. For this, it suffices to show that its image is dense in $\mathfrak H$. In other words, we have to show that every $\xi\in \mathfrak H$ that is orthogonal to $\hat f$ for every $f\in L_2(\s_M,M)$ is equal to zero. Given $A\in \mathcal Q^\Sigma_{\mu,\nu}$, let $\mathcal K_A$ denote the set of all $B\in D_\nu$ such that $B\times A\in D_M$. It is clear that $\mathcal K_A$ is closed under finite intersections and $\sigma(\mathcal K_A) = \sigma(D_\nu)$. For a set $C\subset \s_M$, let $X_C$ denote the function on $\s_M$ that is equal to unity on $C$ and vanishes on $\s_M\setminus C$. For $A\in \mathcal Q^\Sigma_{\mu,\nu}$ and $B\in \mathcal K_A$, the function $X_{B\times A}$ is obviously $M$-square-integrable and, therefore, we have
\[
\int_B \langle[\chi_{A}]_{\mu(s)},\xi(s)\rangle\,d\nu(s) = \int \langle\hat X_{B\times A}(s),\xi(s)\rangle\,d\nu(s) = \langle \hat X_{B\times A},\xi\rangle = 0.
\]
Applying Lemma~\ref{l_vanish} to $\mathcal K = \mathcal K_A$, we conclude that $\langle[\chi_{A}]_{\mu(s)},\xi(s)\rangle = 0$ for $\nu$-a.e. $s$ for every $A \in \mathcal Q^\Sigma_{\mu,\nu}$. Acting as in the proof of Proposition~\ref{propmm}, we choose a sequence $A_1,A_2,\ldots$ of elements of $\mathcal Q^\Sigma_{\mu,\nu}$ such that the linear span of $[\chi_{A_i}]_{\mu(s)}$ is dense in $\mathfrak S(s)$ for $\nu$-a.e. $s$. Then, for $\nu$-a.e. $s$, $\xi(s)$ is orthogonal to all vectors $[\chi_{A_i}]_{\mu(s)}$ and, hence, is equal to zero. Thus, $\xi=0$ and the unitarity of $Q$ is proved.

\par\medskip\noindent 3.
By Corollary~\ref{cor_fubini}, $\hat g(s)$ is $\mu(s)$-measurable for $\nu$-a.e. $s$. Let $C_1\subset C_2\subset\ldots$ be a sequence of elements of $D_M$ such that $\s_M = \bigcup_{j=1}^\infty C_j$ and $g$ is $M$-essentially bounded on $C_j$ for every $j=1,2,\ldots$. For $i,j=1,2,\ldots$, we define the function $h_{ij}$ on $\s_M$ by setting $h_{ij}(s,E) = \chi_{A_i}(E)X_{C_j}(s,E)$, where $A_i$ are as in the proof of~(2). It is clear that $h_{ij}$ is $M$-square-integrable and, therefore, $\hat h_{ij}$ is a $\nu$-measurable section of $\mathfrak S$ for all $i,j$. To prove the $\nu$-measurability of the family $s\to \mathcal T^{\mu(s)}_{\hat g(s)}$, it suffices to show that $\hat h_{ij}(s)$ belongs to the domain of $\mathcal T^{\mu(s)}_{\hat g(s)}$ and the linear span of the vectors $(\hat h_{ij}(s),\mathcal T^{\mu(s)}_{\hat g(s)}\hat h_{ij}(s))$ is dense in the graph of $\mathcal T^{\mu(s)}_{\hat g(s)}$ for $\nu$-a.e. $s$. For $j=1,2,\ldots$, let $P_j$ be a $\nu$-a.e. defined map such that $P_j(s) = \mathcal T^{\mu(s)}_{\chi_{C_{j,s}}}$ for $\nu$-a.e. $s$, where $C_{j,s}=\{E: (s,E)\in C_j\}$. Then, for $\nu$-a.e. $s$, $P_j(s)$ is an orthogonal projection commuting with $\mathcal T^{\mu(s)}_{\hat g(s)}$ and satisfying the equality
\begin{equation}\label{gij}
\hat h_{ij}(s) = P_j(s) [\chi_{A_i}(s)]_{\mu(s)}
\end{equation}
for all $i,j$. In view of Lemma~\ref{cor_prod}, $\hat g(s)$ is $\mu(s)$-essentially bounded on $C_{j,s}$ for $\nu$-a.e. $s$. Hence $\mathrm{Im}\,P_j(s)$ is contained in the domain of $\mathcal T^{\mu(s)}_{\hat g(s)}$ and $\mathcal T^{\mu(s)}_{\hat g(s)}P_j(s)$ is a bounded operator for $\nu$-a.e. $s$. In particular, it follows from~(\ref{gij}) that $\hat h_{ij}(s)$ is in the domain of $\mathcal T^{\mu(s)}_{\hat g(s)}$ for $\nu$-a.e. $s$. Let $\mathcal G_s$ be the subset of the graph of $\mathcal T^{\mu(s)}_{\hat g(s)}$ consisting of all its elements $(\psi,\tilde\psi)$ such that $\psi,\tilde\psi\in \mathrm{Im}\,P_j(s)$ for some $j=1,2,\ldots$. For $\nu$-a.e. $s$, we have $\lim_{j\to\infty} P_j(s)\psi=\psi$ for every $\psi\in \mathfrak S(s)$. As $P_j(s)$ commute with $\mathcal T^{\mu(s)}_{\hat g(s)}$, this implies that $\mathcal G_s$ is dense in the graph of $\mathcal T^{\mu(s)}_{\hat g(s)}$ for $\nu$-a.e. $s$. Given $(\psi,\tilde\psi)\in \mathcal G_s$ and $\varepsilon>0$, we can find a finite linear combination $\tau$ of vectors $[\chi_{A_i}]_{\mu(s)}$ such that $\|\psi-\tau\|<\varepsilon$. Then we have
\begin{align}
&\|\psi - P_j(s)\tau\|=\|P_j(s)(\psi-\tau)\|<\varepsilon,\nonumber\\
&\left\|\tilde\psi - \mathcal T^{\mu(s)}_{\hat g(s)} P_j(s)\tau\right\|=\left\|\mathcal T^{\mu(s)}_{\hat g(s)} P_j(s)(\psi-\tau)\right\|< \left\|\mathcal T^{\mu(s)}_{\hat g(s)} P_j(s)\right\|\varepsilon,\nonumber
\end{align}
where $j$ is such that $\psi,\tilde\psi\in \mathrm{Im}\,P_j(s)$.
In view of~(\ref{gij}), $P_j(s)\tau$ is a linear combination of $\hat h_{ij}(s)$, and it follows from the above inequalities that the linear span of the vectors $(\hat h_{ij}(s),\mathcal T^{\mu(s)}_{\hat g(s)}\hat h_{ij}(s))$ is dense in $\mathcal G_s$ and, hence, in the graph of $\mathcal T^{\mu(s)}_{\hat g(s)}$ for $\nu$-a.e. $s$. The $\nu$-measurability of the family $s\to \mathcal T^{\mu(s)}_{\hat g(s)}$ is thus proved. Set $T = \int^\oplus \mathcal T^{\mu(s)}_{\hat g(s)}\,d\nu(s)$. Let $f$ belong to the domain of $\mathcal T^M_g$ and $F = \mathcal T^M_g f$. In view of Lemma~\ref{cor_prod}, $\hat F(s)$ is $\mu(s)$-equivalent to $\hat g(s)\hat f(s)$ for $\nu$-a.e. $s$. This means that $Qf=\hat f\in D_T$ and $Q\mathcal T^M_g f = TQf$. It follows that $T$ is an extension of $Q\mathcal T^M_g Q^{-1}$. To finish the proof, we have to show that $Q^{-1}\xi$ belongs to the domain of $\mathcal T^M_g$ for any $\xi\in D_T$. Let $\varphi = g Q^{-1}\xi$. For $\nu$-a.e. $s$, we have $(T\xi)(s|E) = \hat g(s|E)\xi(s|E) = \varphi(s,E)$ for $\mu(s)$-a.e. $E$. As $T\xi\in \mathfrak H$, it follows that the function $E\to \varphi(s,E)$ is $\mu(s)$-square-integrable for $\nu$-a.e. $s$ and the function $s\to \int |\varphi(s,E)|^2 \,d\mu(s|E)$ is $\nu$-integrable. In view of Corollary~\ref{cor_fubini}, this implies that $\varphi$ is $M$-square-integrable, i.e., $Q^{-1}\xi$ is in the domain of $\mathcal T^M_g$.
\end{proof}

\section{Eigenfunction expansions}
\label{s11}

In this section, we prove the results formulated in Sec.~\ref{sec_results}.

Let $\mathcal B_\R$ denote the Borel $\sigma$-algebra on $\R$. Proposition~\ref{prop1} follows immediately from the next lemma.
\begin{lemma}\label{lmutheta12}
Let $\phi\in\R$ and $\theta$ be a Borel real function on $A^\phi$. Then $\mu^\phi_\theta$ is a $(\nu_0,\mathcal B_\R)$-measurable family of measures and $M=(\mathcal B_\R)\mbox{-}\!\int^\oplus\mu^\phi_\theta(s)\,d\nu_0(s)$ is the unique Borel measure on $S\times\R$ satisfying the conditions of Proposition~$\mathrm{\ref{prop1}}$.
\end{lemma}
\begin{proof}
Let $\mathcal K$ be the set of compact subsets of $\R$. Clearly, $\mathcal K$ is contained in the domain of $\mu^\phi_\theta(s)$ for all $s\in S$. Hence, the $(\nu_0,\mathcal B_\R)$-measurability of $\mu^\phi_\theta$ is ensured by Lemma~\ref{c_borel} and Corollary~\ref{cormm} (note that $\sigma(\mathcal K)=\mathcal B_\R$).

Let $\mathcal B_S$ and $\mathcal B_{S\times\R}$ denote the Borel $\sigma$-algebras of $S$ and $S\times \R$ respectively. By Definition~\ref{d_dirm}, $M$ satisfies the conditions of Proposition~\ref{p_prod} for $\mu=\mu^\phi_\theta$, $\nu=\nu_0$, $\Sigma = \mathcal B_\R$, and $\tilde\Sigma = \mathcal B_S$. By Lemma~\ref{c_borel}, we have $K'\times K\in \Delta^\Sigma_{\mu,\nu}$ for any compact sets $K'\subset S$ and $K\subset \R$. Since $\mathcal B_{S\times \R} = \sigma(\mathcal B_S\boxtimes \mathcal B_\R)$, Proposition~\ref{p_prod} implies that $M$ is a Borel measure on $S\times \R$ satisfying~(\ref{M(K'timesK)}). If $f$ is an $M$-measurable function, then equality~(\ref{intf(s,E)}) is ensured by Proposition~\ref{p_fubini}. Thus, $M$ satisfies the conditions of Proposition~\ref{prop1}.

Suppose $\tilde M$ is another Borel measure on $S$ satisfying~(\ref{M(K'timesK)}). Let $\mathcal L$ be the set of all sets $A\in \mathcal B_{S\times \R}\cap D_M\cap D_{\tilde M}$ such that $M(A)=\tilde M(A)$. Let $\tilde{\mathcal K}$ be the set of all sets $K'\times K$, where $K'\subset S$ and $K\subset \R$ are compact sets. Then $\tilde{\mathcal K}$ is closed under finite intersections and $\mathcal L$ is an $\alpha$-class containing $\tilde{\mathcal K}$. Since $\mathcal L\subset \sigma(\tilde{\mathcal K})=\mathcal B_{S\times\R}$, Lemma~\ref{l_alpha} implies that $\mathcal L$ is a $\delta$-ring. We hence have $M=\tilde M$ by Lemma~\ref{compl}.
\end{proof}

Let $\phi\in \R$ and $\theta$ be a Borel real function on $A^\phi$. We define the operator-valued map $\mathcal U^\phi_{\theta}$ on $S$ by setting
\begin{equation}\label{Utheta12}
\mathcal U^\phi_{\theta}(s) = \left\{
\begin{matrix}
U_{m+\phi},& s\in S\setminus A^\phi,\\
U_{m+\phi,\theta(s)},& s\in A^\phi,
\end{matrix}
\right.
\end{equation}
for all $s=(m,p)\in S$, where the operators $U_\kappa$ and $U_{\kappa,\vartheta}$ satisfy the conditions of Proposition~\ref{leig1} and Proposition~\ref{leig2} respectively. By Lemma~\ref{lmutheta12}, the assumptions of Proposition~\ref{propmm} are fulfilled for $\Sigma=\mathcal B_\R$, $\mu=\mu^\phi_{\theta}$, and $\nu=\nu_0$. A $\nu_0$-measurable family $\mathfrak S$ of Hilbert spaces satisfying the conditions of Proposition~\ref{propmm} for such $\Sigma$, $\mu$, and $\nu$ will be denoted by $\mathfrak S^\phi_{\theta}$.

\begin{lemma}\label{l_utheta}
Let $\phi\in \R$, $\theta$ be a Borel real function on $A^\phi$, and $\psi\in L_2^c(\R_+)$. Then $s\to \mathcal U^\phi_{\theta}(s)\,\psi$ is a $\nu_0$-measurable section of $\mathfrak S^\phi_{\theta}$.
\end{lemma}
\begin{proof}
As follows from Proposition~\ref{propmm} applied to the set $\mathcal K$ of all compact subsets of $\R$, it suffices to show that
\begin{equation}\nonumber
I_{K}(s)= \int_{K} (\mathcal U^\phi_{\theta}(s)\psi)(E)\,d \mu^\phi_{\theta}(s|E)
\end{equation}
is a Borel function on $S$ for any compact set $K\subset \R$. Given $-1<\kappa<1$, let the function $f^\kappa_K$ on $\R$ be defined by the relation
\[
f^\kappa_K(\vartheta) = \cos(\vartheta-\vartheta_\kappa)\int_K F^{(1)}_\kappa(E)\,d\mathcal V_{\kappa,\vartheta}(E) + \sin(\vartheta-\vartheta_\kappa) \int_K F^{(2)}_\kappa(E)\,d\mathcal V_{\kappa,\vartheta}(E),
\]
where $\vartheta_\kappa$ is defined by~(\ref{varthetakappa}) and the continuous functions $F^{(1)}_\kappa$ and $F^{(2)}_\kappa$ on $\R$ are given by
\[
F^{(1)}_\kappa(E) = \int_0^\infty u^\kappa(E|r)\psi(r)\,dr,\quad F^{(2)}_\kappa(E) = \int_0^\infty w^\kappa(E|r)\psi(r).
\]
It follows from~(\ref{mutheta12}), (\ref{ukappavartheta}), (\ref{Ukappatheta}), and~(\ref{Utheta12}) that $I_K(s) = f^{m+\phi}_K(\theta(s))$ for every $s=(m,p)\in A^\phi$. On the other hand, $p\to I_K(m,p)$ is a constant function on $\R$ for every $m$ satisfying $|m+\phi|\geq 1$. Since $f^\kappa_K$ is a Borel function on $\R$ by Lemma~\ref{l_borel},
we conclude that $I_K$ is a Borel function on $S$.
\end{proof}

In what follows, we set $\mathfrak h=L_2(\R_+)$.

\begin{corollary}\label{c_utheta}
Let $\phi\in \R$ and $\theta$ be a Borel real function on $A^\phi$. Then $\mathcal U^\phi_{\theta}$ is a $\nu_0$-measurable family of operators from $\mathcal I_{\,\mathfrak h,\nu_0}$ to $\mathfrak S^\phi_{\theta}$.
\end{corollary}
\begin{proof}
The statement follows immediately from Lemma~\ref{l_utheta} and Lemma~\ref{l_BBB}.
\end{proof}

Given $\phi\in \R$ and a Borel real function $\theta$ on $A^\phi$, we define the linear operator $U^\phi_{\theta}$ from $L_2(S,\mathfrak h,\nu_0)$ to $\int^\oplus \mathfrak S^\phi_{\theta}(s)\,d\nu_0(s)$ by setting
\begin{equation}\label{Uphi}
U^\phi_{\theta} = \int^\oplus \mathcal U^\phi_{\theta}(s)\,d\nu_0(s).
\end{equation}
Since $\mathcal U^\phi_{\theta}(s)$ is a unitary operator from $\mathfrak h$ to $\mathfrak S^\phi_{\theta}(s)$ for $\nu_0$-a.e. $s$, it follows from Corollary~\ref{c_utheta} and Proposition~\ref{p_di2} that $U^\phi_{\theta}$ is a unitary operator.
If $f$ is an $M^\phi_{\theta}$-measurable complex function, we denote by $f|^\phi_{\theta}$ the $\nu_0$-equivalence class such that $f|^\phi_{\theta}(s)$ is the $\mu^\phi_{\theta}(s)$-equivalence class of the map $E\to f(s,E)$ for $\nu_0$-a.e. $s$. In view of Lemma~\ref{lmutheta12}, all conditions of Proposition~\ref{puo} are fulfilled for $\Sigma=\mathcal B_\R$, $\nu=\nu_0$, $\mu=\mu^\phi_{\theta}$, $M=M^\phi_{\theta}$, $\mathfrak S = \mathfrak S^\phi_{\theta}$, and $\hat f = f|^\phi_{\theta}$. It follows from statement~2 of Proposition~\ref{puo} that the correspondence $f\to f|^\phi_{\theta}$ induces a unitary operator $Q^\phi_{\theta}$ from $L_2(S\times\R,M^\phi_{\theta})$ to $\int^\oplus \mathfrak S^\phi_{\theta}(s)\,d\nu_0(s)$.

Note that the uniqueness of $W$ in Proposition~\ref{propW} is ensured by the density of $L_2^c(\R^3)$ in $L_2(R^3)$. Hence, Proposition~\ref{propW} follows from the next lemma.

\begin{lemma}\label{lemmaW}
Let $\phi\in \R$, $\theta$ be a Borel real function on $A^\phi$, and $V$ be as in Lemma~$\mathrm{\ref{l_unitary}}$. Then the operator $W = (Q^\phi_{\theta})^{-1} U^\phi_{\theta} V$ satisfies the conditions of Proposition~$\mathrm{\ref{propW}}$.
\end{lemma}
\begin{proof}
Let $\Phi\in C_0^\infty(\mathscr O)$ and $\Psi = [\Phi]_\Lambda$. For $s\in S$ and $E\in\R$, let $h_\Psi(s,E)$ denote the right-hand side of~(\ref{WWtheta12}). For $-1<\kappa<1$ and $m\in\Z$, we define the continuous function $F_m^\kappa$ on $\R\times\R\times(\mathscr O)$ by the relation
\[
F_m^\kappa(E,\vartheta,x) = \frac{e^{-ipx_3}}{2\pi\sqrt{r_x}}\left(\frac{x_1-ix_2}{r_x}\right)^{m} u^{\kappa}_{\vartheta}(E|r_x)
\]
for all $E,\vartheta\in\R$ and $x=(x_1,x_2,x_3)\in \mathscr O$, where $r_x=\sqrt{x_1^2+x_2^2}$. Set $f^\kappa_{m,\Psi}(E,\vartheta) = \int F^\kappa_m(E,\vartheta,x)\Psi(x)\,dx$.
By~(\ref{Wtheta12}), (\ref{Jphitheta}), and~(\ref{WWtheta12}), we have $h_\Psi(s,E)=f^{m+\phi}_{m,\Psi}(E,\theta(s))$ for all $s=(m,p)\in A^\phi$ and $E\in\R$. If $|m+\phi|>1$, then (\ref{Wtheta12}), (\ref{Jphitheta}), and (\ref{WWtheta12}) imply that $(p,E)\to h_\Psi(m,p\,;E)$ is a continuous function on $\R^2$. Since $f^\kappa_{m,\Psi}$ is a continuous function on $\R^2$ for all $m\in \Z$ and $-1<\kappa<1$, it follows that $(p,E)\to h_\Psi(m,p\,;E)$ is a Borel function on $\R^2$ for every $m\in\Z$. Thus, $h_\Psi$ is a Borel (and, hence, $M^\phi_{\theta}$-measurable) function on $S\times\R$.

In view of~(\ref{Wtheta12}), passing to the polar coordinates in the $(x_1,x_2)$-plane in the integral in~(\ref{WWtheta12}) yields
\begin{equation}\label{h_Psi}
h_\Psi(s,E) = \int_0^\infty \mathcal J^\phi_\theta(s,E|r) \tilde\Phi(s|r)\,dr
\end{equation}
for all $s\in S$ and $E\in\R$, where the $C_0^\infty(\R_+)$-valued map $\tilde\Phi$ on $S$ is given by~(\ref{tilde}). It follows from~(\ref{Jphitheta}), (\ref{mutheta12}), (\ref{Utheta12}), Proposition~\ref{leig1}, and Proposition~\ref{leig2} that, for every $s\in S$ and $\psi\in L_2^c(\R_+)$, the equality
\begin{equation}\label{(s)psi}
(\mathcal U^\phi_{\theta}(s)\psi)(E)=\int_0^\infty \mathcal J^\phi_\theta(s,E|r) \psi(r)\,dr
\end{equation}
holds for $\mu^\phi_{\theta}(s)$-a.e. $E$. In view of~(\ref{def_v}), equalities~(\ref{Uphi}), (\ref{h_Psi}) and~(\ref{(s)psi}) imply that the function $E\to h_\Psi(s,E)$ is $\mu^\phi_{\theta}(s)$-equivalent to $(U^\phi_{\theta}V\Psi)(s)$ for $\nu_0$-a.e. $s$, whence $U^\phi_{\theta}V\Psi = h_\Psi|^\phi_{\theta}$. Since $U^\phi_{\theta}V\Psi = Q^\phi_{\theta}W\Psi = (W\Psi)|^\phi_{\theta}$, it follows from statement~1 of Proposition~\ref{puo} that $W\Psi$ is $M^\phi_{\theta}$-equivalent to $h_\Psi$. Thus, (\ref{WWtheta12}) holds for $\Psi = [\Phi]_\Lambda$ with $\Phi\in C_0^\infty(\mathscr O)$.

For a general $\Psi\in L_2^c(\R^3)$, we choose a sequence $\Phi_n\in C_0^\infty(\mathscr O)$ such that $\Psi_n = [\Phi_n]_\Lambda$ converge to $\Psi$ in $L_2(\R^3)$ and the supports of all $\Phi_n$ are contained in a fixed compact set for all $n$. Since $\mathcal W^\phi_{\theta}(s,E)$ is locally square-integrable, we have $h_{\Psi_n}(s,E)\to h_\Psi(s,E)$ for all $s\in S$ and $E\in\R$. As $W\Psi_n\to W\Psi$ in $L_2(S\times\R,M^\phi_{\theta})$ by continuity of $W$, it follows that $W\Psi$ is $M^\phi_{\theta}$-equivalent to $h_\Psi$.
\end{proof}

Lemma~\ref{lemmaW} and the uniqueness statement of Proposition~\ref{propW} imply that
\begin{equation}\label{Wthetaphi12}
W^\phi_{\theta} = (Q^\phi_{\theta})^{-1} U^\phi_{\theta} V.
\end{equation}

Theorem~\ref{main_theorem} follows from Proposition~\ref{t2} and the next statement.

\begin{proposition}\label{p_final}
Let $\phi\in\R$, $\theta$ be a Borel real function on $A^\phi$, and $R^\phi_{\theta}$ be defined by~$(\ref{301})$. Then we have $H^\phi_{\theta} = R^\phi_{\theta}$.
\end{proposition}
\begin{proof}
For brevity, we set $\mu = \mu^\phi_{\theta}$, $\mathcal U = \mathcal U^\phi_{\theta}$, $U = U^\phi_{\theta}$, $Q=Q^\phi_{\theta}$, $M = M^\phi_{\theta}$, and $\mathfrak g= \mathfrak f|^\phi_{\theta}$. It follows from~(\ref{mutheta12}), (\ref{mathfrakf}), (\ref{mathcalHphi}), (\ref{htau12}), (\ref{Utheta12}), and Propositions~\ref{leig1} and~\ref{leig2} that
\[
\mathcal U(s)\mathcal H^\phi_{\theta}(s)\mathcal U(s)^{-1}=\mathcal T^{\mu(s)}_{\mathfrak  g(s)}
\]
for $\nu_0$-a.e. $s$. By Lemma~\ref{l_famsaext}, $\mathcal H^\phi_{\theta}$ is a $\nu_0$-measurable family of operators in $\mathfrak h$. Hence, it follows from~(\ref{301}), (\ref{Uphi}), Corollary~\ref{c_utheta}, statement~1 of Proposition~\ref{p_di2}, and Proposition~\ref{p_di1} that $s\to \mathcal T^{\mu(s)}_{\mathfrak  g(s)}$ is a $\nu_0$-measurable family of operators in $\mathfrak S^\phi_{\theta}$ and
\[
UV R^\phi_{\theta} (UV)^{-1} = \int^\oplus \mathcal T^{\mu(s)}_{\mathfrak  g(s)}\,d\nu_0(s).
\]
By statement~3 of Proposition~\ref{puo}, the operator in the right-hand side is equal to $Q\mathcal T^{M}_{\mathfrak f} Q^{-1}$. Equalities~(\ref{Htheta12}) and~(\ref{Wthetaphi12}) therefore imply that $H^\phi_{\theta} = R^\phi_{\theta}$.
\end{proof}

\section*{Acknowledgments}
The author is grateful to I.V.~Tyutin and B.L.~Voronov for useful discussions.

\appendix
\section{Measure theory}
\label{app0}

In this Appendix, we briefly describe the measure-theoretic concepts used in this paper. Our aim here is not only to fix the notation but also to formulate definitions in such a way as to provide a unified treatment of positive and spectral measures. In particular, in contrast to the standard approach (see, e.g., Refs.~\onlinecite{Bauer,Cohn}), we define positive measures as functions on $\delta$-rings rather than $\sigma$-algebras and do not allow them to take infinite values (see Sec.~\ref{measures_sec} below).

\subsection{Rings and algebras of sets}
Recall that a nonempty set of sets $\mathcal Q$ is called a ring of sets if $A\cup B\in \mathcal Q$ and $A\setminus B\in \mathcal Q$ for any $A,B\in \mathcal Q$. As $A\cap B = A\setminus(A\setminus B)$, every ring is closed under finite intersections. A ring $\mathcal Q$ is called a $\sigma$-ring (a $\delta$-ring) if it is closed under countable unions (resp., under countable intersections). For a set of sets $\mathcal Q$, we denote by $\sigma(\mathcal Q)$ ($\delta(\mathcal Q)$) the $\sigma$-ring (resp., $\delta$-ring) generated by $\mathcal Q$, i.e., the smallest $\sigma$-ring (resp., $\delta$-ring) containing $\mathcal Q$.
If $\mathcal Q$ is a $\delta$-ring, then $\sigma(\mathcal Q)$ is just the set of all countable unions of elements of $\mathcal Q$.
A $\sigma$-ring $\mathcal Q$ is said to be countably generated if there is a countable set $\tilde{\mathcal Q}\subset \mathcal Q$ such that $\mathcal Q = \sigma(\tilde{\mathcal Q})$.
A $\sigma$-ring $\mathcal Q$ is called a $\sigma$-algebra if it has the largest element with respect to inclusion.

The $\sigma$-ring generated by all open sets of a topological space $X$ is obviously a $\sigma$-algebra. It is called the Borel $\sigma$-algebra of $X$ and its elements are called Borel subsets of $X$.  A map $f$ is called a Borel map from $X$ to a
topological space $Y$ if $X$ is contained in the domain $D_f$ of $f$, $f(X)\subset Y$, and  $f^{-1}(A)\cap X$ is a Borel subset of $X$ for any Borel subset $A$ of $Y$.

\subsection{Measures}
\label{measures_sec}
Let $\mathfrak A$ be a topological Abelian group and $\nu$ be an $\mathfrak A$-valued map. Let $\mathcal N_\nu$ denote the set of all $N\in D_\nu$ with the property: if $N'\subset N$ and $N'\in D_\nu$, then $\nu(N') = 0$. We say that $\nu$ is an $\mathfrak A$-valued $\sigma$-additive function if $\nu(A) = \sum_{i\in I} \nu(A_i)$ for any $A\in D_\nu$ and any countable partition $A = \bigcup_{i\in I} A_i$ with $A_i\in D_\nu$. An $\mathfrak A$-valued $\sigma$-additive function is called an $\mathfrak A$-valued measure if its domain $D_\nu$ is a $\delta$-ring and the following completeness conditions are satisfied:
\begin{itemize}
\item[(a)] If $A\in \sigma(D_\nu)$
and the family $\{\nu(A_i)\}_{i\in I}$ is summable in $\mathfrak A$ for any countable partition $A = \bigcup_{i\in I} A_i$ with $A_i\in D_\nu$, then $A\in D_\nu$.

\item[(b)] If $N\in \mathcal N_\nu$ and $N'\subset N$, then $N'\in D_\nu$ (and, hence, $N'\in \mathcal N_\nu$).
\end{itemize}

Elements of $\mathcal N_\nu$ are called $\nu$-null sets and elements of $\sigma(D_\nu)$ are called $\nu$-measur\-able sets.
A measure $\nu$ is called $\sigma$-finite if $\sigma(D_\nu)$ is a $\sigma$-algebra. In this case, we denote by $\s_\nu$ the largest element of $\sigma(D_\nu)$ with respect to inclusion. A measure $\nu$ is called finite if it is $\sigma$-finite and $\s_\nu\in D_\nu$. A measure $\nu$ is called positive if it is $\R$-valued and $\nu(A)\geq 0$ for any $A\in D_\nu$.

\begin{remark}
\label{rem1}
In most expositions of measure theory (see, e.g., Refs~\onlinecite{Bauer,Cohn}), a positive measure is defined as a $\sigma$-additive function on a $\sigma$-algebra taking values in the extended real semi-axis $\bar \R_+=[0,\infty]$. Such an $\bar\R_+$-valued measure $\nu$ is called $\sigma$-finite if the largest set in $D_\nu$ is a countable union of elements of $D_\nu$ with a finite measure. It is called complete if every subset of an element of $D_\nu$ with zero measure also belongs to $D_\nu$.  Complete $\sigma$-finite $\bar\R_+$-valued measures can be naturally identified with positive $\sigma$-finite measures in our sense (note that a complete $\sigma$-finite $\bar\R_+$-valued measure $\nu$ restricted to the set of all elements of $D_\nu$ with a finite measure is a positive $\sigma$-finite measure in our sense). The advantage of our definition is that it does not involve the extended real axis and, therefore, makes it possible to treat positive measures in the same way as vector-valued measures.
\end{remark}

Given a $\sigma$-ring $\mathcal Q$, a measure $\nu$ is called $\mathcal Q$-compatible if $\mathcal Q\subset \sigma(D_\nu)$ and for every $A\in D_\nu$, there is an $N\in \mathcal N_\nu$ such that $A\cup N\in \mathcal Q$. In this case, $D_\nu$ consists of all sets of the form $B\cup N$, where $B\in \mathcal Q\cap D_\nu$ and $N\in \mathcal N_\nu$. If $X$ is a topological space and $\mathcal B$ is its Borel $\sigma$-algebra, then a $\mathcal B$-compatible measure is called a Borel measure on $X$.

\begin{remark}
Positive Borel measures on a topological space $X$ are usually defined as $\bar\R_+$-valued $\sigma$-additive functions on the Borel $\sigma$-algebra of $X$. If such a measure is $\sigma$-finite, then its completion corresponds (as described in Remark~\ref{rem1}) to a positive Borel measure on $X$ in our sense. Since we consider only complete measures (see condition~(b) above), we introduce the notion of $\mathcal Q$-compatibility to relate measure and topology.
\end{remark}

\begin{lemma}\label{compl}
Let $\mathfrak A$ be a sequentially complete topological Abelian group and $\nu$ be an $\mathfrak A$-valued $\sigma$-additive function such that $D_\nu$ is a $\delta$-ring. Then there is a unique $\sigma(D_\nu)$-compatible $\mathfrak A$-valued measure $\hat\nu$ such that $D_\nu\subset D_{\hat \nu}$ and $\hat\nu$ coincides with $\nu$ on $D_\nu$. If $\nu$ satisfies~$(\mathrm a)$, then $D_\nu = D_{\hat\nu}\cap \sigma(D_\nu)$.
\end{lemma}

The measure $\hat \nu$ is called the completion of the $\sigma$-additive function $\nu$.

\begin{proof}
We give the proof only for positive $\nu$ because the lemma is used in this paper only in this case. The existence and uniqueness of $\hat\nu$ is then guaranteed by the well-known results on the extension of positive measures. Let $\nu$ satisfy~(a) and $A\in D_{\hat\nu}\cap\sigma(D_\nu)$. Let $\{A_i\}_{i\in I}$ be a countable partition of $A$ such that $A_i\in D_\nu$ for all $i\in I$. Since $\hat\nu(A_i)=\nu(A_i)$ for all $i\in I$, the $\sigma$-additivity of $\hat\nu$ implies that the family $\{\nu(A_i)\}_{i\in I}$ is summable. Condition~(a) hence ensures that $A\in D_\nu$ and, therefore, $D_\nu \supset D_{\hat\nu}\cap \sigma(D_\nu)$. As the opposite inclusion also obviously holds, the lemma is proved.
\end{proof}

Let $\nu$ be an $\mathfrak A$-valued measure and $A$ be a $\nu$-measurable set. The restriction $\nu|_A$ of $\nu$ to $A$ is, by definition, the restriction of the map $\nu$ to the domain $D_{\nu|_A}$ consisting of all elements of $D_\nu$ that are contained in $A$. Clearly, $\nu|_A$ is a $\sigma$-finite measure for any $\nu$-measurable set $A$ and $\s_{\nu|_A} = A$.

Let $\nu_1$ and $\nu_2$ be positive measures and $\mathcal Q$ be the $\sigma$-ring generated by all sets of the form $A_1\times A_2$, where $A_1\in D_{\nu_1}$ and $A_2\in D_{\nu_2}$. Then there is a unique positive $\mathcal Q$-compatible measure $\nu$ such that $\nu(A_1\times A_2)=\nu_1(A_1)\nu_2(A_2)$ for every $A_1\in D_{\nu_1}$ and $A_2\in D_{\nu_2}$. This measure is called the product of $\nu_1$ and $\nu_2$ and is denoted by $\nu_1\times\nu_2$.

\subsection{Standard measures}
\label{s_standard}

Let $\mathfrak A$ be a topological Abelian group.
An $\mathfrak A$-valued measure $\nu$ is called standard if it is $\sigma$-finite and there exists a complete separable metric space $X$ such that $X\subset \s_\nu$, $\s_\nu\setminus X$ is a $\nu$-null set, and $\nu|_X$ is a Borel measure on $X$.
If $\nu$ is a standard measure and $A$ is a $\nu$-measurable set, then $\nu|_A$ is also a standard measure.

\begin{remark}
Standard measures were first introduced in~Ref.~\onlinecite{Mackey}. Unlike Ref.~\onlinecite{Mackey}, we consider only complete measures, and standard measures in our sense are actually the completions of those in the sense of~Ref.~\onlinecite{Mackey}. In the probabilistic context, such complete measures were introduced in~Ref.~\onlinecite{Rohlin} under the name of Lebesgue measure spaces. The class of standard measures is broad enough to include most measures encountered in applications and, at the same time, is narrow enough to exclude measures with a pathological behavior.
\end{remark}

\subsection{Measurable maps}

Let $\nu$ be a $\sigma$-finite $\mathfrak A$-valued measure. A map $f$ is said to be defined $\nu$-a.e. if $\s_\nu\setminus D_f$ is a $\nu$-null set. Given a set $X$, a map $f$ is said to be an $X$-valued $\nu$-a.e. defined map if $\s_\nu\setminus f^{-1}(X)$ is a $\nu$-null set. If $X$ is a topological space, then a map $f$ is called an $X$-valued $\nu$-measurable map if $f$ is an $X$-valued $\nu$-a.e. defined map and $f^{-1}(B)\cap \s_\nu$ is a $\nu$-measurable set for any Borel subset $B$ of $X$. We say that some property $P(s)$ holds for $\nu$-a.e. $s$, if there is a $\nu$-null set $N$ such that $P(s)$ holds for all $s\in \s_\nu\setminus N$.

Given a $\nu$-measurable set $A$, we say that $f$ is an $X$-valued $\nu$-measurable map on $A$ if $f$ is an $X$-valued $\nu|_A$-measurable map. A property $P(s)$ is said to hold for $\nu$-a.e. $s\in A$ if it holds for $\nu|_A$-a.e. $s$.

A complex $\nu$-a.e. defined function $f$ is said to be $\nu$-essentially bounded if there is $C>0$ such that $|f(s)|\leq C$ for $\nu$-a.e. $s$.

\begin{remark}
\label{rem2}
Note that the domain $D_f$ of a $\nu$-a.e. defined or $\nu$-measurable map $f$ is not assumed to be contained in $\s_\nu$. Moreover, $X$-valued $\nu$-a.e. defined (and, in particular, $\nu$-measurable) maps are allowed to take values outside $X$ on some $\nu$-null set contained in $\s_\nu\cap D_f$. This implies in particular, that using ``the set of all $X$-valued $\nu$-measurable maps'' would lead to the same kind of set-theoretic problems as the use of ``the set of all sets''.
\end{remark}

Let $X$ be a complete separable metric space, $f$ be an $X$-valued $\nu$-a.e. defined map, and $f_1,f_2,\ldots$ be a sequence of $X$-valued $\nu$-measurable maps such that $f_n(s)$ converge to $f(s)$ in $X$ as $n\to\infty$ for $\nu$-a.e. $s$. Then $f$ is $\nu$-measurable.

\subsection{Equivalence classes}

Let $\nu$ be a $\sigma$-finite $\mathfrak A$-valued measure. Two $\nu$-a.e. defined maps $f$ and $g$ are called $\nu$-equivalent if $f(s)=g(s)$ for $\nu$-a.e. $s$. All $\nu$-a.e. defined maps fall into disjoint classes of $\nu$-equivalent maps, which are called $\nu$-equivalence classes. In every $\nu$-equivalence class, we choose an arbitrary fixed element whose domain of definition is contained in $\s_\nu$. Given a $\nu$-a.e. defined map $f$, we denote by $[f]_\nu$ such a chosen element belonging to the $\nu$-equivalence class containing $f$. Thus, the map $[f]_\nu$ is $\nu$-equivalent to $f$ for any $\nu$-a.e. defined map $f$, and we have $[f]_\nu=[g]_\nu$ for every pair of $\nu$-equivalent maps $f$ and $g$. If $X$ is a topological space, then $f$ is an $X$-valued $\nu$-measurable map if and only if so is $[f]_\nu$.

In this paper, $\nu$-equivalence classes \textit{per se} are not used. Whenever we speak of $\nu$-equivalence classes, we always refer to representatives of the form $[f]_\nu$, where $f$ is a $\nu$-a.e. defined map. There are two reasons for such a redefinition of the notion of a $\nu$-equivalence class. First, the same arguments as in Remark~\ref{rem2} show that the $\nu$-equivalence classes in the true sense cannot be considered as well-defined sets. Hence, using them as elements of sets is not satisfactory from the viewpoint of foundations of mathematics. Choosing a fixed representative in each class allows us to circumvent this difficulty. Second, since $\nu$-equivalence classes in our sense are just some maps, all definitions and notations introduced for maps become directly applicable to $\nu$-equivalence classes.

\subsection{Integrable functions}
Let $\nu$ be a $\sigma$-finite positive measure. A $\nu$-measurable complex function $f$ is called $\nu$-integrable if
\begin{equation}\nonumber
I_\nu(f)=\sup_{A_i\in D_\nu,\,c_i\in\R} \sum_{i=1}^n c_i\nu(A_i)<\infty,
\end{equation}
where the supremum is taken over all finite sets $A_1,\ldots,A_n$ of disjoint elements of $D_\nu$ and $c_1,\ldots,c_n\in\R$ such that $c_i\leq |f(s)|$ for $\nu$-a.e. $s\in A_i$. The integral $\int f(s)\,d\nu(s)$ of a $\nu$-integrable function $f$ is a complex number that is uniquely determined by the conditions that $\int f(s)\,d\nu(s)$ be linear in $f$ and coincide with $I_\nu(f)$ if $f(s)\geq 0$ for $\nu$-a.e. $s$.

Clearly, passing to a $\nu$-equivalent function does not affect its $\nu$-integrability. In particular, a $\nu$-a.e. defined function $f$ is $\nu$-integrable if and only if $[f]_\nu$ is $\nu$-integrable.
If $A$ is a $\nu$-measurable set and $f$ is a $\nu|_A$-integrable function, then we say that $f$ is a $\nu$-integrable function on $A$. In this case, we write $\int_A f(s)\,d\nu(s)$ in place of $\int f(s)\,d\nu|_A(s)$.

\begin{remark}
As for $\nu$-measurable functions, the domain of a $\nu$-integrable function is not assumed to be contained in $\s_\nu$.
\end{remark}

\subsection{$L_2$-spaces}\label{s_l2}

Let $\nu$ be a $\sigma$-finite positive measure and $\mathfrak h$ be a separable Hilbert space.\footnote{Note that a Hilbert space is separable if it is either finite-dimensional or infinite-dimensional with a countable orthonormal basis. In particular, we consider $\C$ as a one-dimensional Hilbert space with the scalar product ${\langle}\alpha,\beta{\rangle} =\bar \alpha\beta$,
$\alpha,\beta\in \C$.}
We denote by $\mathcal M(\mathfrak h,\nu)$ the set of all elements $[f]_\nu$, where $f$ is an $\mathfrak H$-valued $\nu$-measurable map. The set $\mathcal M(\mathfrak h,\nu)$ has a natural structure of a vector space over $\mathbb C$ (for any $\mathfrak h$-valued $\nu$-a.e. defined maps $f$ and $g$ and any $k\in\mathbb C$, we set $[f]_{\nu} + [g]_{\nu} = [f+g]_{\nu}$ and $k[f]_{\nu} = [kf]_{\nu}$).  Given a $\nu$-measurable set $A$, we denote by $L_2(A,\mathfrak h,\nu)$ the linear subspace of $\mathcal M(\mathfrak h,\nu|_A)$ consisting of all its elements $f$ such that $s\to \|f(s)\|^2$ is a $\nu$-integrable function. Defining the scalar product of $f,g\in L_2(A,\mathfrak h,\nu)$ by the formula
\[
\langle f,g\rangle = \int_A \langle f(s),g(s)\rangle\,d\nu(s),
\]
we make $L_2(A,\mathfrak h,\nu)$ into a Hilbert space.

For $\mathfrak h=\C$, we denote the space $L_2(A,\mathfrak h,\nu)$ by $L_2(A,\nu)$. If $A\subset \R^n$ and $\nu$ is the Lebesgue measure on $\R^n$, the space $L_2(A,\nu)$ is denoted by $L_2(A)$.

Given a $\nu$-measurable complex function $g$, we denote by $\mathcal T^{\nu}_g$ the operator of multiplication by $g$ in $L_2(\s_\nu,\nu)$. By definition, the graph of $\mathcal T^\nu_g$ consists of all pairs $(f_1,f_2)$ of elements of $L_2(\s_\nu,\nu)$ such that $f_2(s) = g(s)f_1(s)$ for $\nu$-a.e. $s$. The operator $\mathcal T^{\nu}_g$ is closed and densely defined and
its adjoint is equal to $\mathcal T^{\nu}_{\bar g}$, where $\bar g$ is the complex conjugate function of $g$. In particular, if $g$ is real, then $\mathcal T^{\nu}_{g}$ is self-adjoint.

\subsection{Spectral measures}
\label{s_spec_meas}
We refer the reader to Chapter~5 of~Ref.~\onlinecite{BirmanSolomjak} for a detailed exposition of the theory of spectral measures. Here, we only give a brief summary of the facts needed in this paper. Let $\mathfrak H$ be a separable Hilbert space and $L(\mathfrak H)$ be the space of bounded everywhere defined linear operators in $\mathfrak H$ endowed with the strong operator topology.
A finite $L(\mathfrak H)$-valued measure $\mathcal E$ is called a spectral measure in $\mathfrak H$ if $\mathcal E(A)$ is an orthogonal projection in $\mathfrak H$ for any $A\in D_{\mathcal E}$ and $\mathcal E(\s_{\mathcal E})$ is the identity operator in $\mathfrak H$.
For any $\Psi\in \mathfrak H$,
we define the positive measure $\mathcal E_\Psi$ as the completion of the positive $\sigma$-additive function $A\to \langle \mathcal E(A)\Psi,\Psi\rangle$ on $D_{\mathcal E}$, where $\langle
\cdot,\cdot\rangle$ is the scalar product on~$\mathfrak H$.

Let $\mathcal E$ be a spectral measure. Given an $\mathcal E$-measurable complex function $f$, the integral $J^{\mathcal E}_f$ of $f$ with
respect to $\mathcal E$ is defined as the unique linear operator in $\mathfrak H$ such that
\begin{align}
&D_{J^{\mathcal E}_f}=\left\{\Psi\in \mathfrak H : \int |f(s)|^2\,d\mathcal E_\Psi(s)<\infty \right\} \label{dom} \\
& \langle\Psi, J^{\mathcal E}_f \Psi\rangle = \int f(s)\,d\mathcal E_\Psi(s),\quad \Psi\in D_{J^{\mathcal E}_f}.\label{spectral_integral}
\end{align}
For any $\mathcal E$-measurable complex function $f$, the operator $J^{\mathcal E}_f$ is closed and densely defined and its adjoint is equal to $J^{\mathcal E}_{\bar f}$.

For every normal operator $T$, there is a unique Borel spectral measure $\mathcal E_T$ on $\C$ such that $J^{\mathcal E_T}_{g}=T$, where
$g$ is the identity function on $\C$. The operators $\mathcal E_T(A)$, where $A$ is a Borel subset of $\C$, are called the spectral
projections of $T$. If $f$ is an $\mathcal E_T$-measurable complex function, then the operator $J_f^{\mathcal E_T}$ is also denoted as $f(T)$.

\section{Direct integral decompositions of operators in Hilbert space}
\label{app1}

One of the main mathematical tools used in this paper are the direct integral decompositions of operators in Hilbert space. In its original form~\cite{Neumann}, the theory of such decompositions (also known as von~Neumann's reduction theory) is applicable only to bounded operators\footnote{A detailed account of the reduction theory for bounded operators can be found in~Ref.~\onlinecite{Dixmier}.} and, therefore, is insufficient for quantum-mechanical applications, where unbounded operators play a prominent role. An extension of the reduction theory to unbounded operators was given in~Ref.~\onlinecite{Nussbaum} (most results of~Ref.~\onlinecite{Nussbaum} were earlier formulated without proofs in~Ref.~\onlinecite{Barriere1951}). The approach proposed in~Ref.~\onlinecite{Nussbaum} is based on the observation~\cite{Stone1951} that the properties of a closed operator in a Hilbert space can be encoded in a set of four bounded operators forming its so called characteristic matrix. This allows one to reformulate problems concerning the direct integrals of arbitrary closed operators in terms of direct integrals of bounded operators. Although this approach makes it possible to extend the results of the original von~Neumann's theory to unbounded operators, it seems to be somewhat inconvenient for applications to concrete problems, where the computation of the characteristic matrices may turn out to be a difficult task. In particular, the definition of the measurability of families of operators given in~Ref.~\onlinecite{Nussbaum} may be difficult to verify for concrete examples. From the viewpoint of applications, the definition of measurability proposed in~Ref.~\onlinecite{Barriere1951} and formulated directly in terms of the graphs of operators seems to be more suitable (both definitions are actually equivalent, see Remark~\ref{rem_measurability} below).

In this appendix, we give a self-contained exposition of the theory of direct integral decompositions of (generally, unbounded) operators in Hilbert space. Our treatment is based on the formulation of measurability proposed in~Ref.~\onlinecite{Barriere1951} and does not involve characteristic matrices. Such an approach allows us to derive all results in a straightforward manner: unlike Ref.~\onlinecite{Nussbaum}, we do not need the theory for bounded operators as a prerequisite. Moreover, it becomes possible to give a concise treatment of direct integral decompositions of sums and products of operators (see Propositions~\ref{p_di1} and~\ref{p_di+} below), whose original analysis in~Ref.~\onlinecite{Lennon1974} involves lengthy computations of characteristic matrices.

In what follows, we use the measure-theoretic framework described in Appendix~\ref{app0}. Throughout this appendix, $\nu$ denotes a $\sigma$-finite positive measure.

\subsection{Measurable families of Hilbert spaces}
\label{app1_1}
A $\nu$-a.e. defined map $\mathfrak S$ is called a $\nu$-a.e. defined family of Hilbert spaces if $\mathfrak S(s)$ is a separable Hilbert space for $\nu$-a.e. $s$. A map $\xi$ is called a $\nu$-a.e. defined section of $\mathfrak S$ if $\xi(s)\in \mathfrak S(s)$ for $\nu$-a.e. $s$.

A $\nu$-a.e. defined family of Hilbert spaces $\mathfrak S$ is called $\nu$-measurable if it is equipped with a sequence $\{\xi_i^{\mathfrak S}\}_{i=1}^\infty$ of $\nu$-a.e. defined sections of $\mathfrak S$ such that the linear span of the vectors $\xi_1^{\mathfrak S}(s),\xi_2^{\mathfrak S}(s),\ldots$ is dense in $\mathfrak S(s)$ for $\nu$-a.e. $s$, and $s\to \langle \xi_i^{\mathfrak S}(s),\xi_j^{\mathfrak S}(s)\rangle$ is a $\nu$-measurable complex function for any $i,j=1,2,\ldots$. A $\nu$-a.e. defined section $\xi$ of $\mathfrak S$ is called $\nu$-measurable if $s\to \langle \xi_i^{\mathfrak S}(s),\xi(s)\rangle$ is a $\nu$-measurable complex function for any $i=1,2,\ldots$.

We say that $\nu$-measurable families $\mathfrak S_1$ and $\mathfrak S_2$ of Hilbert spaces are $\nu$-identical if they are $\nu$-equivalent and have the same $\nu$-measurable sections.

A sequence $e_1,e_2,\ldots$ of elements of a Hilbert space $\mathfrak H$ is said to be a generalized orthonormal system in $\mathfrak H$ if $\langle e_i,e_j\rangle=0$ for $i\neq j$ and $\|e_i\|$ is equal to either $1$ or $0$ for any $i=1,2,\ldots$.

Let $\mathfrak S$ be a $\nu$-a.e. defined family of Hilbert spaces and $\xi_1,\xi_2,\ldots$ be a sequence of $\nu$-a.e. defined sections of $\mathfrak S$. We say that $\nu$-a.e. defined sections $\tilde\xi_1,\tilde\xi_2,\ldots$ of $\mathfrak S$ constitute an orthonormal sequence associated with $\xi_1,\xi_2,\ldots$ if the following conditions hold
\begin{enumerate}
\item[(a)] For $\nu$-a.e. $s$, the sequence $\tilde \xi_1(s),\tilde\xi_2(s),\ldots$ is a generalized orthonormal system in $\mathfrak S(s)$ whose linear span coincides with that of $\xi_1(s),\xi_2(s),\ldots$.

\item[(b)] There exist $\nu$-measurable complex functions $f_{ij}$ defined for $i\leq j$ such that $\tilde\xi_j(s)=\sum_{i=1}^j f_{ij}(s)\xi_i(s)$
for all $j=1,2,\ldots$ and $\nu$-a.e. $s$.
\end{enumerate}

\begin{lemma}\label{l_orth}
Let $\mathfrak S$ be a $\nu$-a.e. defined family of Hilbert spaces and $\xi_1,\xi_2,\ldots$ be a sequence of $\nu$-a.e. defined sections of $\mathfrak S$ such that $s\to\langle \xi_i(s),\xi_j(s)\rangle$ is a $\nu$-measurable complex function for any $i,j=1,2,\ldots$. Then there exists an orthonormal sequence associated with $\xi_1,\xi_2,\ldots$.
\end{lemma}
\begin{proof}
The required sequence is constructed by applying the standard orthogonalization procedure to $\xi_1(s),\xi_2(s),\ldots$, see the proof of Lemma~1 in Sec.~II.1.2 of~Ref.~\onlinecite{Dixmier} for details.
\end{proof}

If $\xi$ and $\eta$ are $\nu$-measurable sections of a $\nu$-measurable family $\mathfrak S$ of Hilbert spaces, then $s\to \langle \xi(s),\eta(s)\rangle$ is a $\nu$-measurable complex function. Indeed, by Lemma~\ref{l_orth}, there exists an orthonormal sequence $\tilde\xi_1,\tilde\xi_2,\ldots$ associated with $\xi_1^{\mathfrak S},\xi_2^{\mathfrak S},\ldots$. By~(a), the linear span of $\tilde\xi_1(s),\tilde\xi_2(s),\ldots$ is dense in $\mathfrak S(s)$ for $\nu$-a.e. $s$. This implies that $\langle \xi(s),\eta(s)\rangle = \sum_{i=1}^\infty \langle \xi(s),\tilde\xi_i(s)\rangle \langle \tilde\xi_i(s),\eta(s)\rangle$ for $\nu$-a.e. $s$. Hence the required statement follows because $s\to \langle \xi(s),\tilde\xi_i(s)\rangle$ and $s\to \langle \tilde\xi_i(s),\eta(s)\rangle$ are $\nu$-measurable complex functions by~(b) and the $\nu$-measurability of $\xi$ and $\eta$.

\begin{lemma}\label{l_meas_sec}
Let $\mathfrak S$ be a $\nu$-measurable family of Hilbert spaces and $\xi_1,\xi_2,\ldots$ be $\nu$-measurable sections of $\mathfrak S$. Let $\xi$ be a $\nu$-a.e. defined map such that $\xi(s)$ belongs to the closed linear span of $\xi_1(s),\xi_2(s),\ldots$ for $\nu$-a.e. $s$. If $s\to \langle \xi_j(s),\xi(s)\rangle$ is a $\nu$-measurable function for all $j=1,2,\ldots$, then $\xi$ is a $\nu$-measurable section of $\mathfrak S$.
\end{lemma}
\begin{proof}
By Lemma~\ref{l_orth}, there exists an orthonormal sequence $\tilde \xi_1,\tilde\xi_2,\ldots$ associated with $\xi_1,\xi_2,\ldots$. By condition~(b), $\tilde \xi_j$ is a $\nu$-measurable section of $\mathfrak S$ for every $j=1,2,\ldots$. By condition~(a), $\tilde \xi_1(s),\tilde\xi_2(s),\ldots$ is a generalized orthonormal system whose closed linear span coincides with that of $\xi_1(s),\xi_2(s),\ldots$ and, hence, contains $\xi(s)$ for $\nu$-a.e. $s$. This implies that $\xi(s) = \sum_{j=1}^\infty \langle \tilde\xi_j(s),\xi(s)\rangle \tilde\xi_j(s)$ for $\nu$-a.e. $s$ and, therefore, $\xi$ is a $\nu$-measurable section of $\mathfrak S$.
\end{proof}

\subsection{Direct sums of measurable families} Let $\mathfrak S_1$ and $\mathfrak S_2$ be $\nu$-measurable families of Hilbert spaces. Then there is a unique (up to $\nu$-identity) $\nu$-measurable family $\mathfrak S$ of Hilbert spaces satisfying the conditions
\begin{enumerate}
\item[$(1)$] $\mathfrak S(s) = \mathfrak S_1(s)\oplus \mathfrak S_2(s)$ for $\nu$-a.e. $s$.

\item[$(2)$] A $\nu$-a.e. defined map $\zeta$ is a $\nu$-measurable section of $\mathfrak S$ if and only if there are $\nu$-measurable sections $\xi$ and $\eta$ of $\mathfrak S_1$ and $\mathfrak S_2$ respectively such that $\zeta(s)=(\xi(s),\eta(s))$ for $\nu$-a.e. $s$.
\end{enumerate}
Indeed, every $\nu$-a.e. defined map $\mathfrak S$ satisfying~(1) becomes a $\nu$-measurable family of Hilbert spaces satisfying~(2) if we choose $\xi^{\mathfrak S}_1,\xi^{\mathfrak S}_2,\ldots$ in such a way that $\xi^{\mathfrak S}_{2k}(s)=(\xi^{\mathfrak S_1}_{k}(s),0)$ and $\xi^{\mathfrak S}_{2k-1}(s)=(0,\xi^{\mathfrak S_2}_{k}(s))$ for $\nu$-a.e. $s$ and every $k=1,2,\ldots$. The $\nu$-measurable family $\mathfrak S$ of Hilbert spaces satisfying~(1) and~(2) is called the direct sum of $\mathfrak S_1$ and $\mathfrak S_2$ and is denoted by $\mathfrak S_1\oplus_\nu\mathfrak S_2$.

\subsection{Constant and discrete families}\label{s_const_fam}
Let $\mathfrak h$ be a separable Hilbert space. Then there exists a unique (up to $\nu$-identity) $\nu$-measurable family $\mathfrak S$ of Hilbert spaces satisfying the conditions
\begin{enumerate}
\item[$(1)$] $\mathfrak S(s) = \mathfrak h$ for $\nu$-a.e. $s$.

\item[$(2)$] A $\nu$-a.e. defined map $\xi$ is a $\nu$-measurable section of $\mathfrak S$ if and only if $\xi$ is an $\mathfrak h$-valued $\nu$-measurable map.
\end{enumerate}
Indeed, let $e_1,e_2,\ldots$ be an orthonormal basis in $\mathfrak h$ and $\mathfrak S$ be a $\nu$-a.e. defined map satisfying~(1). To make $\mathfrak S$ into a $\nu$-measurable family, we require $\xi^{\mathfrak S}_1,\xi^{\mathfrak S}_2,\ldots$ to be $\nu$-a.e. defined maps such that $\xi^{\mathfrak S}_i(s)=e_i$ for $\nu$-a.e. $s$. If $\xi$ is a $\nu$-measurable section of $\mathfrak S$, then $\xi$ is an $\mathfrak H$-valued $\nu$-measurable map because $\xi(s) = \sum_i \langle \xi^{\mathfrak S}_i(s),\xi(s)\rangle e_i$ for $\nu$-a.e. $s$ and, therefore, $\mathfrak S$ satisfies~(2). The $\nu$-measurable family $\mathfrak S$ of Hilbert spaces satisfying~(1) and~(2) is denoted by $\mathcal I_{\mathfrak h,\nu}$.

Suppose now that $\mathscr S$ is a countable set, $\nu$ is a counting measure on $\mathscr S$ and $\mathfrak S$ is a $\nu$-a.e. defined family of Hilbert spaces (i.e., $\mathscr S\subset D_{\mathfrak S}$ and $\mathfrak S(s)$ is a separable Hilbert space for every $s\in\mathscr S$). Since every $\nu$-a.e. defined complex function is $\nu$-measurable, there is a unique (up to $\nu$-identity) $\nu$-measurable family $\mathfrak S_{\!\!\mathscr S}$ of Hilbert spaces such that $\mathfrak S_{\!\!\mathscr S}(s)=\mathfrak S(s)$ for every $s\in\mathscr S$. Every $\nu$-a.e. defined section of $\mathfrak S_{\!\!\mathscr S}$ is $\nu$-measurable.

\subsection{Measurable families of operators}
Let $\mathfrak S$ be a $\nu$-measurable family of Hilbert spaces. A $\nu$-a.e. defined map $\mathfrak S'$ is said to be a $\nu$-a.e. defined family of subspaces of $\mathfrak S$ if $\mathfrak S'(s)$ is a linear (not necessarily closed) subspace of $\mathfrak S(s)$ for $\nu$-a.e. $s$. A $\nu$-a.e. defined family $\mathfrak S'$ of subspaces of $\mathfrak S$ is called $\nu$-measurable if there is a sequence $\xi_1,\xi_2,\ldots$ of $\nu$-measurable sections of $\mathfrak S$ such that the linear span of $\xi_1(s),\xi_2(s),\ldots$ is dense in $\mathfrak S'(s)$ for $\nu$-a.e. $s$ (such a sequence will be called a $\nu$-measurable basis in $\mathfrak S'$).

Let $\mathfrak S_1$ and $\mathfrak S_2$ be $\nu$-measurable families of Hilbert spaces. A $\nu$-a.e. defined map $\mathcal R$ is called a $\nu$-a.e. defined family of operators from $\mathfrak S_1$ to $\mathfrak S_2$ if $\mathcal R(s)$ is an operator (possibly not everywhere defined and unbounded) from $\mathfrak S_1(s)$ to $\mathfrak S_2(s)$ for $\nu$-a.e. $s$. A $\nu$-a.e. defined family $\mathcal R$ of operators from $\mathfrak S_1$ to $\mathfrak S_2$ is called $\nu$-measurable if there are sequences $\xi_1,\xi_2,\ldots$ and $\eta_1,\eta_2,\ldots$ of $\nu$-measurable sections of $\mathfrak S_1$ and $\mathfrak S_2$ respectively such that the linear span of the vectors $(\xi_j(s),\eta_j(s))$ is dense in the graph $G_{\mathcal R(s)}$ of the operator $\mathcal R(s)$ for $\nu$-a.e. $s$. In other words, $\mathcal R$ is $\nu$-measurable if $s\to G_{\mathcal R(s)}$ is a $\nu$-measurable family of subspaces of $\mathfrak S_1\oplus_\nu \mathfrak S_2$.

Given a $\nu$-measurable family $\mathfrak S$ of Hilbert spaces, we say that $\mathcal R$ is a $\nu$-a.e. defined ($\nu$-measurable) family of operators in $\mathfrak S$ if it is a $\nu$-a.e. defined (resp., $\nu$-measurable) family of operators from $\mathfrak S$ to $\mathfrak S$. If $\mathfrak h$ is a separable Hilbert space, then $\nu$-a.e. defined ($\nu$-measurable) families of operators in $\mathcal I_{\mathfrak h,\nu}$ are called $\nu$-a.e. defined (resp., $\nu$-measurable) families of operators in $\mathfrak h$.

\begin{lemma}\label{l_BBB}
Let $\mathfrak S_1$ and $\mathfrak S_2$ be $\nu$-measurable families of Hilbert spaces and $\mathcal R$ be a $\nu$-a.e. defined family of operators from $\mathfrak S_1$ to $\mathfrak S_2$ such that $\mathcal R(s)$ is an everywhere defined bounded operator for $\nu$-a.e. $s$. Then the following statements hold
\begin{enumerate}
\item[$1.$] If $\mathcal R$ is $\nu$-measurable then $s\to\mathcal R(s)\xi(s)$ is a $\nu$-measurable section of $\mathfrak S_2$ for any $\nu$-measurable section $\xi$ of $\mathfrak S_1$.
\item[$2.$] Suppose there exists a $\nu$-measurable basis $\xi_1,\xi_2,\ldots$ in $\mathfrak S_1$ such that $s\to\mathcal R(s)\xi_j(s)$ is a $\nu$-measurable section of $\mathfrak S_2$ for all $j=1,2,\ldots$. Then $\mathcal R$ is $\nu$-measurable.
\end{enumerate}
\end{lemma}
\begin{proof}
1. Let $\xi$ be a $\nu$-measurable section of $\mathfrak S_1$ and $\xi_1,\xi_2,\ldots$ and $\eta_1,\eta_2,\ldots$ be $\nu$-measurable sections of $\mathfrak S_1$ and $\mathfrak S_2$ respectively such that the linear span of the vectors $(\xi_j(s),\eta_j(s))$ is dense in $G_{\mathcal R(s)}$ for $\nu$-a.e. $s$. By Lemma~~\ref{l_orth}, there exists an orthonormal sequence $\tilde\xi_1,\tilde\xi_2,\ldots$ associated with $\xi_1,\xi_2,\ldots$. By condition~(b) of Sec.~\ref{app1_1}, we have $\mathcal R(s)\tilde\xi_j(s) = \sum_{i=1}^j f_{ij}(s)\eta_i(s)$ with some $\nu$-measurable functions $f_{ij}$ for any $j=1,2,\ldots$ and $\nu$-a.e. $s$. Hence, $s\to \mathcal R(s)\tilde \xi_j(s)$ is a $\nu$-measurable section of $\mathfrak S_2$ for any $j=1,2,\ldots$. By condition~(a) of Sec.~\ref{app1_1}, the linear span of $\tilde\xi_1(s),\tilde\xi_2(s),\ldots$ coincides with that of $\xi_1(s),\xi_2(s),\ldots$ and, therefore, is dense in $\mathfrak S_1(s)$ for $\nu$-a.e. $s$. This implies that $\xi(s) = \sum_{j=1}^\infty \langle \tilde\xi_j(s),\xi(s)\rangle \tilde\xi_j(s)$ for $\nu$-a.e. $s$. By the continuity of $\mathcal R(s)$, it follows that $\mathcal R(s)\xi(s) = \sum_{j=1}^\infty \langle \tilde\xi_j(s),\xi(s)\rangle \mathcal R(s)\tilde\xi_j(s)$ for $\nu$-a.e. $s$ and, therefore, $s\to\mathcal R(s)\xi(s)$ is a $\nu$-measurable section of $\mathfrak S_2$.

\par\medskip\noindent 2. Let $\nu$-measurable sections $\eta_1,\eta_2,\ldots$ of $\mathfrak S_2$ be such that $\eta_j(s)=\mathcal R(s)\xi_j(s)$ for all $j=1,2,\ldots$ and $\nu$-a.e. $s$. As the linear span of $\xi_1(s),\xi_2(s),\ldots$ is dense in $\mathfrak S_1(s)$ and $\mathcal R(s)$ is everywhere defined and continuous, the linear span of $(\xi_j(s),\eta_j(s))$ is dense in $G_{\mathcal R(s)}$ for $\nu$-a.e. $s$. This means that $\mathcal R$ is $\nu$-measurable.
\end{proof}

\begin{lemma}\label{l_BBB0}
Let $\mathfrak S_1$ and $\mathfrak S_2$ be $\nu$-measurable families of Hilbert spaces and $\mathcal R$ be a $\nu$-measurable family of operators from $\mathfrak S_1$ to $\mathfrak S_2$ such that $\mathcal R(s)$ is an everywhere defined bounded operator for $\nu$-a.e. $s$. Let $\mathfrak S'_1$ be a $\nu$-measurable family of subspaces of $\mathfrak S_1$. Then the images of $\mathfrak S'_1(s)$ under $\mathcal R(s)$ constitute a $\nu$-measurable family of subspaces of $\mathfrak S_2$.
\end{lemma}
\begin{proof}
Let $\xi_1,\xi_2,\ldots$ be a $\nu$-measurable basis in $\mathfrak S'_1$. By the continuity of $\mathcal R(s)$, the linear span of $\mathcal R(s)\xi_1(s),\mathcal R(s)\xi_2(s),\ldots$ is dense in the image of $\mathfrak S'_1(s)$ under $\mathcal R(s)$ for $\nu$-a.e. $s$. Hence the statement follows because $s\to \mathcal R(s)\xi_j(s)$ is a $\nu$-measurable section of $\mathfrak S_2$ for any $j=1,2,\ldots$ by Lemma~\ref{l_BBB}.
\end{proof}

\begin{lemma}\label{l_BBB1}
Let $\mathfrak S$ be a $\nu$-measurable family of Hilbert spaces and $\mathfrak S'$ and $\mathcal P$ be $\nu$-a.e. defined maps such that $\mathfrak S'(s)$ is a closed subspace of $\mathfrak S(s)$ and $\mathcal P(s)$ is the orthogonal projection of $\mathfrak S(s)$ onto $\mathfrak S'(s)$ for $\nu$-a.e. $s$. Then $\mathcal P$ is a $\nu$-measurable family of operators in $\mathfrak S$ if and only if $\mathfrak S'$ is a $\nu$-measurable family of subspaces of~$\mathfrak S$.
\end{lemma}
\begin{proof}
If $\mathcal P$ is $\nu$-measurable, then so is $\mathfrak S'$ by Lemma~\ref{l_BBB0}. Let $\mathfrak S'$ be $\nu$-measurable and $\xi_1,\xi_2,\ldots$ be a $\nu$-measurable basis in $\mathfrak S'$. Let $\xi$ be a $\nu$-measurable section of $\mathfrak S$. For $\nu$-a.e. $s$, we have $\langle \mathcal P(s)\xi(s),\xi_j(s)\rangle = \langle \xi(s),\xi_j(s)\rangle$. Hence, $s\to \langle \mathcal P(s)\xi(s),\xi_j(s)\rangle$ is a $\nu$-measurable section of $\mathfrak S$ for all $j=1,2,\ldots$. In view of Lemma~\ref{l_meas_sec}, this implies that $s\to \mathcal P(s)\xi(s)$ is a $\nu$-measurable section of $\mathfrak S$ and, therefore, $\mathcal P$ is $\nu$-measurable by Lemma~\ref{l_BBB}.
\end{proof}

Let $\mathfrak H'$ and $\mathfrak H''$ be closed subspaces of a Hilbert space $\mathfrak H$ and $P'$ and $P''$ be the orthogonal projections of $\mathfrak H$ onto $\mathfrak H'$ and $\mathfrak H''$ respectively. Then we have $\mathfrak H'\cap\mathfrak H'' = \mathrm{Ker}\,R$, where $R=1_{\mathfrak H}-(P'+P'')/2$ and $1_{\mathfrak H}$ denotes the identity operator in $\mathfrak H$. Indeed, suppose $\psi\in \mathrm{Ker}\,R$ and $\psi\notin \mathfrak H'\cap\mathfrak H''$. Assume, for definiteness, that $\psi\notin \mathfrak H'$. Then $\|P'\psi\|<\|\psi\|$. Since $\psi = (P'\psi+P''\psi)/2$ and $\|P''\psi\|\leq\|\psi\|$, this implies that $\|\psi\|<\|\psi\|$. We thus obtain a contradiction and the statement is proved.

\begin{lemma}\label{l_BBB2}
Let $\mathfrak S$ be a $\nu$-measurable family of Hilbert spaces and $\mathfrak S'$ and $\mathfrak S''$ be $\nu$-measurable families of closed subspaces of $\mathfrak S$. Then
\begin{enumerate}
\item[$1.$] If $\mathfrak S''(s)\subset \mathfrak S'(s)$ for $\nu$-a.e. $s$, then $s\to \mathfrak S'(s)\ominus \mathfrak S''(s)$, where $\mathfrak S'(s)\ominus \mathfrak S''(s)$ is the orthogonal complement of $\mathfrak S''(s)$ in $\mathfrak S'(s)$, is a $\nu$-measurable family of subspaces of $\mathfrak S$.
\item[$2.$] $s\to \mathfrak S'(s)\cap \mathfrak S''(s)$ is a $\nu$-measurable family of subspaces of $\mathfrak S$.
\end{enumerate}
\end{lemma}
\begin{proof}
1. Let $\mathcal P'$ and $\mathcal P''$ be $\nu$-a.e. defined maps such that $\mathcal P'(s)$ and $\mathcal P''(s)$ are orthogonal projections of $\mathfrak S(s)$ onto $\mathfrak S'(s)$ and $\mathfrak S''(s)$ respectively for $\nu$-a.e. $s$. By Lemma~\ref{l_BBB1}, $\mathcal P'$ and $\mathcal P''$ are $\nu$-measurable families of operators in $\mathfrak S$. It follows from Lemma~\ref{l_BBB} that $s\to \mathcal P'(s)-\mathcal P''(s)$ is also a $\nu$-measurable family of operators in $\mathfrak S$. As $\mathcal P'(s)-\mathcal P''(s)$ is the orthogonal projection of $\mathfrak S(s)$ onto $\mathfrak S'(s)\ominus \mathfrak S''(s)$ for $\nu$-a.e. $s$, the desired statement follows from Lemma~\ref{l_BBB1}.
\par\medskip\noindent
2. Let $\nu$-measurable families $\mathcal P'$ and $\mathcal P''$ of operators in $\mathfrak S$ be as in the proof of~(1) and $\mathcal R$ be a $\nu$-a.e. defined map such that $\mathcal R(s) = 1_{\mathfrak S(s)}- (\mathcal P'(s) + \mathcal P''(s))/2$ for $\nu$-a.e. $s$. It follows from Lemma~\ref{l_BBB} that $\mathcal R$ is a $\nu$-measurable family of operators in $\mathfrak S$. Since $\mathfrak S'(s)\cap \mathfrak S''(s) = \mathrm{Ker}\,\mathcal R(s) = (\mathrm{Im}\,\mathcal R(s))^\bot$ for $\nu$-a.e. $s$, the statement follows from Lemma~\ref{l_BBB0} and~(1).
\end{proof}

\begin{lemma}\label{l_BBB3}
Let $\mathfrak S_1$ and $\mathfrak S_2$ be $\nu$-measurable families of Hilbert spaces and $\mathcal R$ be a $\nu$-measurable family of operators from $\mathfrak S_1$ to $\mathfrak S_2$. Then
\begin{enumerate}
\item[$1.$] If $\mathcal R(s)$ is invertible for $\nu$-a.e. $s$, then $s\to \mathcal R(s)^{-1}$ is a $\nu$-measurable family of operators from $\mathfrak S_2$ to $\mathfrak S_1$.
\item[$2.$] If $\mathcal R(s)$ is closed for $\nu$-a.e. $s$, then $s\to\mathrm{Ker}\,\mathcal R(s)$ is a $\nu$-measurable family of subspaces of $\mathfrak S_1$.
\end{enumerate}
\end{lemma}
\begin{proof}
For $\nu$-a.e. $s$, let $\mathcal S(s)$ be everywhere defined bounded operator from $\mathfrak S_1(s)\oplus \mathfrak S_2(s)$ to $\mathfrak S_2(s)\oplus \mathfrak S_1(s)$ taking $(\psi_1,\psi_2)$ to $(\psi_2,\psi_1)$. Obviously, $\mathcal S$ is a $\nu$-measurable family of operators from $\mathfrak S_1\oplus_\nu \mathfrak S_2$ to $\mathfrak S_2\oplus_\nu \mathfrak S_1$. If $\mathcal R(s)$ is invertible for $\nu$-a.e. $s$, then $G_{\mathcal R(s)^{-1}}$ coincides with the image of $G_{\mathcal R(s)}$ under $\mathcal S(s)$ for $\nu$-a.e. $s$ and, therefore, (1) follows from Lemma~\ref{l_BBB0}. Suppose now that $\mathcal R(s)$ is closed for $\nu$-a.e. $s$.
Since $s\to \mathfrak S_1(s)\times\{0\}$ is a $\nu$-measurable family of closed subspaces of $\mathfrak S_1\oplus_\nu\mathfrak S_2$ and $G_{\mathcal R(s)}\cap (\mathfrak S_1(s)\times\{0\})=\mathrm{Ker}\,\mathcal R(s)\times\{0\}$ for $\nu$-a.e. $s$, statement~2 of Lemma~\ref{l_BBB2} implies that $s\to \mathrm{Ker}\,\mathcal R(s)\times\{0\}$ is a $\nu$-measurable family of closed subspaces of $\mathfrak S_1\oplus_\nu\mathfrak S_2$, whence (2) obviously follows.
\end{proof}

\begin{lemma}\label{l_BBB4}
Let $\mathfrak S_1$, $\mathfrak S_2$, and $\mathfrak S_3$ be $\nu$-measurable families of Hilbert spaces and $\mathcal R_1$ and $\mathcal R_2$ be  $\nu$-measurable families of closed operators from $\mathfrak S_1$ to $\mathfrak S_2$ and from $\mathfrak S_2$ to $\mathfrak S_3$ respectively. Then $s\to \mathcal R_2(s)\mathcal R_1(s)$ is a $\nu$-measurable family of operators from $\mathfrak S_1$ to $\mathfrak S_3$.
\end{lemma}
\begin{proof}
Let $\mathfrak S'$ and $\mathfrak S''$ be $\nu$-a.e. defined maps such that, for $\nu$-a.e. $s$, $\mathfrak S'(s)$ and $\mathfrak S''(s)$ are subspaces of $\mathfrak S_1(s)\oplus \mathfrak S_2(s)\oplus \mathfrak S_3(s)$ consisting of all $(\psi_1,\psi_2,\psi_3)$ with $(\psi_1,\psi_2)\in G_{\mathcal R_1(s)}$ and $(\psi_2,\psi_3)\in G_{\mathcal R_2(s)}$ respectively. Clearly, both $\mathfrak S'$ and $\mathfrak S''$ are $\nu$-measurable families of closed subspaces of $\mathfrak S_1\oplus_\nu \mathfrak S_2\oplus_\nu \mathfrak S_3$. Let $\mathcal S$ be a $\nu$-a.e. defined map such that, for $\nu$-a.e. $s$, $\mathcal S(s)$ is the everywhere defined bounded operator from $\mathfrak S_1(s)\oplus \mathfrak S_2(s)\oplus \mathfrak S_3(s)$ to $\mathfrak S_1(s)\oplus\mathfrak S_3(s)$ taking $(\psi_1,\psi_2,\psi_3)$ to $(\psi_1,\psi_3)$. Then $G_{\mathcal R_2(s)\mathcal R_1(s)}$ is the image of $\mathfrak S'(s)\cap\mathfrak S''(s)$ under $\mathcal S(s)$ for $\nu$-a.e. $s$. Hence, the result follows from Lemma~\ref{l_BBB0} and statement~2 of Lemma~\ref{l_BBB2}.
\end{proof}

\begin{lemma}\label{l_BBB4a}
Let $\mathfrak S_1$ and $\mathfrak S_2$ be $\nu$-measurable families of Hilbert spaces and $\mathcal R_1$ and $\mathcal R_2$ be $\nu$-measurable families of closed operators from $\mathfrak S_1$ to $\mathfrak S_2$. Then $s\to \mathcal R_2(s)+\mathcal R_1(s)$ is a $\nu$-measurable family of operators from $\mathfrak S_1$ to $\mathfrak S_2$.
\end{lemma}
\begin{proof}
Let $\mathfrak S'$ and $\mathfrak S''$ be $\nu$-a.e. defined maps such that, for $\nu$-a.e. $s$, $\mathfrak S'(s)$ and $\mathfrak S''(s)$ are subspaces of $\mathfrak S_1(s)\oplus \mathfrak S_2(s)\oplus \mathfrak S_2(s)$ consisting of all $(\psi_1,\psi_2,\psi_3)$ with $(\psi_1,\psi_2)\in G_{\mathcal R_1(s)}$ and $(\psi_1,\psi_3)\in G_{\mathcal R_2(s)}$ respectively. Clearly, both $\mathfrak S'$ and $\mathfrak S''$ are $\nu$-measurable families of closed subspaces of $\mathfrak S_1\oplus_\nu \mathfrak S_2\oplus_\nu \mathfrak S_2$. Let $\mathcal S$ be a $\nu$-a.e. defined map such that, for $\nu$-a.e. $s$, $\mathcal S(s)$ is the everywhere defined bounded operator from $\mathfrak S_1(s)\oplus \mathfrak S_2(s)\oplus \mathfrak S_2(s)$ to $\mathfrak S_1(s)\oplus\mathfrak S_2(s)$ taking $(\psi_1,\psi_2,\psi_3)$ to $(\psi_1,\psi_2+\psi_3)$. Then $G_{\mathcal R_2(s)+\mathcal R_1(s)}$ is the image of $\mathfrak S'(s)\cap\mathfrak S''(s)$ under $\mathcal S(s)$ for $\nu$-a.e. $s$. Hence, the result follows from Lemma~\ref{l_BBB0} and statement~2 of Lemma~\ref{l_BBB2}.
\end{proof}

\begin{lemma}\label{l_BBB5}
Let $\mathfrak S_1$ and $\mathfrak S_2$ be $\nu$-measurable families of Hilbert spaces and $\mathcal R$ be a $\nu$-measurable family of closed operators from $\mathfrak S_1$ to $\mathfrak S_2$. Let $\xi$ be a $\nu$-measurable section of $\mathfrak S_1$ such that $\xi(s)\in D_{\mathcal R(s)}$ for $\nu$-a.e. $s$. Then $s\to \mathcal R(s)\xi(s)$ is a $\nu$-measurable section of $\mathfrak S_2$.
\end{lemma}
\begin{proof}
Let $\zeta$ be a $\nu$-a.e. defined section of $\mathfrak S_1\oplus_\nu \mathfrak S_2$ such that $\zeta(s)=(\xi(s),\mathcal R(s)\xi(s))$ for $\nu$-a.e. $s$. It suffices to show that $\zeta$ is $\nu$-measurable. Let $\xi_1,\xi_2,\ldots$ be a $\nu$-measurable basis in $\mathfrak S_1$. Let $\mathcal P$ be a $\nu$-a.e. defined map such that $\mathcal P(s)$ is the orthogonal projection of $\mathfrak S_1(s)\oplus\mathfrak S_2(s)$ onto $G_{\mathcal R(s)}$ for $\nu$-a.e. $s$. For $j=1,2,\ldots$, let $\zeta_j$ be a $\nu$-a.e. defined map such that $\zeta_j(s) = \mathcal P(s)(\xi_j(s),0)$ for $\nu$-a.e. $s$. By Lemma~\ref{l_BBB1}, $\mathcal P$ is a $\nu$-measurable family of operators in $\mathfrak S_1\oplus_\nu\mathfrak S_2$, and Lemma~\ref{l_BBB} implies that $\zeta_j$ is a $\nu$-measurable section of $\mathfrak S_1\oplus_\nu\mathfrak S_2$ for all $j=1,2,\ldots$. Moreover, the linear span of $\zeta_1(s),\zeta_2(s),\ldots$ is dense in $G_{\mathcal R(s)}$ for $\nu$-a.e. $s$.\footnote{Indeed, let $\mathfrak H_1$ and $\mathfrak H_2$ be Hilbert spaces, $R$ be a closed operator from $\mathfrak H_1$ to $\mathfrak H_2$, $P$ be the orthogonal projection of $\mathfrak H_1\oplus\mathfrak H_2$ onto $G_R$, and $\psi_1,\psi_2,\ldots$ be elements of $\mathfrak H_1$ whose linear span is dense in $\mathfrak H_1$. If $(\psi,\psi')\in G_R$ is orthogonal to all vectors $P(\psi_j,0)$, then ${\langle} \psi,\psi_j{\rangle} = {\langle} (\psi,\psi'),P(\psi_j,0){\rangle} = 0$ for all $j$ and, hence, $\psi=0$. As $G_R$ is a graph of an operator, it follows that $\psi'=0$. Thus, the linear span of $P(\psi_j,0)$ is dense in $G_R$.} As $\langle \zeta_j(s),\zeta(s)\rangle = \langle \xi_j(s),\xi(s)\rangle$ for $\nu$-a.e. $s$, we conclude that $s\to \langle \zeta_j(s),\zeta(s)\rangle$ is a $\nu$-measurable function for all $j=1,2,\ldots$. Hence, $\zeta$ is $\nu$-measurable by Lemma~\ref{l_meas_sec}.
\end{proof}

\begin{remark}
The definition of measurability of a family of operators given in this section is mainly the same as in~Ref.~\onlinecite{Barriere1951}. The only difference is that, in contrast to~Ref.~\onlinecite{Barriere1951}, we do not require operators to be closed and densely defined. This is essential for Propositions~\ref{l_BBB4} and~\ref{l_BBB4a} because the properties of being closed and densely defined are not inherited by sums and products of operators.
\end{remark}

\begin{remark}\label{rem_measurability}
Let $R$ be a closed operator in a Hilbert space $\mathfrak H$, $P(R)$ be the orthogonal projection of $\mathfrak H\oplus\mathfrak H$ onto $G_R$, and $\pi_{1,2}(\mathfrak H)\colon \mathfrak H\oplus\mathfrak H\to \mathfrak H$ and $j_{1,2}(\mathfrak H)\colon \mathfrak H\to\mathfrak H\oplus\mathfrak H$ be the canonical projections and embeddings respectively. For $i,k=1,2$, we set $P_{ik}(R)= \pi_i(\mathfrak H) P(R) j_k(\mathfrak H)$. The $2\times 2$-matrix composed of bounded operators $P_{ik}(R)$ is called the characteristic matrix of (generally, unbounded) operator $R$. In~Ref.~\onlinecite{Nussbaum}, a $\nu$-a.e. defined family $\mathcal R$ of closed operators in a $\nu$-measurable family $\mathfrak S$ of Hilbert spaces was called $\nu$-measurable if $s\to P_{ik}(\mathcal R(s))\xi(s)$ is a $\nu$-measurable section of $\mathfrak S$ for all $i,k=1,2$ and every $\nu$-measurable section $\xi$ of $\mathfrak S$. It follows from Lemmas~\ref{l_BBB} and~\ref{l_BBB1} that this definition is equivalent to that given in this section because
\[
P(\mathcal R(s)) =\sum_{i,k=1}^2 j_i(\mathfrak S(s))P_{ik}(\mathcal R(s))\pi_k(\mathfrak S(s))
\]
for $\nu$-a.e. $s$ and $s\to \pi_{1,2}(\mathfrak S(s))$ and $s\to j_{1,2}(\mathfrak S(s))$ are obviously $\nu$-measurable families (in our sense) of operators from $\mathfrak S\oplus_\nu \mathfrak S$ to $\mathfrak S$ and from $\mathfrak S$ to $\mathfrak S\oplus_\nu \mathfrak S$ respectively.
\end{remark}

\begin{remark}
Let $\mathcal R$ be a $\nu$-a.e. defined family of operators in a $\nu$-measurable family $\mathfrak S$ of Hilbert spaces. Suppose $\mathcal R(s)$ is self-adjoint for $\nu$-a.e. $s$. Then $\mathcal R(s)+i$ has an everywhere defined bounded inverse for $\nu$-a.e. $s$ and Lemmas~\ref{l_BBB} and~\ref{l_BBB3} imply that $\mathcal R$ is $\nu$-measurable if and only if $s\to (\mathcal R(s)+i)^{-1}\xi(s)$ is a $\nu$-measurable section of $\mathfrak S$ for any $\nu$-measurable section $\xi$ of $\mathfrak S$. In~Ref.~\onlinecite{RS4}, the last condition was adopted as the definition of measurability for families of self-adjoint operators.
\end{remark}

\begin{remark}
A $\nu$-a.e. defined family $\mathcal R$ of operators in a $\nu$-measurable family $\mathfrak S$ of Hilbert spaces is said to be weakly $\nu$-measurable~\cite{Nussbaum} if $s\to \mathcal R(s)\xi(s)$ is a $\nu$-measurable section of $\mathfrak S$ for every $\nu$-measurable section $\xi$ of $\mathfrak S$ satisfying $\xi(s)\in D_{\mathcal R(s)}$ for $\nu$-a.e. $s$. Lemma~\ref{l_BBB5} states that every $\nu$-measurable family of closed operators is also weakly $\nu$-measurable. In~Ref.~\onlinecite{Nussbaum}, where Lemma~\ref{l_BBB5} was originally proved, a question was posed whether this statement can be reverted, i.e., whether every weakly $\nu$-measurable family of closed operators is $\nu$-measurable. In~Ref.~\onlinecite{Gesztesy2012}, it was shown by constructing a counterexample that the answer is negative.
\end{remark}

\subsection{Direct integrals of Hilbert spaces}

Given a $\nu$-measurable family $\mathfrak S$ of Hilbert spaces, we denote by $\mathcal M(\mathfrak S,\nu)$ the set of all $\nu$-equivalence classes $[f]_\nu$, where $f$ is a $\nu$-measurable section of $\mathfrak S$. Clearly, $\mathcal M(\mathfrak S,\nu)$ has a natural structure of a complex vector space (for any $\nu$-measurable sections $\xi$ and $\eta$ and any $k\in\mathbb C$, we set $[\xi]_{\nu} + [\eta]_{\nu} = [\xi+\eta]_{\nu}$ and $k[\xi]_{\nu} = [k\xi]_{\nu}$). Now suppose $\mathfrak S'$ is a $\nu$-a.e. defined family of subspaces of $\mathfrak S$. We denote by $(\mathfrak S)\mbox{-}\!\int^\oplus\mathfrak S'(s)\,d\nu(s)$ the linear subspace of the space $\mathcal M(\mathfrak S,\nu)$ consisting of all its elements $\xi$ such that $\xi(s)\in \mathfrak S'(s)$ for $\nu$-a.e. $s$ and $s\to\|\xi(s)\|^2$ is a $\nu$-integrable function. The space $(\mathfrak S)\mbox{-}\!\int^\oplus\mathfrak S'(s)\,d\nu(s)$ is endowed with the scalar product defined by the relation
\[
\langle \xi,\eta\rangle = \int \langle \xi(s),\eta(s)\rangle\,d\nu(s).
\]
For any $\nu$-measurable family $\mathfrak S$, the space $(\mathfrak S)\mbox{-}\!\int^\oplus\mathfrak S(s)\,d\nu(s)$ is complete (the proof is essentially the same as that of completeness of ordinary $L_2$-spaces) and, hence, is a Hilbert space.

As a rule, the $\nu$-measurable family $\mathfrak S$ can be easily deduced from the context. So we usually omit the prefix $(\mathfrak S)\mbox{-}$ and write  $\int^\oplus\mathfrak S'(s)\,d\nu(s)$ in place of $(\mathfrak S)\mbox{-}\!\int^\oplus\mathfrak S'(s)\,d\nu(s)$. The simplest examples of direct integrals are $L_2$-spaces and countable direct sums of Hilbert spaces, which can be represented using the constant and discrete families of Hilbert spaces respectively (see Sec.~\ref{s_const_fam}). Indeed, for any separable Hilbert space $\mathfrak h$, we obviously have
\begin{equation}\label{l2dirint}
L_2(\s_\nu,\mathfrak h,\nu) = \int^\oplus\mathcal I_{\mathfrak h,\nu}(s)\,d\nu(s).
\end{equation}
Similarly, if $\mathscr S$ is a countable set, $\nu$ is a counting measure on $\mathscr S$, and $\mathfrak S$ is a $\nu$-a.e. defined family of Hilbert spaces, then
\begin{equation}\label{3dirint}
\bigoplus_{s\in\mathscr S} \mathfrak S(s) = \int^\oplus \mathfrak S_{\!\!\mathscr S}(s)\,d\nu(s).
\end{equation}

Given a $\nu$-a.e. defined section $\xi$ of $\mathfrak S$ and a $\nu$-measurable set $A$, we denote by $\xi^A$ the $\nu$-equivalence class such that $\xi^A(s)=\xi(s)$ for $\nu$-a.e. $s\in A$ and $\xi^A(s)=0$ for $\nu$-a.e. $s\in \s_\nu\setminus A$.

\begin{lemma}\label{l_di2}
Let $\mathfrak S$ be a $\nu$-measurable family of Hilbert spaces and $\mathfrak S'$ be a $\nu$-measurable family of subspaces of $\mathfrak S$. Then there is a $\nu$-measurable basis in $\mathfrak S'$ consisting of elements of $\int^\oplus\mathfrak S'(s)\,d\nu(s)$.
\end{lemma}
\begin{proof}
Let $\xi_1,\xi_2,\ldots$ be a $\nu$-measurable basis in $\mathfrak S'$. Multiplying $\xi_j$ by suitable $\nu$-measurable functions, we can ensure that $\|\xi_j(s)\|\leq 1$ for all $j=1,2,\ldots$ and $\nu$-a.e. $s$. Let $A_1,A_2,\ldots$ be elements of $D_\nu$ such that $\s_\nu=\bigcup_{k=1}^\infty A_k$. Then $\xi^{A_k}_j$ with $j,k=1,2,\ldots$ obviously constitute the required $\nu$-measurable basis in $\mathfrak S'$.
\end{proof}

\begin{lemma}\label{l_di1}
Let $\mathfrak S$ and $\mathfrak S'$ be as in Lemma~$\mathrm{\ref{l_di2}}$
and $\mathfrak H'=\int^\oplus \mathfrak S'(s)\,d\nu(s)$. Then $\mathfrak H^{\prime\bot} = \int^\oplus \mathfrak S'(s)^\bot\,d\nu(s)$ and $\overline{\mathfrak H'}=\int^\oplus\overline{\mathfrak S'(s)}\,d\nu(s)$.
\end{lemma}
\begin{proof}
Let $\eta\in \mathfrak H^{\prime\bot}$ and $\xi\in \mathfrak H'$. As $\xi^A\in \mathfrak H'$ for any $A\in D_\nu$, we obtain $\int_A\langle \xi(s),\eta(s)\rangle\,d\nu(s) = \langle \xi_A,\eta\rangle=0$ and, hence, $\langle \xi(s),\eta(s)\rangle=0$ for $\nu$-a.e. $s$. By Lemma~\ref{l_di2}, there is a $\nu$-measurable basis $\xi_1,\xi_2,\ldots$ in $\mathfrak S'$ such that $\xi_j\in \mathfrak H'$ for all $j=1,2,\ldots$. Since $\langle \xi_j(s),\eta(s)\rangle=0$ for $\nu$-a.e. $s$ and all $j=1,2,\ldots$, we have $\eta(s)\in \mathfrak S'(s)^\bot$ for $\nu$-a.e. $s$. This implies that $\mathfrak H^{\prime\bot}=\int^\oplus\mathfrak S'(s)^\bot\,d\nu(s)$. The equality for the closure of $\mathfrak H'$ now follows from the relations $\overline{\mathfrak H'} = (\mathfrak H^{\prime\bot})^\bot$ and $\overline{\mathfrak S'(s)}=(\mathfrak S'(s)^\bot)^\bot$.
\end{proof}

Let $\mathfrak S$ be a $\nu$-measurable family of Hilbert spaces and $\mathfrak H = \int^\oplus\mathfrak S(s)\,d\nu(s)$. Given a $\nu$-measurable complex function $g$, we denote by $\mathcal T^{\nu,\mathfrak S}_g$ the operator of multiplication by $g$ in $\mathfrak H$. By definition, the graph of $\mathcal T^{\nu,\mathfrak S}_g$ consists of all pairs $(\xi_1,\xi_2)\in \mathfrak H\oplus \mathfrak H$ such that $\xi_2(s) = g(s)\xi_1(s)$ for $\nu$-a.e. $s$. The operator $\mathcal T^{\nu,\mathfrak S}_g$ is closed and densely defined and
its adjoint is equal to $\mathcal T^{\nu,\mathfrak S}_{\bar g}$, where $\bar g$ is the complex conjugate function of $g$. In particular, if $g$ is real, then $\mathcal T^{\nu,\mathfrak S}_{g}$ is self-adjoint. If $g$ is $\nu$-essentially bounded, then $\mathcal T^{\nu,\mathfrak S}_{g}$ is everywhere defined and bounded.

\begin{lemma}\label{l_di3}
Let $\mathfrak S$ be a $\nu$-measurable family of Hilbert spaces, $\mathfrak H = \int^\oplus\mathfrak S(s)\,d\nu(s)$, and $\mathfrak H'$ be a closed linear subspace of $\mathfrak H$ such that $\mathcal T^{\nu,\mathfrak S}_{g}\xi\in \mathfrak H'$ for any $\xi\in \mathfrak H'$ and every $\nu$-essentially bounded complex function $g$. Then there is a $\nu$-measurable family $\mathfrak S'$ of closed subspaces of $\mathfrak S$ such that $\mathfrak H' = \int^\oplus \mathfrak S'(s)\,d\nu(s)$.
\end{lemma}
\begin{proof}
By the hypothesis, the projection $P$ of $\mathfrak H$ onto $\mathfrak H'$ commutes with $\mathcal T^{\nu,\mathfrak S}_{g}$ for every $\nu$-essentially bounded function $g$. By Lemma~\ref{l_di2}, there is a $\nu$-measurable basis $\xi_1,\xi_2,\ldots$ in $\mathfrak S$ consisting of elements of $\mathfrak H$. For each $j=1,2,\ldots$, we set $\eta_j=P\xi_j$. Let $\mathfrak S'$ be a $\nu$-measurable family of subspaces of $\mathfrak S$ such that $\mathfrak S'(s)$ is the closed linear span of $\eta_1(s),\eta_2(s),\ldots$ for $\nu$-a.e. $s$. Let $\tilde {\mathfrak H} = \int^\oplus \mathfrak S'(s)\,d\nu(s)$. If $\xi\in \mathfrak H$ is such that $\langle \xi^A_j,\xi\rangle=\int_A\langle \xi_j(s),\xi(s)\rangle\,d\nu(s)=0$ for every $\nu$-measurable set $A$ and $j=1,2,\ldots$, then $\langle\xi_j(s),\xi(s)\rangle=0$ for all $j$ and $\nu$-a.e. $s$ and, therefore, $\xi=0$. This means that the linear span of all $\xi^A_j$ is dense in $\mathfrak H$. Note that $\eta^A_j = \mathcal T^{\nu,\mathfrak S}_{\chi_A}P\xi_j= P\xi^A_j$, where $\chi_A(s)=1$ for $s\in A$ and $\chi_A(s)=0$ for $s\in\s_\nu\setminus A$. Hence, the linear span of $\eta^A_j$ is dense in $\mathfrak H'$. Since $\tilde {\mathfrak H}$ is closed and contains all $\eta^A_j$, we conclude that $\mathfrak H'\subset \tilde {\mathfrak H}$. Let $\eta\in \mathfrak H^{\prime\bot}$. Then we have $\langle \eta^A_j,\eta\rangle=\int_A\langle \eta_j(s),\eta(s)\rangle\,d\nu(s)=0$ and, hence, $\langle\eta_j(s),\eta(s)\rangle=0$ for all $j$ and $\nu$-a.e. $s$. This implies that $\eta(s)\in \mathfrak S'(s)^\bot$ for $\nu$-a.e. $s$, i.e., $\eta\in \tilde{\mathfrak H}^\bot$. We therefore obtain $\mathfrak H^{\prime\bot}\subset \tilde{\mathfrak H}^\bot$. As $\mathfrak H'$ is closed, it follows that $\tilde{\mathfrak H}\subset \mathfrak H'$ and, hence, $\tilde{\mathfrak H}=\mathfrak H'$.
\end{proof}

\subsection{Direct integrals of operators}

Let $\mathfrak S_1$ and $\mathfrak S_2$ be $\nu$-measurable families of Hilbert spaces. Let $\mathfrak H_{1,2} = \int^\oplus \mathfrak S_{1,2}(s)\,d\nu(s)$ and $\mathcal R$ be a $\nu$-a.e.
defined family of operators from $\mathfrak S_1$ to $\mathfrak S_2$.
The direct integral $\int^{\oplus} \mathcal R(s)\,d\nu(s)$ of the family $\mathcal R$ is defined as the linear operator from $\mathfrak H_1$ to $\mathfrak H_2$ whose graph consists of all pairs $(\xi,\eta)\in \mathfrak H_1\oplus\mathfrak H_2$ such that $\xi(s)\in D_{\mathcal R(s)}$ and $\mathcal R(s)\xi(s) = \eta(s)$ for $\nu$-a.e. $s$.

Let $\mathfrak H = \int^{\oplus} \mathfrak S_1(s)\oplus\mathfrak S_2(s)\,d\nu(s)$. Then there is a unique unitary operator $U_{\mathfrak S_1,\mathfrak S_2}\colon \mathfrak H_1\oplus\mathfrak H_2\to\mathfrak H$ such that $(U_{\mathfrak S_1,\mathfrak S_2}(\xi,\eta))(s) = (\xi(s),\eta(s))$ for any $(\xi,\eta)\in \mathfrak H_1\oplus \mathfrak H_2$ and $\nu$-a.e. $s$. We call $U_{\mathfrak S_1,\mathfrak S_2}$ the natural isomorphism between $\mathfrak H_1\oplus\mathfrak H_2$ and $\mathfrak H$.

The next statement follows immediately from the definition of $U_{\mathfrak S_1,\mathfrak S_2}$.
\begin{lemma}\label{l_di3a}
Let $\mathfrak S_1$ and $\mathfrak S_2$ be $\nu$-measurable families of Hilbert spaces, $\mathcal R$ be a $\nu$-a.e.
defined family of operators from $\mathfrak S_1$ to $\mathfrak S_2$, and $R = \int^\oplus\mathcal R(s)\,d\nu(s)$. Then we have
$U_{\mathfrak S_1,\mathfrak S_2}(G_R) = (\mathfrak S_1\oplus_\nu\mathfrak S_2)\mbox{-}\!\int^{\oplus} G_{\mathcal R(s)}\,d\nu(s)$.
\end{lemma}

\begin{lemma}\label{l_di4}
Let $\mathfrak H_{1,2} = \int^\oplus \mathfrak S_{1,2}(s)\,d\nu(s)$, where $\mathfrak S_1$ and $\mathfrak S_2$ are $\nu$-measurable families of Hilbert spaces,  and $G$ be a closed linear subspace of $\mathfrak H_1\oplus\mathfrak H_2$. Suppose $\mathcal G$ is a $\nu$-measurable family of closed subspaces of $\mathfrak S_1\oplus_\nu\mathfrak S_2$ such that $U_{\mathfrak S_1,\mathfrak S_2}(G) = \int^\oplus \mathcal G(s)\,d\nu(s)$. Then $G$ is a graph of an operator from $\mathfrak H_1$ to $\mathfrak H_2$ if and only if $\mathcal G(s)$ is a graph of an operator from $\mathfrak S_1(s)$ to $\mathfrak S_2(s)$ for $\nu$-a.e. $s$.
\end{lemma}
\begin{proof}
Let $U=U_{\mathfrak S_1,\mathfrak S_2}$ and $\mathfrak S'$ be a $\nu$-a.e. defined map such that $\mathfrak S'(s) = \{0\}\times \mathfrak S_2(s)$ for $\nu$-a.e. $s$. Clearly, $\mathfrak S'$ is a $\nu$-measurable family of closed subspaces of $\mathfrak S_1\oplus_\nu \mathfrak S_2$. Let $\mathfrak H'=\int^\oplus\mathfrak S'(s)\,d\nu(s)$ and $\mathfrak H''= \int^\oplus\mathfrak S'(s)\cap \mathcal G(s)\,d\nu(s)$. Since $\mathfrak H'' = \mathfrak H'\cap U(G)$ and $U^{-1}(\mathfrak H') = \{0\}\times\mathfrak H_2$, we have $G\cap (\{0\}\times\mathfrak H_2) = U^{-1}(\mathfrak H'')$. Hence, $G$ is a graph of an operator if and only if $\mathfrak H''$ is trivial. By statement~2 of Lemma~\ref{l_BBB2}, $s\to \mathfrak S'(s)\cap \mathcal G(s)$ is a $\nu$-measurable family of subspaces of $\mathfrak S_1\oplus_\nu \mathfrak S_2$ and, therefore, $\mathfrak H''$ is trivial if and only if $\mathfrak S'(s)\cap \mathcal G(s)$ is trivial for $\nu$-a.e. $s$, i.e., if $\mathcal G(s)$ is a graph of an operator for $\nu$-a.e. $s$.
\end{proof}

\begin{lemma}\label{l_di4a}
Let $\mathfrak S_1$ and $\mathfrak S_2$ be $\nu$-measurable families of Hilbert spaces, $\mathcal R$ be a $\nu$-measurable family of operators from $\mathfrak S_1$ to $\mathfrak S_2$, and $R = \int^\oplus \mathcal R(s)\,d\nu(s)$. Then $R$ is closable if and only if $\mathcal R(s)$ is closable for $\nu$-a.e. $s$, in which case $\overline{R} = \int^\oplus \overline{\mathcal R(s)}\,d\nu(s)$.
\end{lemma}
\begin{proof}
Let $U=U_{\mathfrak S_1,\mathfrak S_2}$. As $U(G_R) = \int^{\oplus} G_{\mathcal R(s)}\,d\nu(s)$, Lemma~\ref{l_di1} yields
\begin{equation}\label{eq_di1}
U(\overline{G}_R) = \int^{\oplus} \overline{G}_{\mathcal R(s)}\,d\nu(s).
\end{equation}
The operator $R$ is closable if and only if $\overline{G}_R$ is a graph of an operator. By Lemma~\ref{l_di4}, the latter condition holds if and only if $\overline{G}_{\mathcal R(s)}$ is a graph of an operator for $\nu$-a.e. $s$, i.e., if $\mathcal R(s)$ is closable for $\nu$-a.e. $s$. Suppose now that $R$ is closable and $R'=\int^\oplus \overline{\mathcal R(s)}\,d\nu(s)$. As $\overline{G}_R = G_{\overline{R}}$ and $\overline{G}_{\mathcal R(s)} = G_{\overline{\mathcal R(s)}}$ for $\nu$-a.e. $s$, equality~(\ref{eq_di1}) and Lemma~\ref{l_di3a} imply that $U(G_{\overline{R}}) = U(G_{R'})$ and, hence, $R=R'$.
\end{proof}

\begin{proposition}\label{l_di4b}
Let $\mathfrak S_1$ and $\mathfrak S_2$ be $\nu$-measurable families of Hilbert spaces and $\mathcal R$ and $\mathcal R'$ be $\nu$-measurable families of closed operators from $\mathfrak S_1$ to $\mathfrak S_2$. Let $R = \int^\oplus \mathcal R(s)\,d\nu(s)$, and $R' = \int^\oplus \mathcal R'(s)\,d\nu(s)$. If $R$ is an extension of $R'$, then $\mathcal R(s)$ is an extension of $\mathcal R'(s)$ for $\nu$-a.e. $s$. If $R=R'$, then $\mathcal R(s)=\mathcal R'(s)$ for $\nu$-a.e.~$s$.
\end{proposition}
\begin{proof}
Let $U=U_{\mathfrak S_1,\mathfrak S_2}$.
Since $U(G_{R'}) = \int^\oplus G_{\mathcal R'(s)}\,d\nu(s)$, Lemma~\ref{l_di2} implies that there exists a sequence $\zeta_1,\zeta_2,\ldots$ of elements of $U(G_{R'})$ such that the linear span of $\zeta_1(s),\zeta_2(s),\ldots$ is dense in $G_{\mathcal R'(s)}$ for $\nu$-a.e. $s$. If $R$ is an extension of $R'$, then $U(G_{R'})\subset U(G_R)$ and, hence, $\zeta_j\in U(G_R)$ for all $j=1,2,\ldots$. Since $U(G_{R}) = \int^\oplus G_{\mathcal R(s)}\,d\nu(s)$, it follows that $\zeta_j(s)\in G_{\mathcal R(s)}$ for $\nu$-a.e. $s$ and all $j=1,2,\ldots$. The linear span of $\zeta_1(s),\zeta_2(s),\ldots$ is therefore contained in $G_{\mathcal R(s)}$ for $\nu$-a.e. $s$. As $\mathcal R(s)$ is closed, it follows that $G_{\mathcal R'(s)}\subset G_{\mathcal R(s)}$, i.e., $\mathcal R(s)$ is an extension of $\mathcal R'(s)$ for $\nu$-a.e. $s$. If $R=R'$, then, by the above, $\mathcal R(s)$ and $\mathcal R'(s)$ are extensions of each other and, therefore, are equal for $\nu$-a.e. $s$.
\end{proof}

\begin{proposition}\label{p_di}
Let $\mathfrak S_1$, $\mathfrak S_2$, $\mathfrak H_1$, and $\mathfrak H_2$ be as in Lemma~$\mathrm{\ref{l_di4}}$. Let $\mathcal R$ be a $\nu$-measurable family of closed operators from $\mathfrak S_1$ to $\mathfrak S_2$. Then $R=\int^\oplus \mathcal R(s)\,d\nu(s)$ is a closed operator from $\mathfrak H_1$ to $\mathfrak H_2$ satisfying the condition:
\begin{itemize}
\item[{\rm (M)}] If $\xi\in D_R$ and $g$ is a $\nu$-measurable $\nu$-essentially bounded complex function, then $\mathcal T^{\nu,\mathfrak S_1}_g\xi\in D_R$ and $R\mathcal T^{\nu,\mathfrak S_1}_g\xi=\mathcal T^{\nu,\mathfrak S_2}_gR\xi$.
\end{itemize}
Conversely, if $R$ is a closed operator from $\mathfrak H_1$ to $\mathfrak H_2$ satisfying~{\rm (M)}, then there is a unique (up to $\nu$-equivalence) $\nu$-measurable family $\mathcal R$ of closed operators from $\mathfrak S_1$ to $\mathfrak S_2$ such that $R=\int^\oplus \mathcal R(s)\,d\nu(s)$.
\end{proposition}
\begin{proof}
Let $U=U_{\mathfrak S_1,\mathfrak S_2}$, $\mathcal R$ be a $\nu$-measurable family of closed operators from $\mathfrak S_1$ to $\mathfrak S_2$, and $R=\int^\oplus \mathcal R(s)\,d\nu(s)$. By Lemma~\ref{l_di4a}, $R$ is closed. If $\xi\in D_R$ and $g$ is a $\nu$-measurable $\nu$-essentially bounded function, then $(\mathcal T^{\nu,\mathfrak S_1}_g\xi,\mathcal T^{\nu,\mathfrak S_2}_gR\xi)$ belongs to $G_R$ by the very definition of the direct integral of operators. This means that (M) is fulfilled. Conversely, let $R$ be a closed operator from $\mathfrak H_1$ to $\mathfrak H_2$ satisfying~{\rm (M)}. If $(\xi,\eta)\in G_R$ and $g$ is a $\nu$-measurable $\nu$-essentially bounded function, then $(\mathcal T^{\nu,\mathfrak S_1}_g\xi,\mathcal T^{\nu,\mathfrak S_2}_g\eta)$ belongs to $G_R$ by~(M) and, hence, $\mathcal T^{\nu,\mathfrak S_1\oplus_\nu \mathfrak S_2}_gU(\xi,\eta)=U(\mathcal T^{\nu,\mathfrak S_1}_g\xi,\mathcal T^{\nu,\mathfrak S_2}_g\eta)$ belongs to $U(G_R)$. This means that $U(G_R)$ is invariant under $\mathcal T^{\nu,\mathfrak S_1\oplus_\nu \mathfrak S_2}_g$. By Lemma~\ref{l_di3}, there is a $\nu$-measurable family $\mathcal G$ of closed subspaces of $\mathfrak S_1\oplus_\nu \mathfrak S_2$ such that $U(G_R) = \int^\oplus \mathcal G(s)\,d\nu(s)$. By Lemma~\ref{l_di4}, there is a $\nu$-measurable family $\mathcal R$ of closed operators from $\mathfrak S_1$ to $\mathfrak S_2$ such that $\mathcal G(s) = G_{\mathcal R(s)}$ for $\nu$-a.e. $s$. If $R'=\int^\oplus \mathcal R(s)\,d\nu(s)$, then it follows from Lemma~\ref{l_di3a} that $U(G_{R'}) = U(G_R)$ and, hence, $R = R'$. The uniqueness of $\mathcal R$ is ensured by Proposition~\ref{l_di4b}.
\end{proof}

Let $\mathfrak H_1$ and $\mathfrak H_2$ be Hilbert spaces and $R$ be an operator from $\mathfrak H_1$ to $\mathfrak H_2$. We define the linear subspaces $G^\circ_R$ and $G^*_R$ of $\mathfrak H_2\oplus\mathfrak H_1$ by setting
\[
G^{\circ}_R = \{(-\psi_2,\psi_1): (\psi_1,\psi_2)\in G_R\},\quad G^*_R = (G^{\circ}_R)^\bot.
\]
The operator $R$ is densely defined if and only if $G^*_R$ is the graph of an operator, in which case $G_{R^*}=G^*_R$.

\begin{proposition}\label{p_di2}
Let $\mathfrak S_1$ and $\mathfrak S_2$ be $\nu$-measurable families of Hilbert spaces, $\mathcal R$ be a $\nu$-measurable family of operators from $\mathfrak S_1$ to $\mathfrak S_2$, and $R=\int^\oplus\mathcal R(s)\,d\nu(s)$. Then the following statements hold:
\begin{enumerate}
\item[$1.$] Suppose $\mathcal R(s)$ is closed for $\nu$-a.e. $s$. The operator $R$ is invertible if and only if $\mathcal R(s)$ is invertible for $\nu$-a.e. $s$, in which case $s\to \mathcal R(s)^{-1}$ is a $\nu$-measurable family of operators from $\mathfrak S_2$ to $\mathfrak S_1$ and $R^{-1}= \int^\oplus \mathcal R(s)^{-1}\,d\nu(s)$.
\item[$2.$] The operator $R$ is densely defined if and only if $\mathcal R(s)$ is densely defined for $\nu$-a.e. $s$, in which case $s\to \mathcal R(s)^{*}$ is a $\nu$-measurable family of operators from $\mathfrak S_2$ to $\mathfrak S_1$ and $R^{*}= \int^\oplus \mathcal R(s)^{*}\,d\nu(s)$.
\end{enumerate}
\end{proposition}
\begin{proof} Let $\mathfrak H_{1,2}=\int^\oplus \mathfrak S_{1,2}(s)\,d\nu(s)$.
\par\noindent
1. By statement~2 of Lemma~\ref{l_BBB3}, $s\to \mathrm{Ker}\,\mathcal R(s)$ is a $\nu$-measurable family of subspaces of $\mathfrak S_1$. As $\mathrm{Ker}\,R = \int^\oplus \mathrm{Ker}\,\mathcal R(s)\,d\nu(s)$, we conclude that $R$ is invertible if and only if $\mathrm{Ker}\,\mathcal R(s)=\{0\}$ for $\nu$-a.e. $s$, i.e., if $\mathcal R(s)$ is invertible for $\nu$-a.e. $s$. Let $R$ be invertible and $R'= \int^\oplus \mathcal R(s)^{-1}\,d\nu(s)$. For any $(\eta,\xi)\in \mathfrak H_2\oplus \mathfrak H_1$, we have
\begin{multline}\nonumber
(\eta,\xi)\in G_{R^{-1}} \Longleftrightarrow (\xi,\eta)\in G_R \Longleftrightarrow (\xi(s),\eta(s))\in G_{\mathcal R(s)}\mbox{ for $\nu$-a.e. $s$} \Longleftrightarrow \\ (\eta(s),\xi(s))\in G_{\mathcal R(s)^{-1}}\mbox{ for $\nu$-a.e. $s$} \Longleftrightarrow (\eta,\xi)\in G_{R'}.
\end{multline}
We thus have $G_{R^{-1}}=G_{R'}$ and, hence, $R=R'$.  By statement~1 of Lemma~\ref{l_BBB3}, $s\to \mathcal R(s)^{-1}$ is a $\nu$-measurable family of operators from $\mathfrak S_2$ to $\mathfrak S_1$.

\par\medskip\noindent
2. Let $U_{12}=U_{\mathfrak S_1,\mathfrak S_2}$ and $U_{21}=U_{\mathfrak S_2,\mathfrak S_1}$. For $\nu$-a.e. $s$, let $\mathcal S(s)$ be everywhere defined bounded operator from $\mathfrak S_1(s)\oplus \mathfrak S_2(s)$ to $\mathfrak S_2(s)\oplus \mathfrak S_1(s)$ taking $(\psi_1,\psi_2)$ to $(-\psi_2,\psi_1)$. Clearly, $\mathcal S$ is a $\nu$-measurable family of unitary operators from $\mathfrak S_1\oplus_\nu\mathfrak S_2$ to $\mathfrak S_2\oplus_\nu \mathfrak S_1$. Note that $G^\circ_{\mathcal R(s)}$ is the image of $G_{\mathcal R(s)}$ under $\mathcal S(s)$ for $\nu$-a.e. $s$. Hence, $s\to G^\circ_{\mathcal R(s)}$ is a $\nu$-measurable family of subspaces of $\mathfrak S_2\oplus_\nu \mathfrak S_1$ by Lemma~\ref{l_BBB0}. Let $S = \int^\oplus\mathcal S(s)\,d\nu(s)$ and $\tilde S$ be the unitary operator from $\mathfrak H_1\oplus \mathfrak H_2$ to $\mathfrak H_2\oplus \mathfrak H_1$ defined by the relation $\tilde S(\xi,\eta)=(-\eta,\xi)$. It is straightforward to check that $U_{21}\tilde S = S U_{12}$. As $G^{\circ}_R = \tilde S(G_R)$, we have $U_{21}(G^{\circ}_R) = S(U_{12}(G_R))$. Since $U_{12}(G_R)=\int^\oplus G_{\mathcal R(s)}\,d\nu(s)$ by Lemma~\ref{l_di3a} and the image of $\int^\oplus G_{\mathcal R(s)}\,d\nu(s)$ under $S$ is equal to $\int^\oplus G^{\circ}_{\mathcal R(s)}\,d\nu(s)$, we conclude that $U_{21}(G^{\circ}_R) = \int^\oplus G^{\circ}_{\mathcal R(s)}\,d\nu(s)$. It now follows from Lemmas~\ref{l_BBB2} and~\ref{l_di1} and the unitarity of $U_{21}$ that $s\to G^*_{\mathcal R(s)}$ is a $\nu$-measurable family of subspaces of $\mathfrak S_2\oplus_\nu \mathfrak S_1$ and
\begin{equation}\label{eq_di3}
U_{21}(G^{*}_R) = \int^\oplus G^{*}_{\mathcal R(s)}\,d\nu(s).
\end{equation}
In view of~(\ref{eq_di3}), Lemma~\ref{l_di4} implies that $R$ is densely defined if and only if $G^{*}_{\mathcal R(s)}$ is the graph of an operator for $\nu$-a.e. $s$, i.e., if $\mathcal R(s)$ is densely defined for $\nu$-a.e. $s$. Suppose $R$ is densely defined and $R'=\int^\oplus \mathcal R(s)^{*}\,d\nu(s)$. As  $G^{*}_{\mathcal R(s)} = G_{\mathcal R(s)^{*}}$ for $\nu$-a.e. $s$, it follows that $s\to \mathcal R(s)^*$ is a $\nu$-measurable family of operators from $\mathfrak S_2$ to $\mathfrak S_1$. Finally, since $G^{*}_R = G_{R^{*}}$, equality~(\ref{eq_di3}) and Lemma~\ref{l_di3a} imply that $U_{21}(G_{R^{*}}) = U_{21}(G_{R'})$ and, hence, $R^{*}=R'$.
\end{proof}

\begin{corollary}\label{c_di}
Let $\mathfrak S$ be a $\nu$-measurable family of Hilbert spaces, and $\mathcal R$ be a $\nu$-measurable family of closed operators in $\mathfrak S$. The operator $\int^\oplus\mathcal R(s)\,d\nu(s)$ is self-adjoint if and only if $\mathcal R(s)$ is self-adjoint for $\nu$-a.e. $s$.
\end{corollary}
\begin{proof}
The statement follows from Propositions~\ref{p_di2} and~\ref{l_di4b}.
\end{proof}

\begin{proposition}\label{p_di1}
Let $\mathfrak S_1$, $\mathfrak S_2$, and $\mathfrak S_3$ be $\nu$-measurable families of Hilbert spaces and $\mathcal R_1$ and $\mathcal R_2$ be $\nu$-measurable families of closed operators from $\mathfrak S_1$ to $\mathfrak S_2$ and from $\mathfrak S_2$ to $\mathfrak S_3$ respectively. Let $R_{1,2}=\int^\oplus \mathcal R_{1,2}(s)\,d\nu(s)$. The operator $R_2R_1$ is densely defined if and only if the operator $\mathcal R_2(s)\mathcal R_1(s)$ is densely defined for $\nu$-a.e. $s$.
The operator $R_2R_1$ is closable if and only if the operator $\mathcal R_2(s)\mathcal R_1(s)$ is closable for $\nu$-a.e. $s$, in which case $s\to \overline{\mathcal R_2(s)\mathcal R_1(s)}$ is a $\nu$-measurable family of operators from $\mathfrak S_1$ to $\mathfrak S_3$ and $\overline{R_2R_1} = \int^\oplus \overline{\mathcal R_2(s)\mathcal R_1(s)}\,d\nu(s)$.
\end{proposition}
\begin{proof}
Let $\mathfrak H_i = \int^\oplus \mathfrak S_i(s)\,d\nu(s)$, $i=1,2,3$, and $R = \int^\oplus \mathcal R_2(s)\mathcal R_1(s)\,d\nu(s)$. Clearly, $R$ is an extension of $R_2R_1$. Let $(\xi,\eta)\in G_R$. Then $\xi(s)\in D_{\mathcal R_2(s)\mathcal R_1(s)}$ and, hence, $\xi(s)\in D_{\mathcal R_1(s)}$ for $\nu$-a.e. $s$. Let $\tilde\eta$ be a $\nu$-a.e. defined section of $\mathfrak S_2$ such that $\tilde\eta(s)=\mathcal R_1(s)\xi(s)$ for $\nu$-a.e. $s$. By Lemma~\ref{l_BBB5}, $\tilde\eta$ is $\nu$-measurable. Let $A_1\subset A_2\subset\ldots$ be a sequence of elements of $D_\nu$ such that $\s_\nu=\bigcup_{k=1}^\infty A_k$ and $\|\tilde\eta(s)\|\leq k$ for $\nu$-a.e. $s\in A_k$.
Set $\xi_k=\xi^{A_k}$, $\eta_k=\eta^{A_k}$, and $\tilde\eta_k=\tilde\eta^{A_k}$
We obviously have $\xi_k\in \mathfrak H_1$, $\tilde\eta_k\in \mathfrak H_2$, and $\eta_k\in \mathfrak H_3$ for all $k$. As $\tilde \eta(s)\in D_{\mathcal R_2(s)}$ and $\eta(s)=\mathcal R_2(s)\tilde\eta(s)$ for $\nu$-a.e. $s$, we have $(\xi_k,\tilde\eta_k)\in G_{R_1}$ and $(\tilde\eta_k,\eta_k)\in G_{R_2}$ and, therefore, $(\xi_k,\eta_k)\in G_{R_2R_1}$ for all $k$. Since $(\xi_k,\eta_k)\to (\xi,\eta)$ in $\mathfrak H_1\oplus\mathfrak H_3$ as $k\to\infty$, we conclude that $G_{R_2R_1}$ is dense in $G_R$ and, hence, $D_{R_2R_1}$ is dense in $D_R$. In view of statement~2 of Proposition~\ref{p_di2}, it follows that $R_2R_1$ is densely defined if and only if $\mathcal R_2(s)\mathcal R_1(s)$ is densely defined for $\nu$-a.e. $s$. Since $\overline{G}_R = \overline{G}_{R_2R_1}$, the operator $R_2R_1$ is closable if and only if $R$ is closable. In view of Lemma~\ref{l_di4a}, the latter condition holds if and only if $\mathcal R_2(s)\mathcal R_1(s)$ is closable for $\nu$-a.e. $s$ (note that $s\to R_2(s)\mathcal R_1(s)$ is a $\nu$-measurable family of operators from $\mathfrak S_1$ to $\mathfrak S_3$ by Lemma~\ref{l_BBB4}). If $R_2R_1$ is closable, then Lemma~\ref{l_di4a} implies that $\overline{R_2R_1} = \overline{R}=\int^\oplus \overline{\mathcal R_2(s)\mathcal R_1(s)}\,d\nu(s)$.
\end{proof}

\begin{proposition}\label{p_di+}
Let $\mathfrak S_1$ and $\mathfrak S_2$ be $\nu$-measurable families of Hilbert spaces, $\mathcal R_1$ and $\mathcal R_2$ be $\nu$-measurable families of closed operators from $\mathfrak S_1$ to $\mathfrak S_2$, and $R_{1,2}=\int^\oplus \mathcal R_{1,2}(s)\,d\nu(s)$. The operator $R_1+R_2$ is densely defined if and only if the operator $\mathcal R_1(s)+\mathcal R_2(s)$ is densely defined for $\nu$-a.e. $s$.
The operator $R_1+R_2$ is closable if and only if the operator $\mathcal R_1(s)+\mathcal R_2(s)$ is closable for $\nu$-a.e. $s$, in which case $s\to \overline{\mathcal R_1(s)+\mathcal R_2(s)}$ is a $\nu$-measurable family of operators from $\mathfrak S_1$ to $\mathfrak S_2$ and $\overline{R_1+R_2} = \int^\oplus \overline{\mathcal R_1(s)+\mathcal R_2(s)}\,d\nu(s)$.
\end{proposition}
\begin{proof}
Let $\mathfrak H_{1,2} = \int^\oplus \mathfrak S_{1,2}(s)\,d\nu(s)$ and $R = \int^\oplus (\mathcal R_1(s)+\mathcal R_2(s))\,d\nu(s)$. Clearly, $R$ is an extension of $R_1+R_2$. Let $(\xi,\eta)\in G_R$. Then $\xi(s)\in D_{\mathcal R_1(s)}\cap D_{\mathcal R_2(s)}$ for $\nu$-a.e. $s$. Let $\eta^{(1)}$ and $\eta^{(2)}$ be $\nu$-a.e. defined sections of $\mathfrak S_2$ such that $\eta^{(1,2)}(s)=\mathcal R_{1,2}(s)\xi(s)$ for $\nu$-a.e. $s$. By Lemma~\ref{l_BBB5}, $\eta^{(1)}$ and $\eta^{(2)}$ are $\nu$-measurable. Let $A_1\subset A_2\subset\ldots$ be a sequence of elements of $D_\nu$ such that $\s_\nu=\bigcup_{k=1}^\infty A_k$ and $\|\eta^{(1)}(s)\|+\|\eta^{(2)}(s)\|\leq k$ for $\nu$-a.e. $s\in A_k$.
Set $\xi_k=\xi^{A_k}$, $\eta_k=\eta^{A_k}$, $\eta^{(1,2)}_k=(\eta^{(1,2)})^{A_k}$.
We obviously have $\xi_k\in \mathfrak H_1$ and $\eta_k,\eta^{(1,2)}_k\in \mathfrak H_2$ for all $k$.
As $(\xi_k,\eta^{(1,2)}_k)\in G_{R_{1,2}}$ and $\eta_k=\eta^{(1)}_k + \eta^{(2)}_k$, we have $(\xi_k,\eta_k)\in G_{R_1+R_2}$ for all $k$.
Since $(\xi_k,\eta_k)\to (\xi,\eta)$ in $\mathfrak H_1\oplus\mathfrak H_2$ as $k\to\infty$, we conclude that $G_{R_1+R_2}$ is dense in $G_R$ and, hence, $D_{R_1+R_2}$ is dense in $D_R$. In view of statement~2 of Proposition~\ref{p_di2}, it follows that $R_1+R_2$ is densely defined if and only if $\mathcal R_1(s)+\mathcal R_2(s)$ is densely defined for $\nu$-a.e. $s$. Since $\overline{G}_R = \overline{G}_{R_1+R_2}$, the operator $R_1+R_2$ is closable if and only if $R$ is closable. In view of Lemma~\ref{l_di4a}, the latter condition holds if and only if $\mathcal R_1(s)+\mathcal R_2(s)$ is closable for $\nu$-a.e. $s$ (note that $s\to R_1(s)+\mathcal R_2(s)$ is a $\nu$-measurable family of operators from $\mathfrak S_1$ to $\mathfrak S_2$ by Lemma~\ref{l_BBB4a}). If $R_1+R_2$ is closable, then Lemma~\ref{l_di4a} implies that $\overline{R_1+R_2} = \overline{R}=\int^\oplus \overline{\mathcal R_1(s)+\mathcal R_2(s)}\,d\nu(s)$.
\end{proof}

\begin{remark}
Propositions~\ref{l_di4b}, \ref{p_di}, and~\ref{p_di2} were proved in~Ref.~\onlinecite{Nussbaum} (see also~Ref.~\onlinecite{Lance1975}). They are generalizations of corresponding statements proved in~Ref.~\onlinecite{Neumann} for bounded operators. Propositions~\ref{p_di1} and~\ref{p_di+} were obtained in~Ref.~\onlinecite{Dixon1971} under  additional assumption that $\nu$ is a finite Borel measure on a metrizable compact space and were proved in~Ref.~\onlinecite{Lennon1974} in their general form.
\end{remark}

\bibliography{smirnov-bib}

\begin{thebibliography}{49}%
\makeatletter
\providecommand \@ifxundefined [1]{%
 \@ifx{#1\undefined}
}%
\providecommand \@ifnum [1]{%
 \ifnum #1\expandafter \@firstoftwo
 \else \expandafter \@secondoftwo
 \fi
}%
\providecommand \@ifx [1]{%
 \ifx #1\expandafter \@firstoftwo
 \else \expandafter \@secondoftwo
 \fi
}%
\providecommand \natexlab [1]{#1}%
\providecommand \enquote  [1]{``#1''}%
\providecommand \bibnamefont  [1]{#1}%
\providecommand \bibfnamefont [1]{#1}%
\providecommand \citenamefont [1]{#1}%
\providecommand \href@noop [0]{\@secondoftwo}%
\providecommand \href [0]{\begingroup \@sanitize@url \@href}%
\providecommand \@href[1]{\@@startlink{#1}\@@href}%
\providecommand \@@href[1]{\endgroup#1\@@endlink}%
\providecommand \@sanitize@url [0]{\catcode `\\12\catcode `\$12\catcode
  `\&12\catcode `\#12\catcode `\^12\catcode `\_12\catcode `\%12\relax}%
\providecommand \@@startlink[1]{}%
\providecommand \@@endlink[0]{}%
\providecommand \url  [0]{\begingroup\@sanitize@url \@url }%
\providecommand \@url [1]{\endgroup\@href {#1}{\urlprefix }}%
\providecommand \urlprefix  [0]{URL }%
\providecommand \Eprint [0]{\href }%
\providecommand \doibase [0]{http://dx.doi.org/}%
\providecommand \selectlanguage [0]{\@gobble}%
\providecommand \bibinfo  [0]{\@secondoftwo}%
\providecommand \bibfield  [0]{\@secondoftwo}%
\providecommand \translation [1]{[#1]}%
\providecommand \BibitemOpen [0]{}%
\providecommand \bibitemStop [0]{}%
\providecommand \bibitemNoStop [0]{.\EOS\space}%
\providecommand \EOS [0]{\spacefactor3000\relax}%
\providecommand \BibitemShut  [1]{\csname bibitem#1\endcsname}%
\let\auto@bib@innerbib\@empty
\bibitem [{\citenamefont {Albeverio}\ \emph {et~al.}(1988)\citenamefont
  {Albeverio}, \citenamefont {Gesztesy}, \citenamefont {H{\o}egh-Krohn},\ and\
  \citenamefont {Holden}}]{Albeverio1988}%
  \BibitemOpen
  \bibfield  {author} {\bibinfo {author} {\bibfnamefont {S.}~\bibnamefont
  {Albeverio}}, \bibinfo {author} {\bibfnamefont {F.}~\bibnamefont {Gesztesy}},
  \bibinfo {author} {\bibfnamefont {R.}~\bibnamefont {H{\o}egh-Krohn}}, \ and\
  \bibinfo {author} {\bibfnamefont {H.}~\bibnamefont {Holden}},\ }\href@noop {}
  {\emph {\bibinfo {title} {Solvable models in quantum mechanics}}},\ Texts and
  Monographs in Physics\ (\bibinfo  {publisher} {Springer-Verlag, New York},\
  \bibinfo {year} {1988})\BibitemShut {NoStop}%
\bibitem [{\citenamefont {Gitman}, \citenamefont {Tyutin},\ and\ \citenamefont
  {Voronov}(2012)}]{GTV2012}%
  \BibitemOpen
  \bibfield  {author} {\bibinfo {author} {\bibfnamefont {D.~M.}\ \bibnamefont
  {Gitman}}, \bibinfo {author} {\bibfnamefont {I.~V.}\ \bibnamefont {Tyutin}},
  \ and\ \bibinfo {author} {\bibfnamefont {B.~L.}\ \bibnamefont {Voronov}},\
  }\href@noop {} {\emph {\bibinfo {title} {Self-adjoint extensions in quantum
  mechanics; general theory and applications to Schr{\"o}dinger and Dirac
  equations with singular potentials}}},\ \bibinfo {series} {Progress in
  Mathematical Physics}, Vol.~\bibinfo {volume} {62}\ (\bibinfo  {publisher}
  {Birkh\"auser/Springer, New York},\ \bibinfo {year} {2012})\BibitemShut
  {NoStop}%
\bibitem [{\citenamefont {Akhiezer}\ and\ \citenamefont
  {Glazman}(1981)}]{AkhiezerGlazman}%
  \BibitemOpen
  \bibfield  {author} {\bibinfo {author} {\bibfnamefont {N.~I.}\ \bibnamefont
  {Akhiezer}}\ and\ \bibinfo {author} {\bibfnamefont {I.~M.}\ \bibnamefont
  {Glazman}},\ }\href@noop {} {\emph {\bibinfo {title} {Theory of linear
  operators in {H}ilbert space. {V}ol. {I}}}},\ \bibinfo {series} {Monographs
  and Studies in Mathematics}, Vol.~\bibinfo {volume} {9}\ (\bibinfo
  {publisher} {Pitman (Advanced Publishing Program), Boston, Mass.-London},\
  \bibinfo {year} {1981})\ \bibinfo {note} {translated from the third Russian
  edition by E. R. Dawson, Translation edited by W. N. Everitt}\BibitemShut
  {NoStop}%
\bibitem [{\citenamefont {Smirnov}(2013)}]{JMAA2013}%
  \BibitemOpen
  \bibfield  {author} {\bibinfo {author} {\bibfnamefont {A.~G.}\ \bibnamefont
  {Smirnov}},\ }\bibfield  {title} {\enquote {\bibinfo {title} {Generators of
  von~{N}eumann algebras associated with spectral measures},}\ }\href@noop {}
  {\bibfield  {journal} {\bibinfo  {journal} {J. Math. Anal. Appl.}\ }\textbf
  {\bibinfo {volume} {398}},\ \bibinfo {pages} {501--507} (\bibinfo {year}
  {2013})}\BibitemShut {NoStop}%
\bibitem [{Note1()}]{Note1}%
  \BibitemOpen
  \bibinfo {note} {The counting measure on $\protect \mathscr S$ is, by
  definition, the measure assigning to every finite subset of $\protect
  \mathscr S$ the number of points in this set.}\BibitemShut {Stop}%
\bibitem [{\citenamefont {von Neumann}(1949)}]{Neumann}%
  \BibitemOpen
  \bibfield  {author} {\bibinfo {author} {\bibfnamefont {J.}~\bibnamefont {von
  Neumann}},\ }\bibfield  {title} {\enquote {\bibinfo {title} {On rings of
  operators. {R}eduction theory},}\ }\href@noop {} {\bibfield  {journal}
  {\bibinfo  {journal} {Ann. of Math. (2)}\ }\textbf {\bibinfo {volume} {50}},\
  \bibinfo {pages} {401--485} (\bibinfo {year} {1949})}\BibitemShut {NoStop}%
\bibitem [{\citenamefont {Nussbaum}(1964)}]{Nussbaum}%
  \BibitemOpen
  \bibfield  {author} {\bibinfo {author} {\bibfnamefont {A.~E.}\ \bibnamefont
  {Nussbaum}},\ }\bibfield  {title} {\enquote {\bibinfo {title} {Reduction
  theory for unbounded closed operators in {H}ilbert space},}\ }\href@noop {}
  {\bibfield  {journal} {\bibinfo  {journal} {Duke Math. J.}\ }\textbf
  {\bibinfo {volume} {31}},\ \bibinfo {pages} {33--44} (\bibinfo {year}
  {1964})}\BibitemShut {NoStop}%
\bibitem [{\citenamefont {Pallu de~la Barri{\`e}re}(1951)}]{Barriere1951}%
  \BibitemOpen
  \bibfield  {author} {\bibinfo {author} {\bibfnamefont {R.}~\bibnamefont
  {Pallu de~la Barri{\`e}re}},\ }\bibfield  {title} {\enquote {\bibinfo {title}
  {D\'ecomposition des op\'erateurs non born\'es dans les sommes continues
  d'espaces de {H}ilbert},}\ }\href@noop {} {\bibfield  {journal} {\bibinfo
  {journal} {C. R. Acad. Sci. Paris}\ }\textbf {\bibinfo {volume} {232}},\
  \bibinfo {pages} {2071--2073} (\bibinfo {year} {1951})}\BibitemShut {NoStop}%
\bibitem [{Note2()}]{Note2}%
  \BibitemOpen
  \bibinfo {note} {Throughout the paper, a.e. means either 'almost every' or
  'almost everywhere'.}\BibitemShut {Stop}%
\bibitem [{\citenamefont {Aharonov}\ and\ \citenamefont {Bohm}(1959)}]{AB1959}%
  \BibitemOpen
  \bibfield  {author} {\bibinfo {author} {\bibfnamefont {Y.}~\bibnamefont
  {Aharonov}}\ and\ \bibinfo {author} {\bibfnamefont {D.}~\bibnamefont
  {Bohm}},\ }\bibfield  {title} {\enquote {\bibinfo {title} {Significance of
  electromagnetic potentials in the quantum theory},}\ }\href@noop {}
  {\bibfield  {journal} {\bibinfo  {journal} {Phys. Rev. (2)}\ }\textbf
  {\bibinfo {volume} {115}},\ \bibinfo {pages} {485--491} (\bibinfo {year}
  {1959})}\BibitemShut {NoStop}%
\bibitem [{\citenamefont {Tyutin}()}]{Tyutin1974}%
  \BibitemOpen
  \bibfield  {author} {\bibinfo {author} {\bibfnamefont {I.~V.}\ \bibnamefont
  {Tyutin}},\ }\bibfield  {title} {\enquote {\bibinfo {title} {Electron
  scattering on a solenoid},}\ }\href@noop {} {\ }\bibinfo {note}
  {{a}rXiv:0801.2167 [quant-ph] (translated from preprint FIAN 27, P.N.~Lebedev
  Physical Institute, Moscow, 1974)}\BibitemShut {NoStop}%
\bibitem [{\citenamefont {Adami}\ and\ \citenamefont
  {Teta}(1998)}]{AdamiTeta1998}%
  \BibitemOpen
  \bibfield  {author} {\bibinfo {author} {\bibfnamefont {R.}~\bibnamefont
  {Adami}}\ and\ \bibinfo {author} {\bibfnamefont {A.}~\bibnamefont {Teta}},\
  }\bibfield  {title} {\enquote {\bibinfo {title} {On the {A}haronov-{B}ohm
  {H}amiltonian},}\ }\href@noop {} {\bibfield  {journal} {\bibinfo  {journal}
  {Lett. Math. Phys.}\ }\textbf {\bibinfo {volume} {43}},\ \bibinfo {pages}
  {43--53} (\bibinfo {year} {1998})}\BibitemShut {NoStop}%
\bibitem [{\citenamefont {D{\c{a}}browski}\ and\ \citenamefont
  {{\v{S}}{\v{t}}ov{\'{\i}}{\v{c}}ek}(1998)}]{DabrowskiStovicek1998}%
  \BibitemOpen
  \bibfield  {author} {\bibinfo {author} {\bibfnamefont {L.}~\bibnamefont
  {D{\c{a}}browski}}\ and\ \bibinfo {author} {\bibfnamefont {P.}~\bibnamefont
  {{\v{S}}{\v{t}}ov{\'{\i}}{\v{c}}ek}},\ }\bibfield  {title} {\enquote
  {\bibinfo {title} {Aharonov-{B}ohm effect with {$\delta$}-type
  interaction},}\ }\href@noop {} {\bibfield  {journal} {\bibinfo  {journal} {J.
  Math. Phys.}\ }\textbf {\bibinfo {volume} {39}},\ \bibinfo {pages} {47--62}
  (\bibinfo {year} {1998})}\BibitemShut {NoStop}%
\bibitem [{Note3()}]{Note3}%
  \BibitemOpen
  \bibinfo {note} {We endow $S$ with the topology induced from ${\protect
  \mathbb R}^2$. The Borel $\sigma $-algebra of $S$ consists of all Borel
  subsets of ${\protect \mathbb R}^2$ that are contained in $S$. The Borel
  structure on $A^\phi $ is induced from $S$.}\BibitemShut {Stop}%
\bibitem [{Note4()}]{Note4}%
  \BibitemOpen
  \bibinfo {note} {For brevity, we let $\protect \mathcal W^\phi _\theta
  (s,E|x)$ denote the value of the function $\protect \mathcal W^\phi _\theta
  (s,E)$ at the point $x$: $\protect \mathcal W^\phi _\theta (s,E|x)=(\protect
  \mathcal W^\phi _\theta (s,E))(x)$. A similar notation will be used for any
  maps whose values are also maps.}\BibitemShut {Stop}%
\bibitem [{Note5()}]{Note5}%
  \BibitemOpen
  \bibinfo {note} {To compute the limit of $u^\kappa _\vartheta (E|r)$ as
  $\kappa \to 0$, one has to apply L'H\^{o}pital's rule and use the equality
  $\Gamma '(1+n)/\Gamma (1+n)=c_n-\gamma $ (see~Ref.~\protect \rev@citealpnum
  {Bateman}, Sec.~1.7.1, formula~(9)).}\BibitemShut {Stop}%
\bibitem [{Note6()}]{Note6}%
  \BibitemOpen
  \bibinfo {note} {See~Ref.~\protect \rev@citealpnum {Bateman}, Sec.~7.2.4,
  formula~(33).}\BibitemShut {Stop}%
\bibitem [{\citenamefont {Smirnov}()}]{Smirnov2015}%
  \BibitemOpen
  \bibfield  {author} {\bibinfo {author} {\bibfnamefont {A.~G.}\ \bibnamefont
  {Smirnov}},\ }\href@noop {} {\enquote {\bibinfo {title} {Eigenfunction
  expansions for the {S}chr\"odinger equation with inverse-square potential},}\
  }\bibinfo {note} {Theor. Math. Phys. (to be published); e-print
  arXiv:1508.07747 [math-ph]}\BibitemShut {NoStop}%
\bibitem [{\citenamefont {Dixmier}(1969)}]{Dixmier}%
  \BibitemOpen
  \bibfield  {author} {\bibinfo {author} {\bibfnamefont {J.}~\bibnamefont
  {Dixmier}},\ }\href@noop {} {\emph {\bibinfo {title} {Les alg\`ebres
  d'op\'erateurs dans l'espace hilbertien (alg\`ebres de von {N}eumann)}}}\
  (\bibinfo  {publisher} {Gauthier-Villars \'Editeur, Paris},\ \bibinfo {year}
  {1969})\ \bibinfo {note} {deuxi{\`e}me {\'e}dition, revue et augment{\'e}e,
  Cahiers Scientifiques, Fasc. XXV}\BibitemShut {NoStop}%
\bibitem [{Note7()}]{Note7}%
  \BibitemOpen
  \bibinfo {note} {Here and subsequently, $S_\nu $ denotes the largest $\nu
  $-measurable set, see Sec.~\ref {measures_sec}.}\BibitemShut {Stop}%
\bibitem [{Note8()}]{Note8}%
  \BibitemOpen
  \bibinfo {note} {See Sec.~\ref {s_standard}.}\BibitemShut {Stop}%
\bibitem [{Note9()}]{Note9}%
  \BibitemOpen
  \bibinfo {note} {Recall that a closed densely defined linear operator $T$ in
  a Hilbert space is called normal if the operators $TT^*$ and $T^*T$ have the
  same domain of definition and coincide thereon. In particular, self-adjoint
  and unitary operators are normal.}\BibitemShut {Stop}%
\bibitem [{\citenamefont {Fuglede}(1950)}]{Fuglede}%
  \BibitemOpen
  \bibfield  {author} {\bibinfo {author} {\bibfnamefont {B.}~\bibnamefont
  {Fuglede}},\ }\bibfield  {title} {\enquote {\bibinfo {title} {A commutativity
  theorem for normal operators},}\ }\href@noop {} {\bibfield  {journal}
  {\bibinfo  {journal} {Proc. Nat. Acad. Sci. U. S. A.}\ }\textbf {\bibinfo
  {volume} {36}},\ \bibinfo {pages} {35--40} (\bibinfo {year}
  {1950})}\BibitemShut {NoStop}%
\bibitem [{\citenamefont {Na{\u\i}mark}(1968)}]{Naimark}%
  \BibitemOpen
  \bibfield  {author} {\bibinfo {author} {\bibfnamefont {M.~A.}\ \bibnamefont
  {Na{\u\i}mark}},\ }\href@noop {} {\emph {\bibinfo {title} {Linear
  differential operators. {P}art {II}: {L}inear differential operators in
  {H}ilbert space}}},\ With additional material by the author, and a supplement
  by V. \`E. Ljance. Translated from the Russian by E. R. Dawson. English
  translation edited by W. N. Everitt\ (\bibinfo  {publisher} {Frederick Ungar
  Publishing Co.},\ \bibinfo {address} {New York},\ \bibinfo {year}
  {1968})\BibitemShut {NoStop}%
\bibitem [{\citenamefont {Teschl}(2009)}]{Teschl}%
  \BibitemOpen
  \bibfield  {author} {\bibinfo {author} {\bibfnamefont {G.}~\bibnamefont
  {Teschl}},\ }\href@noop {} {\emph {\bibinfo {title} {Mathematical methods in
  quantum mechanics, with applications to Schr{\"o}dinger operators}}},\
  \bibinfo {series} {Graduate Studies in Mathematics}, Vol.~\bibinfo {volume}
  {99}\ (\bibinfo  {publisher} {American Mathematical Society},\ \bibinfo
  {address} {Providence, RI},\ \bibinfo {year} {2009})\BibitemShut {NoStop}%
\bibitem [{\citenamefont {Weidmann}(1987)}]{Weidmann}%
  \BibitemOpen
  \bibfield  {author} {\bibinfo {author} {\bibfnamefont {J.}~\bibnamefont
  {Weidmann}},\ }\href@noop {} {\emph {\bibinfo {title} {Spectral theory of
  ordinary differential operators}}},\ \bibinfo {series} {Lecture Notes in
  Mathematics}, Vol.\ \bibinfo {volume} {1258}\ (\bibinfo  {publisher}
  {Springer-Verlag, Berlin},\ \bibinfo {year} {1987})\BibitemShut {NoStop}%
\bibitem [{Note10()}]{Note10}%
  \BibitemOpen
  \bibinfo {note} {Throughout this section and Sec.~\ref {s_meas}, all
  equivalence classes are taken with respect to $\lambda _{a,b}$. We shall drop
  the subscript and write $[f]$ instead of $[f]_{\lambda _{a,b}}$.}\BibitemShut
  {Stop}%
\bibitem [{\citenamefont {Gesztesy}\ and\ \citenamefont
  {Zinchenko}(2006)}]{GesztesyZinchenko}%
  \BibitemOpen
  \bibfield  {author} {\bibinfo {author} {\bibfnamefont {F.}~\bibnamefont
  {Gesztesy}}\ and\ \bibinfo {author} {\bibfnamefont {M.}~\bibnamefont
  {Zinchenko}},\ }\bibfield  {title} {\enquote {\bibinfo {title} {On spectral
  theory for {S}chr\"odinger operators with strongly singular potentials},}\
  }\href@noop {} {\bibfield  {journal} {\bibinfo  {journal} {Math. Nachr.}\
  }\textbf {\bibinfo {volume} {279}},\ \bibinfo {pages} {1041--1082} (\bibinfo
  {year} {2006})}\BibitemShut {NoStop}%
\bibitem [{\citenamefont {Gitman}, \citenamefont {Tyutin},\ and\ \citenamefont
  {Voronov}(2010)}]{GTV2010}%
  \BibitemOpen
  \bibfield  {author} {\bibinfo {author} {\bibfnamefont {D.~M.}\ \bibnamefont
  {Gitman}}, \bibinfo {author} {\bibfnamefont {I.~V.}\ \bibnamefont {Tyutin}},
  \ and\ \bibinfo {author} {\bibfnamefont {B.~L.}\ \bibnamefont {Voronov}},\
  }\bibfield  {title} {\enquote {\bibinfo {title} {Self-adjoint extensions and
  spectral analysis in the {C}alogero problem},}\ }\href@noop {} {\bibfield
  {journal} {\bibinfo  {journal} {J. Phys. A}\ }\textbf {\bibinfo {volume}
  {43}},\ \bibinfo {pages} {145205, 34} (\bibinfo {year} {2010})}\BibitemShut
  {NoStop}%
\bibitem [{\citenamefont {Kostenko}, \citenamefont {Sakhnovich},\ and\
  \citenamefont {Teschl}(2012)}]{KST}%
  \BibitemOpen
  \bibfield  {author} {\bibinfo {author} {\bibfnamefont {A.}~\bibnamefont
  {Kostenko}}, \bibinfo {author} {\bibfnamefont {A.}~\bibnamefont
  {Sakhnovich}}, \ and\ \bibinfo {author} {\bibfnamefont {G.}~\bibnamefont
  {Teschl}},\ }\bibfield  {title} {\enquote {\bibinfo {title}
  {Weyl-{T}itchmarsh theory for {S}chr\"odinger operators with strongly
  singular potentials},}\ }\href {\doibase 10.1093/imrn/rnr065} {\bibfield
  {journal} {\bibinfo  {journal} {Int. Math. Res. Not.}\ }\textbf {\bibinfo
  {volume} {2012}},\ \bibinfo {pages} {1699--1747} (\bibinfo {year}
  {2012})}\BibitemShut {NoStop}%
\bibitem [{Note11()}]{Note11}%
  \BibitemOpen
  \bibinfo {note} {As in Sec.~\ref {s_one}, we write $[f]$ in place of
  $[f]_{\lambda _{a,b}}$, where $\lambda _{a,b}$ is the restriction to $(a,b)$
  of the Lebesgue measure $\lambda $ on ${\protect \mathbb R}$.}\BibitemShut
  {Stop}%
\bibitem [{Note12()}]{Note12}%
  \BibitemOpen
  \bibinfo {note} {We denote by $\xi '(s|r)$ the derivative of $\xi (s)$ at
  point $r$: $\xi '(s|r)=(\xi (s))'(r)$.}\BibitemShut {Stop}%
\bibitem [{Note13()}]{Note13}%
  \BibitemOpen
  \bibinfo {note} {Given $\psi =(\psi _1,\psi _2)\in \protect \mathfrak h\oplus
  \protect \mathfrak h$, we set $\protect \mathaccentV {bar}016\psi =(\protect
  \mathaccentV {bar}016\psi _1,\protect \mathaccentV {bar}016\psi _2)$, where
  $\protect \mathaccentV {bar}016\psi _1$ and $\protect \mathaccentV
  {bar}016\psi _2$ are elements of $\protect \mathfrak h$ such that $\protect
  \mathaccentV {bar}016\psi _{1,2}(r) = \protect \overline {\psi _{1,2}(r)}$
  for $\lambda $-a.e. $r\in (a,b)$.}\BibitemShut {Stop}%
\bibitem [{\citenamefont {Bauer}(2001)}]{Bauer}%
  \BibitemOpen
  \bibfield  {author} {\bibinfo {author} {\bibfnamefont {H.}~\bibnamefont
  {Bauer}},\ }\href@noop {} {\emph {\bibinfo {title} {Measure and integration
  theory}}},\ \bibinfo {series} {de Gruyter Studies in Mathematics},
  Vol.~\bibinfo {volume} {26}\ (\bibinfo  {publisher} {Walter de Gruyter \&
  Co., Berlin},\ \bibinfo {year} {2001})\ \bibinfo {note} {translated from the
  German by Robert B. Burckel}\BibitemShut {NoStop}%
\bibitem [{Note14()}]{Note14}%
  \BibitemOpen
  \bibinfo {note} {It is easy to see that the intersection of any set of
  $\alpha $-classes is again an $\alpha $-class (note that such an intersection
  is always nonempty because $\varnothing $ is an element of every $\alpha
  $-class). Hence, given a set of sets $\protect \mathcal K$, there exists the
  smallest $\alpha $-class containing $\protect \mathcal K$.}\BibitemShut
  {Stop}%
\bibitem [{\citenamefont {Cohn}(1980)}]{Cohn}%
  \BibitemOpen
  \bibfield  {author} {\bibinfo {author} {\bibfnamefont {D.~L.}\ \bibnamefont
  {Cohn}},\ }\href@noop {} {\emph {\bibinfo {title} {Measure theory}}}\
  (\bibinfo  {publisher} {Birkh\"auser, Boston, Mass.},\ \bibinfo {year}
  {1980})\BibitemShut {NoStop}%
\bibitem [{\citenamefont {Mackey}(1957)}]{Mackey}%
  \BibitemOpen
  \bibfield  {author} {\bibinfo {author} {\bibfnamefont {G.~W.}\ \bibnamefont
  {Mackey}},\ }\bibfield  {title} {\enquote {\bibinfo {title} {Borel structure
  in groups and their duals},}\ }\href@noop {} {\bibfield  {journal} {\bibinfo
  {journal} {Trans. Amer. Math. Soc.}\ }\textbf {\bibinfo {volume} {85}},\
  \bibinfo {pages} {134--165} (\bibinfo {year} {1957})}\BibitemShut {NoStop}%
\bibitem [{\citenamefont {Rohlin}(1949)}]{Rohlin}%
  \BibitemOpen
  \bibfield  {author} {\bibinfo {author} {\bibfnamefont {V.~A.}\ \bibnamefont
  {Rohlin}},\ }\bibfield  {title} {\enquote {\bibinfo {title} {On the
  fundamental ideas of measure theory},}\ }\href@noop {} {\bibfield  {journal}
  {\bibinfo  {journal} {Mat. Sbornik N.S.}\ }\textbf {\bibinfo {volume}
  {25(67)}},\ \bibinfo {pages} {107--150} (\bibinfo {year} {1949})}\BibitemShut
  {NoStop}%
\bibitem [{Note15()}]{Note15}%
  \BibitemOpen
  \bibinfo {note} {Note that a Hilbert space is separable if it is either
  finite-dimensional or infinite-dimensional with a countable orthonormal
  basis. In particular, we consider ${\protect \mathbb C}$ as a one-dimensional
  Hilbert space with the scalar product ${\delimiter "426830A }\alpha ,\beta
  {\delimiter "526930B } =\protect \mathaccentV {bar}016\alpha \beta $, $\alpha
  ,\beta \in {\protect \mathbb C}$.}\BibitemShut {Stop}%
\bibitem [{\citenamefont {Birman}\ and\ \citenamefont
  {Solomjak}(1987)}]{BirmanSolomjak}%
  \BibitemOpen
  \bibfield  {author} {\bibinfo {author} {\bibfnamefont {M.~S.}\ \bibnamefont
  {Birman}}\ and\ \bibinfo {author} {\bibfnamefont {M.~Z.}\ \bibnamefont
  {Solomjak}},\ }\href@noop {} {\emph {\bibinfo {title} {Spectral theory of
  selfadjoint operators in {H}ilbert space}}},\ Mathematics and its
  Applications (Soviet Series)\ (\bibinfo  {publisher} {D. Reidel Publishing
  Co.},\ \bibinfo {address} {Dordrecht},\ \bibinfo {year} {1987})\ \bibinfo
  {note} {translated from the 1980 Russian original by S. Khrushch{\"e}v and V.
  Peller}\BibitemShut {NoStop}%
\bibitem [{Note16()}]{Note16}%
  \BibitemOpen
  \bibinfo {note} {A detailed account of the reduction theory for bounded
  operators can be found in~Ref.~\protect \rev@citealpnum
  {Dixmier}.}\BibitemShut {Stop}%
\bibitem [{\citenamefont {Stone}(1951)}]{Stone1951}%
  \BibitemOpen
  \bibfield  {author} {\bibinfo {author} {\bibfnamefont {M.~H.}\ \bibnamefont
  {Stone}},\ }\bibfield  {title} {\enquote {\bibinfo {title} {On unbounded
  operators in {H}ilbert space},}\ }\href@noop {} {\bibfield  {journal}
  {\bibinfo  {journal} {J. Indian Math. Soc. (N.S.)}\ }\textbf {\bibinfo
  {volume} {15}},\ \bibinfo {pages} {155--192} (\bibinfo {year}
  {1951})}\BibitemShut {NoStop}%
\bibitem [{\citenamefont {Lennon}(1974)}]{Lennon1974}%
  \BibitemOpen
  \bibfield  {author} {\bibinfo {author} {\bibfnamefont {M.~J.~J.}\
  \bibnamefont {Lennon}},\ }\bibfield  {title} {\enquote {\bibinfo {title} {On
  sums and products of unbounded operators in {H}ilbert space},}\ }\href@noop
  {} {\bibfield  {journal} {\bibinfo  {journal} {Trans. Amer. Math. Soc.}\
  }\textbf {\bibinfo {volume} {198}},\ \bibinfo {pages} {273--285} (\bibinfo
  {year} {1974})}\BibitemShut {NoStop}%
\bibitem [{Note17()}]{Note17}%
  \BibitemOpen
  \bibinfo {note} {Indeed, let $\protect \mathfrak H_1$ and $\protect \mathfrak
  H_2$ be Hilbert spaces, $R$ be a closed operator from $\protect \mathfrak
  H_1$ to $\protect \mathfrak H_2$, $P$ be the orthogonal projection of
  $\protect \mathfrak H_1\oplus \protect \mathfrak H_2$ onto $G_R$, and $\psi
  _1,\psi _2,\protect \ldots $ be elements of $\protect \mathfrak H_1$ whose
  linear span is dense in $\protect \mathfrak H_1$. If $(\psi ,\psi ')\in G_R$
  is orthogonal to all vectors $P(\psi _j,0)$, then ${\delimiter "426830A }
  \psi ,\psi _j{\delimiter "526930B } = {\delimiter "426830A } (\psi ,\psi
  '),P(\psi _j,0){\delimiter "526930B } = 0$ for all $j$ and, hence, $\psi =0$.
  As $G_R$ is a graph of an operator, it follows that $\psi '=0$. Thus, the
  linear span of $P(\psi _j,0)$ is dense in $G_R$.}\BibitemShut {Stop}%
\bibitem [{\citenamefont {Reed}\ and\ \citenamefont {Simon}(1978)}]{RS4}%
  \BibitemOpen
  \bibfield  {author} {\bibinfo {author} {\bibfnamefont {M.}~\bibnamefont
  {Reed}}\ and\ \bibinfo {author} {\bibfnamefont {B.}~\bibnamefont {Simon}},\
  }\href@noop {} {\emph {\bibinfo {title} {Methods of modern mathematical
  physics. {IV}. {A}nalysis of operators}}}\ (\bibinfo  {publisher} {Academic
  Press [Harcourt Brace Jovanovich Publishers]},\ \bibinfo {address} {New
  York},\ \bibinfo {year} {1978})\BibitemShut {NoStop}%
\bibitem [{\citenamefont {Gesztesy}\ \emph {et~al.}(2012)\citenamefont
  {Gesztesy}, \citenamefont {Gomilko}, \citenamefont {Sukochev},\ and\
  \citenamefont {Tomilov}}]{Gesztesy2012}%
  \BibitemOpen
  \bibfield  {author} {\bibinfo {author} {\bibfnamefont {F.}~\bibnamefont
  {Gesztesy}}, \bibinfo {author} {\bibfnamefont {A.}~\bibnamefont {Gomilko}},
  \bibinfo {author} {\bibfnamefont {F.}~\bibnamefont {Sukochev}}, \ and\
  \bibinfo {author} {\bibfnamefont {Y.}~\bibnamefont {Tomilov}},\ }\bibfield
  {title} {\enquote {\bibinfo {title} {On a question of {A}. {E}. {N}ussbaum on
  measurability of families of closed linear operators in a {H}ilbert space},}\
  }\href@noop {} {\bibfield  {journal} {\bibinfo  {journal} {Israel J. Math.}\
  }\textbf {\bibinfo {volume} {188}},\ \bibinfo {pages} {195--219} (\bibinfo
  {year} {2012})}\BibitemShut {NoStop}%
\bibitem [{\citenamefont {Lance}(1975)}]{Lance1975}%
  \BibitemOpen
  \bibfield  {author} {\bibinfo {author} {\bibfnamefont {C.}~\bibnamefont
  {Lance}},\ }\bibfield  {title} {\enquote {\bibinfo {title} {Direct integrals
  of left {H}ilbert algebras},}\ }\href@noop {} {\bibfield  {journal} {\bibinfo
   {journal} {Math. Ann.}\ }\textbf {\bibinfo {volume} {216}},\ \bibinfo
  {pages} {11--28} (\bibinfo {year} {1975})}\BibitemShut {NoStop}%
\bibitem [{\citenamefont {Dixon}(1971)}]{Dixon1971}%
  \BibitemOpen
  \bibfield  {author} {\bibinfo {author} {\bibfnamefont {P.~G.}\ \bibnamefont
  {Dixon}},\ }\bibfield  {title} {\enquote {\bibinfo {title} {Unbounded
  operator algebras},}\ }\href@noop {} {\bibfield  {journal} {\bibinfo
  {journal} {Proc. London Math. Soc. (3)}\ }\textbf {\bibinfo {volume} {23}},\
  \bibinfo {pages} {53--69} (\bibinfo {year} {1971})}\BibitemShut {NoStop}%
\bibitem [{\citenamefont {Erd{\'e}lyi}\ \emph {et~al.}(1953)\citenamefont
  {Erd{\'e}lyi}, \citenamefont {Magnus}, \citenamefont {Oberhettinger},\ and\
  \citenamefont {Tricomi}}]{Bateman}%
  \BibitemOpen
  \bibfield  {author} {\bibinfo {author} {\bibfnamefont {A.}~\bibnamefont
  {Erd{\'e}lyi}}, \bibinfo {author} {\bibfnamefont {W.}~\bibnamefont {Magnus}},
  \bibinfo {author} {\bibfnamefont {F.}~\bibnamefont {Oberhettinger}}, \ and\
  \bibinfo {author} {\bibfnamefont {F.~G.}\ \bibnamefont {Tricomi}},\
  }\href@noop {} {\emph {\bibinfo {title} {Higher transcendental functions.
  {V}ols. {I}, {II}}}}\ (\bibinfo  {publisher} {McGraw-Hill Book Company, Inc.,
  New York-Toronto-London},\ \bibinfo {year} {1953})\ \bibinfo {note} {based,
  in part, on notes left by Harry Bateman}\BibitemShut {NoStop}%
\end{thebibliography}%

\end{document}